\documentclass{llncs}
\usepackage{todonotes}
\usepackage{amsmath}
\usepackage{amssymb}

\usepackage{amsthm}
\usepackage{spverbatim}

\renewcommand{\vec}[1]{\mathbf{#1}}

\newcommand{\set}[1]{\{{#1}\}}

\newcommand{\ubf}[1]{\mathrm{UB}_{#1}}
\newcommand{\ubdsf}[1]{\mathrm{UB}_{\mathrm{DS}_{#1}}}
\newcommand{\dersens}[1]{\mathrm{DS}_{#1}}

\newcommand{\gencauchy}[1]{\mathit{GenCauchy}({#1})}
\newcommand{\dpdist}{d_\mathrm{dp}}

\newcommand{\laplace}[1]{\mathit{Lap}({#1})}
\DeclareMathOperator*{\argmin}{argmin}
\DeclareMathOperator*{\argmax}{argmax}

\newcommand{\norm}[1]{\left\|{#1}\right\|}
\newcommand{\abs}[1]{\left|{#1}\right|}

\newcommand{\compbanach}{composite norm}

\newcommand{\injto}{\rightarrowtail}
\newcommand{\distr}[1]{\mathcal{D}({#1})}

\newcommand{\RR}{\mathbb{R}}
\newcommand{\RRpi}{{\RR^+}}
\newcommand{\myexp}[1]{e^{#1}}
\newcommand{\NN}{\mathbb{N}}

\newcommand{\REFBANACHSmooth}{\ref{sec:banach-smoothing}}

\newcommand{\prpre}[1]{\mathbf{Pr}_{pre}[{{#1}}]}
\newcommand{\prpost}[1]{\mathbf{Pr}_{post}[{{#1}}]}
\newcommand{\noised}[1]{\mathcal{M}_{{#1}}}
\newcommand{\pr}[1]{\mathbf{Pr}[{{#1}}]}
\newcommand{\rndvar}[1]{\mathbf{{#1}}}

\newcommand{\fxpre}{f_X}

\newtheorem{fact}[proposition]{Fact}

\title{Achieving Differential Privacy using Methods from Calculus\thanks{This research was funded by the Air Force Research laboratory (AFRL) and Defense
Advanced Research Projects Agency (DARPA) under contract FA8750-16-C-0011. The
views expressed are those of the author(s) and do not reflect the official policy or position of
the Department of Defense or the U.S. Government. This research was also funded by Estonian Research Council, grant no. IUT27-1.}}

\author{Peeter Laud\inst{1} \and Alisa Pankova\inst{1} \and Martin Pettai\inst{1}}
\institute{Cybernetica, Estonia\\
\email{\{peeter.laud,alisa.pankova,martin.pettai\}@cyber.ee}
}

\begin{document}
\maketitle
\begin{abstract}
We introduce \emph{derivative sensitivity}, an analogue to local sensitivity for continuous functions. We use this notion in an analysis that determines the amount of noise to be added to the result of a database query in order to obtain a certain level of differential privacy, and demonstrate that derivative sensitivity allows us to employ powerful mechanisms from calculus to perform the analysis for a variety of queries. We have implemented the analyzer and evaluated its efficiency and precision.

We also show the flexibility of derivative sensitivity in specifying the quantitative privacy notion of the database, as desired by the data owner. Instead of only using the ``number of changed rows'' metric, our metrics can depend on the locations and amounts of changes in a much more nuanced manner. This will help to make sure that the distance is not larger than the data owner desires (which would undermine privacy), thereby encouraging the adoption of differentially private data analysis mechanisms.
\end{abstract}

\section{Introduction}
%Motivation, why our work is interesting.
%\todo[inline]{Peeter writes.}

Differential privacy~\cite{DBLP:conf/icalp/Dwork06} (DP) is one of the most prominent ways to quantitatively define privacy losses from releasing derived information about data collections, and to rigorously argue about the accumulation of these losses in information processing systems. Differentially private information release mechanisms invariably employ the addition of noise somewhere during the processing, hence reducing the utility of the result. For specific information processing tasks, or families of tasks, there exist carefully designed methods to achieve DP with only a little loss in utility~\cite{DBLP:conf/sigmod/HayMMCZ16}; for some methods, this loss asymptotically approaches zero as the size of the processed dataset grows~\cite{DBLP:conf/stoc/NissimRS07}. But a general method for making an information release mechanism differentially private is to add noise of appropriate magnitude to the output of this mechanism. This magnitude depends on the \emph{sensitivity} of the mechanism --- the amount of change of its output when its input is changed by a unit amount. This paper is devoted to the study of computing (or safely approximating) the sensitivity of information release mechanisms given by their source code.

The definition of the ratio of changes of the outputs and inputs of the information release mechanism requires metrics on both of them. For the outputs, we take the common approach of requiring them to be numeric. In fact, we let the output of the mechanism to be a single real number, and the distance between them to be their difference. Our results can be generalized to mechanisms that output histograms, i.e. tuples of numbers.

The definition of the metric on inputs is a richer question, and really reflects how the data owner wants to quantify its privacy. DP can be defined with respect to any metric on the set of possible inputs~\cite{DBLP:conf/pet/ChatzikokolakisABP13}. For database tables and databases, a common metric has been the number of different rows~\cite{DBLP:conf/sigmod/McSherry09}. But the data owner may consider different changes in one row of some database table as different, particularly if it involves some numeric attributes. A small change in such attribute in a single row may be considered as ``small'', and the DP mechanism should strive to hide such change. A large change in the same attribute in a single row may be allowed to be more visible in the output. A change in geographic coordinates may be defined differently from a change in the length or the quantity of something. Finally, when several attributes change in a single row, then there are many reasonable ways of combining their changes (including summing them up, or taking their maximum) in defining how much the row has changed. A variety of definitions for the combination of changes is also possible if several rows change in a table, or several tables change in a database. We note that for inputs to the information release mechanism, the overestimation of their distance leads to the underestimation of the amount of noise added to achieve a certain level of DP. Hence it is crucial to handle a wide range of metrics, from which the data owner may select the one that he considers to best reflect his privacy expectations.

In this paper, we show how different metrics on databases can be defined, and how the sensitivities of queries can be computed with respect of these input metrics. As queries we support a significant subset of SQL with projections, filters, joins and also certain types of subqueries, but no intermediate groupings. The metrics that we support for rows, tables, and databases are $\ell_p$-metrics with $1\leq p\leq \infty$, as well as their combinations and compositions, also supporting the declaration of certain rows or columns as public. Our approach generalizes and at the same time simplifies the notion of \emph{local sensitivity}~\cite{DBLP:conf/stoc/NissimRS07}, tying it with the fundamental notions of functional analysis and expanding the kinds of input data to which it may be applied.

While we consider the theoretical contributions (reported in Sec.~\ref{sec:banach-theory}) of this paper as the most significant, providing a novel point of view of sensitivity for information release mechanisms, we also discuss (Sec.~\ref{sec:sqlapplication}) a concrete implementation of this theory in this paper, applying it to SQL workflows. Our analyzer takes as inputs the description of the database, the description of the metric on it, the query to be performed, as well as the actual contents of the database (as required by local sensitivity based approaches), and returns the value of smoothened \emph{derivative sensitivity} (our name for the analogue of local sensitivity, defined in this paper) of this query at this database. The return value can be used to scale the added noise in order to obtain a certain level of DP. The analyzer is integrated with the database management system in order to perform the computations specific to the contents of the database, as described in Sec.~\ref{sec:eval}.

%Differential privacy.

%Various algorithms for differential privacy. But general methods add noise to functions with bounded sensitivity.

%What we achieve: take a database, have a relative freedom in how the distances are defined, take a query with numeric answer, get the noise level. Generalizes to histogram outputs. The generalization to Banach spaces gives us unprecedented flexibility in systematically defining privacy properties.

%A half-theoretical paper.

\section{Related Work}
%Some interesting works in differential privacy.
%\todo[inline]{Peeter writes.}

Differential privacy was introduced by Dwork~\cite{DBLP:conf/icalp/Dwork06}. PINQ~\cite{DBLP:conf/sigmod/McSherry09} is a worked-out example of using it for providing privacy-preserving replies to database queries. In an implementation, our analyzer of SQL queries would occupy the same place as the PINQ wrapper which analyses LINQ queries and maintains the privacy budget.

It has been recognized that any metric on the set of possible inputs will give us a definition of DP with respect to this metric~\cite{DBLP:conf/pet/ChatzikokolakisABP13,DBLP:conf/birthday/ElSalamounyCP14}. A particular application of this approach has been the privacy-preserving processing of location data~\cite{DBLP:conf/icdcit/ChatzikokolakisPS15}. The personalized differential privacy~\cite{DBLP:conf/popl/EbadiSS15} can also be seen as an instance of using an arbitrary metric, albeit with a more complex set of distances.

We use \emph{norms} to state the privacy requirements on input data. Normed vector spaces have appeared in the DP literature in the context of $K$-normed mechanisms~\cite{DBLP:conf/stoc/HardtT10,awan2018structure}, which extend the one-dimensional and generalize the many-dimensional Laplace mechanism. These mechanisms are rather different from our techniques and they do not explore the use of completeness of norms and differentiability of information release mechanism to find their sensitivity.

Nissim et al.~\cite{DBLP:conf/stoc/NissimRS07} introduce \emph{local sensitivity} and its smooth upper bounds, and use them to give differentially private approximations for certain statistical functions. The local sensitivity of a function is similar to its derivative. This has been noticed~\cite{10.1007/978-3-642-36594-2_26}, but we are not aware of this similarity being thoroughly exploited, except perhaps for devising DP machine learning methods~\cite{DBLP:conf/sigmod/0001LKCJN17}. In this paper, this similarity will play a central role.

A couple of different static approaches for determining sensitivities of SQL queries or their upper bounds have been proposed. Palamidessi and Stronati~\cite{DBLP:journals/corr/abs-1207-0872} apply abstract interpretation to an SQL query, following its abstract syntax tree, combining the sensitivities of relational algebra operations (projection and filter have sensitivity at most $1$, set operations have sensitivity at most $2$, etc.) similarly to~\cite{FeatherweightPINQ}. Additionally, they track the diameters of the domains of the values in the outcome of the query; the sensitivity cannot be larger than the diameter.

The computability of the precise sensitivity of the queries is studied by Arapinis et al.~\cite{DBLP:conf/icalp/ArapinisFG16}. They identify a subclass of queries (\emph{Conjunctive queries} with restricted \texttt{WHERE}-clauses) for which the sensitivity can be precisely determined. However, they also show that the problem is uncomputable in general. In addition, they show how functional dependencies and cardinality constraints may be used to upper-bound sensitivities of join-queries.

Cardinality constraints are also used by Johnson et al.~\cite{ElasticSensitivity} in their abstract interpretation based approach. For a database, they consider the maximum frequency of a value of an attribute in one of the tables, maximized over all databases at most at some distance to the original database, thereby building on the notion of (smooth) local sensitivity.

\section{Preliminaries}

\subsection{Sensitivity and Differential Privacy}\label{sec:sensdef}
Let $X$ be the set of possible databases. We assume that there is a metric $d_X(x,x')$ for $x,x' \in X$, quantifying the difference between two databases. For example, we could define $d_X(x,x') = n$ for two datatables that differ in exactly $n$ rows.

For a set $Y$, let $\distr{Y}$ denote the set of all probability distributions over $Y$. For $n\in\NN$, let $[n]$ denote the set $\{1,\ldots,n\}$. Let $[1,\infty]$ denote the set $\{x\in\RR\,|\,x\geq 1\}\cup\{\infty\}$.

Suppose that someone wants to make a query to the database. If the data is (partially) private, the query output may leak some sensitive information about the data. Noise can be added to the output to solve this problem. One possible privacy definition is that the output should be indistinguishable w.r.t. each individual entry, i.e. if we remove any row from the database, we should not see much change in the distribution of outputs.

\begin{definition}[differential privacy,~\cite{DBLP:conf/icalp/Dwork06}]\label{def:dpriv}
Let $X$ be a metric space and $f:X\rightarrow\distr{Y}$. The mapping $f$ is \emph{$\varepsilon$-differentially private} if for all (measurable) $Y'\subseteq Y$, and for all $x,x'$, where $d_X(x,x')=1$, the following inequality holds:
\begin{equation}\label{eq:eddpdef}
Pr[f(x)\in Y']\leq e^\varepsilon Pr[f(x')\in Y']\enspace.
\end{equation}
\end{definition}

The noise magnitude depends on the difference between the outputs of $f$. The more different the outputs are, the more noise we need to add to make them indistinguishable from each other. This is quantified by the global sensitivity of $f$.

\begin{definition}[global sensitivity]
For $f : X \to Y$, the \emph{global sensitivity} of $f$ is $GS_f = \max_{x,x' \in X} \frac{d_Y(f(x),f(x'))}{d_x(x,x')}$.
\end{definition}

Sensitivity is the main tool in arguing the differential privacy of various information release mechanisms. For mechanisms that add noise to the query output, this value serves as a parameter for the noise distribution. The noise will be proportional to $GS_f$. One suitable noise distribution is $\laplace{\frac{GS_f}{\varepsilon}}$, where $\laplace{\lambda}(z) \propto 2\lambda \cdot e^{-|z|/\lambda}$. The sampled noise will be sufficient to make the output of $f$ differentially private, regardless of the input of $f$.

%Sensitivity is the main tool in arguing the differential privacy of various information release mechanisms. For mechanisms that first compute a ``useful'' function and then add noise to it, the differential privacy of the resulting mechanism is the ratio of the sensitivity of that function and the magnitude of the added noise.

Differential privacy itself can also be seen as an instance of sensitivity. Indeed, define the following distance $\dpdist$ over $\distr{Y}$:
\[
\dpdist(\chi,\chi')=\inf\{\varepsilon\in\RRpi\,|\,\forall y\in Y: \myexp{-\varepsilon}\chi'(y)\leq\chi(y)\leq \myexp{\varepsilon}\chi'(y) \}\enspace.
\]
Then a mechanism $\mathcal{M}$ from $X$ to $Y$ is $\epsilon$-DP iff it is $\epsilon$-sensitive with respect to the distances $d_X$ on $X$ and $\dpdist$ on $\distr{Y}$.

\subsection{Local and Smooth Sensitivity}
Our work extends the results~\cite{DBLP:conf/stoc/NissimRS07}, which makes use of instance-based additive noise. Since noise is always added to the output of a query that is applied to a particular state of the database, and some state may require less noise than the other, the noise magnitude may depend on the data to which the function is applied.

\begin{definition}[local sensitivity]
For $f : X \to Y$ and $x \in X$ , the \emph{local sensitivity} of $f$ at $x$ is $LS_f(x) = \max_{x' \in X} \frac{d_Y(f(x),f(x'))}{d_x(x,x')}$.
\end{definition}

The use of local sensitivity may allow the use of less noise, particularly when the global sensitivity of $f$ is unbounded. However, $LS_f(x)$ may not be directly used to determine the magnitude of the noise, because this magnitude may itself leak something about $x$. To solve this problem, Nissim et al.~\cite{DBLP:conf/stoc/NissimRS07} use a \emph{smooth upper bound} on $LS_f(x)$. It turns out that such an upper bound is sufficient to achieve differential privacy for $f$; potentially with less noise than determined by $GS_f$.

\begin{definition}[smooth bound]\label{def:smooth}
For $\beta > 0$, a function $S : X \to \RR^{+}$ is a \emph{$\beta$-smooth upper bound} on the local sensitivity of $f$ if it satisfies the following requirements:
\begin{itemize}
\item $\forall x \in X : S(x) \geq LS_f(x)$ ;
\item $\forall x, x' \in X : S(x) \leq e^{\beta\cdot d_X(x,x')} S(x')$ .
\end{itemize}
\end{definition}

It has been shown in~\cite{DBLP:conf/stoc/NissimRS07} how to add noise based on the smooth bound on $LS_f$. The statement that we present in Theorem~\ref{thm:localsensnoise} is based on combination of Lemma 2.5 and Example 2 of~\cite{DBLP:conf/stoc/NissimRS07}.

\begin{definition}[generalized Cauchy distribution]\label{def:gencauchy}
For a parameter $\gamma\in\RR_+$, $\gamma>1$, the \emph{generalized Cauchy distribution} $\gencauchy{\gamma}\in\distr{\RR}$ is given by the proportionality
\[
\gencauchy{\gamma}(x) \propto \frac{1}{1+|x|^\gamma}\enspace.
\]
\end{definition}
\begin{theorem}[local sensitivity noise~\cite{DBLP:conf/stoc/NissimRS07}]\label{thm:localsensnoise}
Let $\eta$ be a fresh random variable sampled according to $\gencauchy{\gamma}$. Let $\alpha = \frac{\varepsilon}{4\gamma}$ and $\beta = \frac{\varepsilon}{\gamma}$. For a function $f:X \to \RR$, let $S : X \to \RR$ be a $\beta$-smooth upper bound on the local sensitivity of $f$. Then the information release mechanism $f(x) + \frac{S(x)}{\alpha}\cdot \eta$ is $\varepsilon$-differentially private.
\end{theorem}

\subsection{Norms and Banach spaces}
First, we recall some basics of Banach space theory. Throughout this paper, we denote vectors by $\vec{x}$, and norms by $\norm{\cdot}_N$, where $N$ in the subscript specifies the particular norm.

\begin{definition}[norm and seminorm]\label{def:norm}
A \emph{seminorm} is a function $\norm{\cdot}: V \to \RR$ from a vector space $V$, satisfying the following axioms for all $\vec{x},\vec{y}\in V$:
\begin{itemize}
\item $\norm{\vec{x}} \geq 0$;
\item $\norm{\alpha \vec{x}} = |\alpha|\cdot\norm{\vec{x}}$ (implying that $\norm{\vec{0}}=0$);
\item $\norm{\vec{x} + \vec{y}} \leq \norm{\vec{x}} + \norm{\vec{y}}$ (triangle inequality).
\end{itemize}
Additionally, if $\norm{\vec{x}}=0$ holds only if $\vec{x}=\vec{0}$, then $\norm{\cdot}$ is a \emph{norm}.
\end{definition}

\begin{definition}[$\ell_p$-norm]
Let $X_i \subseteq \RR$, $p \in[1,\infty]$. The $\ell_p$ norm of $\vec{x} \in X_1 \times \cdots \times X_n$, denoted $\norm{\vec{x}}_p$ is defined as
\[\norm{\vec{x}}_p = \left({\sum_{i=1}^{n}\abs{x_i}^p}\right)^{1/p}.\enspace\]
For $p = \infty$, $\ell_{\infty}$ is defined as
\[\norm{\vec{x}}_{\infty} = \lim_{p\to\infty} \left({\sum_{i=1}^{n}\abs{x_i}^p}\right)^{1/p} = \max_{i=1}^{n}\abs{x_i}.\enspace\]
\end{definition}

\begin{definition}[Banach space]\label{def:banach-space}
A \emph{Banach space} is a vector space with a norm that is \emph{complete} (i.e. each converging sequence has a limit).
\end{definition}

Banach spaces combine metric spaces with distances, which are necessary for defining differential privacy. The completeness property allows us to define derivatives. Using the norm of a Banach space, we may generalize the notion of continuous function from real numbers to Banach spaces.
\begin{definition}[Continuous function in Banach space]\label{def:cont-function-banach}
Let $V$ and $W$ be Banach spaces, and $U\subset V$ an open subset of $V$. A function $f : U \to W$ is called \emph{continuous} if
\[\forall \varepsilon > 0:\ \exists \delta > 0:\ \norm{x - x'}_V \leq \delta \implies \norm{f(x) - f(x')}_W \leq \varepsilon\enspace.\]
\end{definition}

The notion of the derivative of a function can be also extended to Banach spaces.
\begin{definition}[Fr\'echet derivative]\label{def:derivative-banach}
Let $V$ and $W$ be Banach spaces, and $U\subset V$ an open subset of $V$. A function $f : U \to W$ is called \emph{Fr\'echet differentiable at $x\in U$} if there exists a bounded linear operator $df_x : V \to W$ such that $\lim_{h \to 0}\frac{\norm{f(x + h) - f(x) - df_x(h)}_W}{\norm{h}_V} = 0$. Such operator $df_x$ is called Fr\'echet derivative of $f$ at the point $x$.
\end{definition}
%Throughout this paper, we do not use the general definition of Fr\'echet derivative directly. For the particular Banach spaces that we study in this paper, we will use partial derivatives which are close to derivatives over $\RR$.

The mean value theorem can be generalized to Banach spaces (to a certain extent).
\begin{theorem}[Mean value theorem (~\cite{baggett1991functional}, Chapter XII)]\label{thm:mvt}
Let $V$ and $W$ be Banach spaces, and $U\subset V$ an open subset of $V$. Let $f : U \to W$, and let $x,x' \in U$. Assume that $f$ is defined and is continuous at each point $(1 - t)x + tx'$ for $0 \leq t \leq 1$, and differentiable for $0 < t < 1$. Then there exists $t^{*} \in (0,1)$ such that
\[\norm{f(x) - f(x')}_W \leq \norm{df_z}_{V \to W} \norm{x - x'}_V \]
for $z = (1 - t^{*})x + t^{*} x'$, where $\norm{\cdot}_{V \to W}$ denotes the norm of operator that maps from $V$ to $W$.
\end{theorem}
%------------------------------

\subsection{Some properties of $\ell_p$-norms}

Let us give several known facts and prove some lemmas that help us in comparing different norms. The proofs of the lemmas are pushed into App.~\ref{app:proofs}, since they are quite straightforward and follow directly from known properties of norms, but the lemmas have too specific forms to be cited directly.

\begin{fact}\label{fact:lpeqv}
For all $(x_1,\ldots,x_m) \in \RR^m$, $p \geq q$, we have \[\norm{x_1,\ldots,x_m}_p \leq \norm{x_1,\ldots,x_m}_q \leq m^{1/q - 1/p}\cdot\norm{x_1,\ldots,x_m}_p\enspace.\]
\end{fact}

\begin{fact}\label{fact:lpsingle}
For all $x \in \RR^+$, we have $\norm{x}_p = x$ and $\norm{x}_{\infty} = x$.
\end{fact}

\begin{lemma}\label{lm:surj}
A norm $N: \RR^m \to \RR^{+}$ is surjective.
\end{lemma}
%\begin{proof}
%First of all, a norm cannot be a constant function $0$ due to the condition $\norm{\vec{x}} = 0 \iff \vec{x} = \vec{0}$. Let $\vec{x}$ be such that $\norm{\vec{x}}_N = y'$ for some $0 \neq y' \in \RR^{+}$. For all $y \in \RR^{+}$, we have $y = \frac{y}{y'}\cdot y' = \frac{y}{y'}\cdot \norm{\vec{x}}_N = \norm{\frac{y}{y'}\cdot\vec{x}}_N$, where $\frac{y}{y'}\cdot\vec{x} \in \RR^m$.
%\end{proof}

\begin{lemma}\label{lm:scalecopy} For all $x \in \RR$, $(\alpha_1,\ldots,\alpha_k) \in \RR^k$, $(y_1,\ldots,y_m) \in \RR^m$:
\[\norm{\alpha_1 \cdot x,\ldots,\alpha_k \cdot x,y_1,\ldots,y_m}_p = \norm{\sqrt[p]{\sum_{i=1}^{k} \alpha^p_i} \cdot x,y_1,\ldots,y_m}_p\enspace.\]
\end{lemma}
%\begin{proof} Since an $\ell_p$-norm is defined for $p \geq 1$, we may raise both sides of equation to the power $p$. We use the definition of $\ell_p$-norm and rewrite the term.
%\begin{eqnarray*}
%\norm{\alpha_1 x,\ldots,\alpha_k x,y_1,\ldots,y_m}^p_p & = & \sum_{i=1}^{k} (\alpha_i x)^p + \sum_{i=1}^{m} y^p_i \\
%& = & \left(\sum_{i=1}^{k} \alpha^p_i\right) \cdot x^p + \sum_{i=1}^{m} y^p_i \\
%& = & \norm{\sqrt[p]{\sum_{i=1}^{k} \alpha^p_i} \cdot x,y_1,\ldots,y_m}^p_p\enspace.
%\end{eqnarray*}\qedhere
%\end{proof}
\noindent Putting $\alpha_i = 1$ for all $i \in [n]$, we get the following corollary.
\begin{corollary}\label{cor:scalecopy}
$\norm{\overbrace{x,\ldots,x}^{k},y_1,\ldots,y_m}_p = \norm{\sqrt[p]{k} \cdot x,y_1,\ldots,y_m}_p$.
\end{corollary}

\begin{lemma}\label{lm:ungroup}
Let $\vec{x} = (x_1,\ldots,x_k) \in \RR^k$, $\vec{y} = (y_1,\ldots,y_n) \in \RR^n$, $\vec{z} = (z_1,\ldots,z_m) \in \RR^m$. If $p \geq q \geq 1$, then
\begin{enumerate}
\item $\norm{\norm{\vec{x}}_q,\norm{\vec{y}}_q, z_1,\ldots,z_m}_p \leq \norm{\norm{\vec{x} | \vec{y}}_q, z_1,\ldots,z_m}_p$;
\item $\norm{\norm{\vec{x}}_p,\norm{\vec{y}}_p, z_1,\ldots,z_m}_q \geq \norm{\norm{\vec{x} | \vec{y}}_p, z_1,\ldots,z_m}_q$;
\end{enumerate}
where $\vec{x} | \vec{y}$ denotes concatenation. If $p = q$, then the inequalities become equalities.
\end{lemma}
%\begin{proof}
%Since $p,q \geq 1$, we may raise both sides of equations to the powers $p$ or $q$. The main inequalities that we use in the proof are $a^n + b^n \leq (a + b)^n$ for $n \geq 1$, and $a^n + b^n \geq (a + b)^n$ for $n \leq 1$.

%\begin{eqnarray*}
%\norm{\norm{\vec{x}}_q,\norm{\vec{y}}_q, z_1,\ldots,z_m}_p^p & = & \left(\sum_{i=1}^{k} x_i^q\right)^{\frac{p}{q}} + \left(\sum_{i=1}^{n} y_i^q\right)^{\frac{p}{q}} + \sum_{i=1}^{m} z_i^p \\
%& \leq & \left(\sum_{i=1}^{k} x_i^q + \sum_{i=1}^{n} y_i^q\right)^{\frac{p}{q}} + \sum_{i=1}^{m} z_i^p \\
%& = & \norm{\vec{x} | \vec{y}}_q^{p} + \sum_{i=1}^{m} z_i^p \\
%& = & \norm{\norm{\vec{x} | \vec{y}}_q, z_1,\ldots,z_m}_p^p\enspace.
%\end{eqnarray*}
%\begin{eqnarray*}
%\norm{\norm{\vec{x}}_p,\norm{\vec{y}}_p, z_1,\ldots,z_m}_q^q & = & \left(\sum_{i=1}^{k} x_i^p\right)^{\frac{q}{p}} + \left(\sum_{i=1}^{n} y_i^p\right)^{\frac{q}{p}} + \sum_{i=1}^{m} z_i^q \\
%& \geq & \left(\sum_{i=1}^{k} x_i^p + \sum_{i=1}^{n} y_i^p\right)^{\frac{q}{p}} + \sum_{i=1}^{m} z_i^q \\
%& = & \norm{\vec{x} | \vec{y}}_p^{q} + \sum_{i=1}^{m} z_i^q \\
%& = & \norm{\norm{\vec{x} | \vec{y}}_p, z_1,\ldots,z_m}_q^q\enspace.
%\end{eqnarray*}
%If $p = q$, then the inequalities in these equation arrays are equalities.
%\end{proof}

\begin{lemma}\label{lm:injmap}
Let $\vec{x} = (x_1,\ldots,x_m) \in \RR^m$, $\vec{y} = (y_1,\ldots,y_n) \in \RR^n$, $m \leq n$. If $p \geq q \geq 1$, and there exists an injective mapping $f: [m] \injto [n]$ such that $|x_i| \leq |y_{f(i)}|$, then $\norm{x_1,\ldots,x_m}_p \leq \norm{y_1,\ldots,y_n}_q$.
\end{lemma}
%\begin{proof}
%By Fact~\ref{fact:lpeqv}, if $p \geq q$, then $\norm{\vec{x}}_p \leq \norm{\vec{x}}_q$, so it suffices to prove $\norm{x_1,\ldots,x_m}_q \leq \norm{y_1,\ldots,y_n}_q$. Since $q \geq 1$, we may instead prove $\norm{x_1,\ldots,x_m}_q^q \leq \norm{y_1,\ldots,y_n}_q^q$. Assuming that there exists an injective mapping $f$ such that $|x_i| \leq |y_{f(i)}|$, we have
%\begin{eqnarray*}
%\norm{x_1,\ldots,x_m}_q^q & = & \sum_{i \in [m]} x_i^q \leq \sum_{i \in [m]} y_{f(i)}^q \\
% & \leq & \sum_{i \in [m]} y_{f(i)}^q + \sum_{j \in [n]\setminus Im(f)} y_j^q \\
% & = & \norm{y_1,\ldots,y_n}_q^q\enspace.
%\end{eqnarray*}\qedhere
%\end{proof}

\section{Derivative Sensitivity}\label{sec:banach-theory}
In this section,  we introduce our new results. We study more precisely the local sensitivity of \emph{continuous} functions. It turns out that, instead of estimating the quantity $\frac{d(f(x),f(x'))}{d(x,x')}$ directly, we may use \emph{derivative} of $f$ to compute the noise for differential privacy.
\subsection{Derivative sensitivity and its properties}\label{sec:ext-dersens-banach}%{{{1

%Previously (in deliverable D1.4~\cite{NAPLES-D14}), we had the following definition of derivative sensitivity:
%\begin{definition}
%Let $X=(\RR^n,d)$ be a metric space, with $d$ being the $\ell_1$-distance.
%Let $f:X\rightarrow\RR$. The \emph{derivative sensitivity} of $f$ is the following mapping from $X$ to $\RR_+$, where $\RR_+$ denotes the set of all non-negative real numbers:
%\[
%\dersens{f}(\vec{x})=\max_i\left|\frac{\partial f}{\partial x_i}(\vec{x})\right|\enspace.
%\]
%\end{definition}
%Here $x_i$ denotes the $i$-th component of the vector of variables $\vec{x}$.

%We extend it to the case where $X$ is any Banach space:

We propose constructions that allow us to achieve differential privacy for functions defined over Banach spaces.

\begin{definition}\label{def:dersens-banach}
Let $X$ be (an open convex subset of) a Banach space.
Let $f:X\rightarrow\RR$.
Let $f$ be Fr\'echet differentiable at each point of
%(the interior of)
$X$.
%(More generally, $X$ may be a convex subset of a Banach space. Then $f$ must be Fr\'echet differentiable at each point
%in the interior of $X$.)
The \emph{derivative sensitivity} of $f$ is the following mapping from $X$ to $\RR_+$, where $\RR_+$ denotes the set of all non-negative real numbers:
\[
\dersens{f}(\vec{x})=\|df_{\vec{x}}\|\enspace.
\]
where $df_{\vec{x}}$ is the Fr\'echet derivative of $f$ at $\vec{x}$ and $\|df_{\vec{x}}\|$ is the operator norm of $df_{\vec{x}}$.
\end{definition}

Similarly to the local sensitivity of~\cite{DBLP:conf/stoc/NissimRS07}, we will need to find smooth upper bounds on derivative sensitivity to compute the noise. We extend the definition of smoothness (Def.~\ref{def:smooth}) to the case where $X$ is any Banach space:
\begin{definition}\label{def:smoothness}
Let $p:X\rightarrow\RR$ and $\beta\in\RR$. The mapping $p$ is \emph{$\beta$-smooth}, if $p(\vec{x})\leq e^{\beta\cdot \|\vec{x}'-\vec{x}\|}\cdot p(\vec{x}')$ for all $\vec{x},\vec{x}'\in X$.
\end{definition}

The next theorem shows how to compute noise for differential privacy. As we see in the proof, crucial in selecting the amount of noise to be added to $f(\vec{x})$ is the knowledge of a $\beta$-smooth upper bound on the derivative sensitivity of $f$. We let $c$ denote such a bound. We consider the same noise distributions as in~\cite{DBLP:conf/stoc/NissimRS07}. %For a parameter $\gamma\in\RR_+$, $\gamma>1$, the \emph{generalized Cauchy distribution} $\gencauchy{\gamma}\in\distr{\RR}$ is given by the proportionality
%\[
%\gencauchy{\gamma}(x) \propto \frac{1}{1+|x|^\gamma}\enspace,
%\]
%where ``usual'' Cauchy distribution is obtained for $\gamma=2$. 

Noise distributed by generalized Cauchy distribution (Def.~\ref{def:gencauchy}), weighed by a smooth upper bound on the derivative sensitivity of $f$, allows us to achieve $\epsilon$-DP.

\begin{theorem}\label{thm:dersenscauchy-banach}
Let $\gamma,b,\beta\in\RR_+$, $\gamma>1$. Let $\epsilon=(\gamma+1)(b+\beta)$. Let $\eta$ be a random variable distributed according to $\gencauchy{\gamma}$. Let $c$ be a $\beta$-smooth upper bound on $\dersens{f}$ for a function $f:X\rightarrow\RR$. Then $g(\vec{x})=f(\vec{x})+\frac{c(\vec{x})}{b}\cdot\eta$ is $\epsilon$-differentially private.
\end{theorem}

\begin{proof}
Let $\dpdist$ be defined as in Sec.~\ref{sec:sensdef}. Let $\eta\sim\gencauchy{\gamma}$. The generalized Cauchy distribution is relatively stable under shifts and stretchings, satisfying the following inequalities for all\\ $a_1,a_2,c_1,c_2\in\RR$~\cite{DBLP:conf/stoc/NissimRS07}:
\begin{align*}
\dpdist(a_1+c_1\cdot\eta, a_2+c_1\cdot\eta)&\leq(\gamma+1)\cdot\left|\frac{a_2-a_1}{c_1}\right| \\
\dpdist(c_1\cdot\eta, c_2\cdot\eta)&\leq (\gamma+1)\cdot\left|\ln\frac{c_2}{c_1}\right|\enspace.
\end{align*}
The combination of these two inequalities gives
\begin{equation}
\dpdist(a_1+c_1\cdot\eta, a_2+c_2\cdot\eta)\leq(\gamma+1)\cdot\left(\frac{|a_2-a_1|}{\max\{|c_1|,|c_2|\}}+\left|\ln\frac{c_2}{c_1}\right|\right)\label{eq:gcshiftstretch}\enspace.
\end{equation}

Let $\vec{x},\vec{x}'\in X$.
%Suppose that they differ only in the $i_0$-th coordinate. W.l.o.g. assume that $x'_{i_0}\geq x_{i_0}$. Denote $L=x'_{i_0}-x_{i_0}$.
Denote $L=\|\vec{x}' - \vec{x}\|$.
We have to show that $\dpdist(g(\vec{x}'),g(\vec{x}))\leq \epsilon L =(\gamma+1)(b+\beta)L$.

We can substitute the definition of $g$ into the left side of the desired inequality above, and using the inequality (\ref{eq:gcshiftstretch}) and the definition of smoothness, obtain
\begin{eqnarray*}
\dpdist(g(\vec{x}),g(\vec{x}')) & = & \dpdist(f(\vec{x})+\frac{c(\vec{x})}{b}\cdot\eta, f(\vec{x}')+\frac{c(\vec{x}')}{b}\cdot\eta) \\
& \leq & (\gamma+1)\cdot \left( b\cdot\frac{|f(\vec{x}')-f(\vec{x})|}{|c(\vec{x})|} + \left|\ln\frac{c(\vec{x}')}{c(\vec{x})} \right|\right) \\
& \leq & (\gamma+1)\cdot \left( b\cdot\frac{|f(\vec{x}')-f(\vec{x})|}{|c(\vec{x})|} + \beta L\right)
\end{eqnarray*}

Unfortunately, we cannot directly bound ${|f(\vec{x}')-f(\vec{x})|}/{|c(\vec{x})|}$ with $L$.
Instead, we can only claim,
using the mean value theorem for Banach spaces and Fr\'echet derivative, that there exists some $\vec{v}$ in the segment
connecting $\vec{x}$ and $\vec{x}'$ satisfying
\[
|f(\vec{x}')-f(\vec{x})| = |df_{\vec{v}}(\vec{x}'-\vec{x})| \le \|df_{\vec{v}}\| \cdot \|\vec{x}'-\vec{x}\| \le c(\vec{v})\cdot L
\]
%Instead, if we let $\vec{y}[v]$ denote the tuple $\vec{x}$, where the $i_0$-th component is replaced with $v$, then we can only claim (using the mean value theorem), that there exists some $v_0$ in the segment $(x_{i_0},x'_{i_0})$ satisfying
%\[
%|f(\vec{x}')-f(\vec{x})|=\left|\frac{\partial f}{\partial x_{i_0}}(\vec{y}[v_0])\right|\cdot(x'_{i_0}-x_{i_0})\leq |c(\vec{y}[v_0])|\cdot L,
%\]
where the last inequality is due to $c$ being an upper bound on the derivative sensitivity of $f$.
However, by using this claim many times, we obtain the necessary inequality as follows.
Let $n\in\NN$ be arbitrary.
Let $\vec{v}_0 = \vec{x}$, $\vec{v}_n = \vec{x}'$,
and $\vec{v}_i = \frac{n-i}{n}\cdot \vec{x} + \frac{i}{n}\cdot \vec{x}'$,
i.e. the values $\vec{v}_0,\ldots,\vec{v}_n$ are evenly distributed in the segment connecting $\vec{x}$ to $\vec{x}'$.
These values are in $X$ because $X$ is convex.
%Let $v_0 = x_{i_0}$, $v_n = x'_{i_0}$ and $v_i = \bigl((n-i)\cdot x_{i_0} + i\cdot x'_{i_0}\bigr)/n$,
%i.e. the values $v_0,\ldots,v_n$ are evenly distributed from $x_{i_0}$ to $x'_{i_0}$.
Again, the mean value theorem implies that there exist $\vec{t}_1,\ldots,\vec{t}_n$
with $\vec{t}_i$ in the segment connecting $\vec{v}_{i-1}$ to $\vec{v}_i$, satisfying
%Again, the mean value theorem implies that there exist $t_1,\ldots,t_n$ with $v_{i-1}\leq t_i\leq v_i$, satisfying

\begin{eqnarray*}
|f(\vec{v}_i)-f(\vec{v}_{i-1})| & = & |df_{\vec{v}_i}(\vec{v}_i-\vec{v}_{i-1})|\leq
\|df_{\vec{v}_i}\| \cdot \|\vec{v}_i-\vec{v}_{i-1}\| \\
& \le & c(\vec{t}_i)\cdot \frac{L}{n} \leq e^{\beta L/n} \cdot c(\vec{v}_{i-1})\cdot \frac{L}{n}
%|f(\vec{y}[v_i])-f(\vec{y}[v_{i-1}])|=\left|\frac{\partial f}{\partial x_{i_0}}(\vec{y}[t_i])\right|\cdot(v_i-v_{i-1})\leq |c(\vec{y}[t_i])|\cdot \frac{L}{n}\leq e^{\beta L/n} \cdot |c(\vec{y}[v_{i-1}])|\cdot \frac{L}{n}
\end{eqnarray*}

for all $i\in\{1,\ldots,n\}$. Here the last inequality follows from the $\beta$-smoothness of $c$. We can use these claims together with the triangle inequality and obtain
\begin{multline*}
\dpdist(g(\vec{x}),g(\vec{x}')) \leq \sum_{i=1}^n \dpdist(g(\vec{v}_{i-1}),g(\vec{v}_i)) = \\
\sum_{i=1}^n\dpdist(f(\vec{v}_{i-1})+\frac{c(\vec{v}_{i-1})}{b}\cdot\eta, f(\vec{v}_i)+\frac{c(\vec{v}_i)}{b}\cdot\eta) \leq \\
(\gamma+1)\sum_{i=1}^n \left( b\cdot \frac{|f(\vec{v}_i)-f(\vec{v}_{i-1})|}{|c(\vec{v}_{i-1})|} + \frac{\beta L}{n}\right) \leq \\ 
(\gamma+1)\sum_{i=1}^n \left( b \cdot e^{\beta L/n} \cdot \frac{L}{n} + \frac{\beta L}{n} \right) = 
(\gamma+1)(be^{\beta L/n}+\beta)L\enspace.
%\dpdist(g(\vec{x}),g(\vec{x}')) \leq \sum_{i=1}^n \dpdist(g(\vec{y}[v_{i-1}]),g(\vec{y}[v_i])) = \\
%\sum_{i=1}^n\dpdist(f(\vec{y}[v_{i-1}])+\frac{c(\vec{y}[v_{i-1}])}{b}\cdot\eta, f(\vec{y}[v_i])+\frac{c(\vec{y}[v_i])}{b}\cdot\eta) \leq \\
%(\gamma+1)\sum_{i=1}^n \left( b\cdot \frac{|f(\vec{y}[v_i])-f(\vec{y}[v_{i-1}])|}{|c(\vec{y}[v_{i-1}])|} + \frac{\beta L}{n}\right) \leq 
%(\gamma+1)\sum_{i=1}^n \left( b \cdot e^{\beta L/n} \cdot \frac{L}{n} + \frac{\beta L}{n} \right) = \\
%(\gamma+1)(be^{\beta L/n}+\beta)L\enspace.
\end{multline*}
This inequality holds for any $n\in\NN$. If $n\rightarrow\infty$ then $e^{\beta L/n}\rightarrow 1$ and we obtain the inequality that we had to show.
%
%If $\vec{x}$ and $\vec{x}'$ differ in more than one coordinate, then we can transform $\vec{x}$ to $\vec{x}'$ by changing one coordinate at a time, and using the triangle inequality.
\end{proof}

\subsection{Computing derivative sensitivity}\label{sec:compute-dersens}
We show how to compute the derivative sensitivity for mappings from certain Banach spaces. The following lemmas are proven in App.~\ref{app:proofs}.
\begin{lemma}\label{lemma:banach-base}
\label{lemma:dersens-lpnorm}
Let $f:\RR^n\to \RR$, and let $\RR^n$ be equipped with the norm $\ell_p$.
Then $\|df_{\vec{x}}\|$ is the $\ell_q$-norm of the gradient vector $\nabla f(\vec{x})$, where $q=\frac{p}{p-1}$ (if $p=1$ then $q=\infty$ and \emph{vice versa}).
%\[
%q =
%\begin{cases}
%    \infty        & \mbox{if $p = 1$} \\
%    1             & \mbox{if $p = \infty$} \\
%    \frac{p}{p-1} & \mbox{if $1 < p < \infty$}
%\end{cases}
%\]
\end{lemma}

The $\ell_q$-norm happens to be the \emph{dual norm} of the $\ell_p$-norm.

\begin{lemma}
\label{lemma:banach-combine}
(a) Let $(V_1,\norm{\cdot}_{V_1})$ and $(V_2,\norm{\cdot}_{V_2})$ be Banach spaces.
Let $V = V_1 \times V_2$.
Let for all $(v_1,v_2)\in V$, \[ \|(v_1,v_2)\|_V = \|(\|v_1\|_{V_1},\|v_2\|_{V_2})\|_p \]
Then $(V,\norm{\cdot}_V)$ is a Banach space.
\\
(b) Suppose furthermore that a function $f : V \rightarrow \RR$ is differentiable at each point of $V$.
Fix a point $v = (v_1,v_2) \in V$.
Let $g : V_1 \rightarrow \RR$ be such that $g(x_1) = f(x_1,v_2)$
and $h : V_2 \rightarrow \RR$ be such that $h(x_2) = f(v_1,x_2)$.
Let $c_1 = \dersens{g}(v_1)$ and $c_2 = \dersens{h}(v_2)$.
Then $\dersens{f}(v) = \|(c_1,c_2)\|_q$ where $\norm{\cdot}_q$ is the dual norm of $\norm{\cdot}_p$.
\end{lemma}

\subsection{Composite norms}\label{sec:compositenorms}

Lemma~\ref{lemma:banach-base} and Lemma~\ref{lemma:banach-combine} together give us a construction of a certain class of norms and Banach spaces over $\RR^n$, summarized in Def.~\ref{def:compbanach}.

\begin{definition}[composite seminorm; variables used by a seminorm]\label{def:compbanach}
Let $\norm{\cdot}_N$ be a seminorm of the vector space $\RR^n$. It is a \emph{composite seminorm} if one of the following holds for all $\vec{x}=(x_1,\ldots,x_n)\in\RR^n$:
\begin{itemize}
\item There exists $i\in[n]$, such that $\norm{\vec{x}}_N=|x_i|$. Such seminorm \emph{uses the variable} $x_i$.
\item There exists a composite seminorm $\norm{\cdot}_M$ and $a\in\RR^+$, such that $\norm{\vec{x}}_N=a\cdot\norm{\vec{x}}_M$. The seminorm $\norm{\cdot}_N$ \emph{uses the same variables} as $\norm{\cdot}_M$.
\item There exist composite seminorms $\norm{\cdot}_{M_1},\ldots,\norm{\cdot}_{M_k}$ and $p\in[1,\infty]$, such that $\norm{\vec{x}}_N=\norm{\norm{\vec{x}}_{M_1},\ldots,\norm{\vec{x}}_{M_k}}_p$. The seminorm $\norm{\cdot}_N$ uses the union of the variables used by all $\norm{\cdot}_{M_i}$.
\end{itemize}
\end{definition}

An example of computing derivative sensitivity for a particular function from a Banach space with {\compbanach} can be found in App.~\ref{app:theoryexample1}. As next, we state some simple lemmas that help in comparing different composite norms. Their proofs can again be found in App.~\ref{app:proofs}.

A composite seminorm in $\RR^n$ can be seen as the semantics of a formal expression over the variables $x_1,\ldots,x_n$, where the term constructors are $\norm{\ldots}_p$ of any arity. We write $N\preceq M$ for two seminorms $N$ and $M$, if $\norm{\vec{x}}_N\leq\norm{\vec{x}}_M$ for all $\vec{x}\in\RR^n$.

\begin{lemma}\label{lm:scaleineq}
Let $N$ and $M$ be two {\compbanach}s over variables $\vec{x} = (x_1,\ldots,x_n)$, such that $N \preceq M$. Let composite seminorms $N',M',V_1,\ldots,V_m$ be such, that $N = N'(V_1,\ldots,V_m)$ and $M = M'(V_1,\ldots,V_m)$, and $V_i$ are mutually disjoint w.r.t. variables $x_j$ that they use. Then $N'(\alpha_1 V_1,\ldots,\alpha_m V_m) \preceq M'(\alpha_1 V_1,\ldots,\alpha_m V_m)$ for any choice of $\alpha_1,\ldots,\alpha_m \in \RR^{+}$.
\end{lemma}
%\begin{proof}
%$N \preceq M$ implies $\norm{x_1,\ldots,x_n}_{N} \leq \norm{x_1,\ldots,x_n}_{M}$ for any valuation $x_1,\ldots,x_n \in \RR^n$. Instead of $N \preceq M$, we may write $N'(V_1,\ldots,V_m) \preceq M'(V_1,\ldots,V_m)$. By Lemma~\ref{lm:surj}, each subnorm $V_i$ is surjective, and since all these subnorms use distinct sets of variables, we have $\norm{v_1,\ldots,v_m}_{N'} \leq \norm{v_1,\ldots,v_m}_{M'}$ for all valuations $(v_1,\ldots,v_m) \in (\RR^+)^m$.

%Since $\alpha_i \in \RR^{+}, v_i \in \RR^{+}$, also $\alpha_i v_i \in \RR^{+}$. Since $\norm{v_1,\ldots,v_m}_{N'} \leq \norm{v_1,\ldots,v_m}_{M'}$ holds for \emph{any} possible valuation of $v_i \in \RR^{+}$, it also holds for $\alpha_i v_i \in \RR^{+}$. This also holds for $v_i = \norm{x_1,\ldots,x_n}_{V_i}$, so we get $N'(\alpha_1 V_1,\ldots,\alpha_m V_m) \preceq M'(\alpha_1 V_1,\ldots,\alpha_m V_m)$.
%\end{proof}

\begin{lemma}\label{lm:hammer}
Let $\vec{x} = (x_1,\ldots,x_m) \in \RR^m$. Let $N$ be a {\compbanach}, defined over variables $\vec{x}$, such that $x_i$ occurs $k_i$ times in $N$. We have $\norm{\alpha_1 x_1,\ldots,\alpha_m x_m}_p \leq \norm{x_1,\ldots,x_m}_N \leq \norm{\alpha_1 x_1,\ldots,\alpha_m x_m}_q$, where:
\begin{itemize}
\item $p$ is the largest used subnorm of $N$;
\item $q$ is the smallest used subnorm of $N$;
\item $\alpha_i = \sqrt[p]{\sum_{j=1}^{k_i} \alpha_{ij}}$, where $\alpha_{ij}$ is the scaling of the $j$-th occurrence of $x_i$.
\end{itemize}
\end{lemma}

%\begin{proof}
%We prove the first inequality, and the proof would be analogous for the second one. Consider any subnorm $\norm{M_1,\ldots,M_k}_r$ of the norm $N$. Since $p$ is the largest used subnorm of $N$, we have $\norm{M_1,\ldots,M_k}_r \succeq \norm{M_1,\ldots,M_k}_p$. Applying this inequality to every possible subnorm of $N$ and substituting $r$ with $p$, we get a {\compbanach} in which all the subnorms are $\norm{\cdot}_p$ for the same $p \geq 1$. It allows us to apply Lemma~\ref{lm:ungroup} and ungroup all the subnorms. We assume that all scalings in $N$ are applied directly to the variables. Applying ungrouping procedure recursively, we finally reach the scaled variables, getting $N = \norm{\alpha_{11} x_1,\ldots, \alpha_{mk_m} x_m}_p$, where some variables may repeat if they were repeating in different subnorms of $N$ before. We may now use Lemma~\ref{lm:scalecopy} to merge repeating variables into one, rewriting 
%\begin{multline*}
%\norm{\alpha_{11} x_1,\ldots,\overbrace{\alpha_{i1} x_i,\ldots,\alpha_{ik_i} x_i}^{k_i},\ldots, \alpha_{mk_m} x_m}_p = \\ = \norm{\alpha_{11} x_1,\ldots,\sqrt[p]{\sum_{j=1}^{k_i} \alpha_{ij}} x_i,\ldots, \alpha_{mk_m} x_m}_p.
%\end{multline*}
% After doing it for all $i \in [m]$, we get $\norm{\alpha_1 x_1,\ldots,\alpha_m x_m}_p$.
%\end{proof}

\subsection{Smoothing}\label{sec:banach-smoothing}%{{{1

In Sec.~\ref{sec:compute-dersens}, we have shown how to compute a valid upper bound on the derivative sensitivity. We now show how to find a \emph{smooth} upper bound on the derivative sensitivity for particular functions. This is similar to finding smooth upper bounds in~\cite{DBLP:conf/stoc/NissimRS07}. For continuous functions, we have an alternative definition of $\beta$-smoothness, which is easier to use in practice, and implies Def.~\ref{def:smoothness}.

For better presentation of the main ideas behind smoothing, we outline essential steps into several lemmas. Their proofs are straightforward and can be found in App.~\ref{app:proofs}.

%\subsubsection{Smoothing functions}%{{{2

%Suppose we have a function $f : \RR^n \rightarrow \RR$ and we want to find its $\beta$-smooth upper bound.

\begin{lemma}\label{lm:smoothness}
Let $X$ be a Banach space. If $\dersens{f}$ exists then $f : X \rightarrow \RR$ is $\beta$-smooth if $\frac{\dersens{f}(x)}{|f(x)|} \le \beta$ for all $x \in X$.

%A differentiable function $f : \RR \rightarrow \RR$ is $\beta$-smooth if $\left|\frac{f'(x)}{f(x)}\right| \le \beta$.
\end{lemma}

As a particular instance of Lemma~\ref{lm:smoothness}, we get that a differentiable function $f : \RR \rightarrow \RR$ is $\beta$-smooth if $\left|\frac{f'(x)}{f(x)}\right| \le \beta$.

\begin{lemma}\label{lm:sumprod}
Let $f(x): \RR \to \RR$ be $\beta_f$-smooth, and let $g(x): \RR \to \RR$ be $\beta_g$-smooth.
\begin{enumerate}
\item If $f(x),g(x) > 0$, then $f(x) + g(x)$ is $\max(\beta_f,\beta_g)$-smooth;
\item $f(x) \cdot g(x)$ is $\beta_f + \beta_g$-smooth;
\item $f(x)\mathbin{/}g(x)$ is $\beta_f + \beta_g$-smooth.
\end{enumerate}
\end{lemma}

\begin{lemma}\label{lm:lpnormderiv}
Let $X_i$ for $i \in \set{1,\ldots,n}$ be Banach spaces, $f_i:X_i\rightarrow\RR$. Let $x = (x_1,\ldots,x_n)$, and let $f(x_1,\ldots,x_n) = \|f_1(x_1),\ldots,f_n(x_n)\|_p$. Then $\frac{\partial f}{\partial x_i}(x) = \frac{\partial f_i}{\partial x_i}(x_i) \cdot \left(\frac{f_i(x_i)}{f(x)}\right)^{p-1} \leq \frac{\partial f_i}{\partial x_i}(x_i)$.
\end{lemma}

\begin{lemma}\label{lm:lpnormderiv2}
Let $X_i$ for $i \in \set{1,\ldots,n}$ be Banach spaces. Let $x = (x_1,\ldots,x_n)$, and let  $f(x) = \|f_1(x),\ldots,f_n(x)\|_p$. Then $\frac{\partial f}{\partial x_i}(x) = \sum_{j=1}^n \left(\frac{f_j(x)}{f(x)}\right)^{p-1} \cdot \frac{\partial f_j}{\partial x_i}(x)$. This can be upper bounded as:
\begin{enumerate}
\item $\sum_{j=1}^n \frac{\partial f_j}{\partial x_i}(x)$;
\item $\max_{j=1}^n \frac{f(x)}{f_j(x)} \cdot \frac{\partial f_j}{\partial x_i}(x)$.
\end{enumerate}
\end{lemma}

\begin{lemma}\label{lm:lpnorm}
Let $X_i$ for $i \in \set{1,\ldots,n}$ be Banach spaces, $X = \prod_{i=1}^{n} X_i$. Let $f_i: X_i \to \RR$ be $\beta_i$-smooth. Then, $f(x_1,\ldots,x_n) = \|f_1(x_1),\ldots,f_n(x_n)\|_p$ is $\norm{(\beta_i)_{i=1}^{n}}_{p}$-smooth as well as $\max_{i=1}^n (\beta_i)$-smooth, where the norm of $X$ is the $\ell_{p/(p-1)}$-combination of the norms of all $X_i$.
%\begin{enumerate}
%\item $\max_{i \in \set{1,\ldots,n}}(\beta_i)$-smooth in $(X,\ell_p)$.
%\item $\norm{(\beta_i)_{i=1}^{n}}_{p}$-smooth in $(X,\ell_{\frac{p}{p-1}})$.
%\end{enumerate}
\end{lemma}

In general, the smoothness is worse if the variables of different $f_i$ are not disjoint. This is shown in the next lemma.
\begin{lemma}\label{lm:lpnorm2}
Let $X_i$ for $i \in \set{1,\ldots,n}$ be Banach spaces, $X = \prod_{i=1}^{n} X_i$. Let $f_i: X \to \RR$ be $\beta^j_i$-smooth for $X_j$. Let $x = (x_1,\ldots,x_n)$. Then, $f(x) = \|f_1(x),\ldots,f_n(x)\|_p$ is $\norm{(\max_j \beta^j_i)_{i=1}^{n}}_{p}$-smooth.
\end{lemma}

\begin{lemma}[a part of Lemma.2.3 of~\cite{DBLP:conf/stoc/NissimRS07}]\label{lm:smoothub}
Let $f: X \to \RR$. For $\beta > 0$, a $\beta$-smooth upper bound on $f$ is
\[g_{f,\beta} = \max_{x' \in X}(f(x) \cdot e^{-\beta \cdot d(x,x')})\enspace.\]
\end{lemma}

\subsection{Smooth upper bounds of functions and derivatives}\label{sec:concretesmoothub}

We now show how to find $\beta$-smooth upper bounds for some particular functions. %The idea behind smoothing is similar to Lemma 2.3 of~\cite{DBLP:conf/stoc/NissimRS07}, where for $x \in X$ the local sensitivity is first maximized over $y \in X$, and then multiplied with $e^{-\beta d(x,y)}$ to get a $\beta$-smooth upper bound. The multiplication needs to be applied only in the range where the initial function is not $\beta$-smooth.
For univariate functions, we consider ordinary derivatives over $\RR$, where the norm is the absolute value. For multivariate functions, we can find the derivative w.r.t. a certain norm that depends on the function, and in general it it is not possible to compute the derivative over an arbitrary norm. We discuss how to solve this problem in Sec.~\ref{sec:qdbnorms}.

To achieve differential privacy, we will also need to find smooth upper bounds on their \emph{derivatives}. Suppose we have a function $f : \RR^n \rightarrow \RR$ and want to find $\beta$-smooth upper bounds on both $f$ and $\dersens{f}$. Let us denote them $\ubf{f}$ and $\ubdsf{f}$ respectively. Below, we give the smooth upper bounds for certain basic functions and constructions. These can be composed to a certain extent and some examples of that can be found in App.~\ref{app:functions}. An example of computing smooth upper bound on derivative sensitivity for a particular function can be found in App.~\ref{app:theoryexample2}.

\textbf{Power function.} Let $f(x) = x^r, r\in\RR_+, x>0$. We have
\[ \frac{f'(x)}{f(x)} = \frac{rx^{r-1}}{x^r} = \frac{r}{x} \]
\[ \left|\frac{r}{x}\right| \le \beta \Leftrightarrow x \ge \frac{|r|}{\beta} \]
For $x \leq \frac{r}{\beta}$, the function $f'(x)$ achieves its maximum at the point $\frac{r}{\beta}$. By Lemma~\ref{lm:smoothub}, a $\beta$-smooth upper bound on $f$ is
\[ \ubf{f}(x) =
\begin{cases}
x^r & \mbox{if $x \ge \frac{r}{\beta}$} \\
e^{\beta x - r} \left(\frac{r}{\beta}\right)^r & \mbox{otherwise}
\end{cases} \]

If $r\ge 1$, we may also find a smooth upper bound on the derivative sensitivity $\dersens{f}$ of $f$. We have
\[ \dersens{f}(x) = |f'(x)| = |r|x^{r-1} \enspace.\]
A $\beta$-smooth upper bound on $\dersens{f}$ is
\[ \ubdsf{f}(x) =
\begin{cases}
rx^{r-1} & \mbox{if $x \ge \frac{r-1}{\beta}$} \\
re^{\beta x - (r-1)} \left(\frac{r-1}{\beta}\right)^{r-1} & \mbox{otherwise}
\end{cases} \]

\textbf{Identity.} Let $f(x) = x, x\in\RR$. Here $x$ may also be negative. By Lemma~\ref{lm:smoothub},  a $\beta$-smooth upper bound on $f$ is
\[ \ubf{f}(x) =
\begin{cases}
|x| & \mbox{if $|x| \ge \frac{1}{\beta}$} \\
\frac{e^{\beta |x| - 1}}{\beta} & \mbox{otherwise}
\end{cases} \]

The upper bound on $\dersens{f}(x)$ is much simpler. We have
\[ \dersens{f}(x) = 1 \enspace,\]
which is trivially $\beta$-smooth for any $\beta$.

\textbf{Exponent.} Let $f(x) = e^{rx}, r\in\RR, x\in\RR$. We have $\dersens{f}(x) = |f'(x)| = |r|e^{rx}$, hence:
\begin{itemize}
\item $\abs{\frac{f'(x)}{f(x)}} = \frac{re^{rx}}{e^{rx}} = r$ ;
\item $\abs{\frac{f''(x)}{f(x)}} = \frac{r^2 e^{rx}}{re^{rx}} = r$.
\end{itemize}
Thus both $f$ and $\dersens{f}$ are $\beta$-smooth if $|r| \le \beta$. Since $r$ is a constant that does not depend on $x$, we constrain ourselves to functions that satisfy this conditions, and for larger $r$ we may just increase $\beta$, which in turn increases the noise of differential privacy.

\textbf{Sigmoid.} Consider the (sigmoid) function $\sigma(x) = \frac{e^{\alpha x}}{e^{\alpha x} + 1}$. This function can be viewed as a continuous approximation of the indicator function $I_{\RR_+}:\RR\to \set{0,1}$, which is less precise for values close to $0$, and the error decreases when $\alpha$ increases. We have:
\begin{itemize}
\item $\sigma'(x) = \frac{\alpha e^{\alpha x}}{(e^{\alpha x} + 1)^2}$;
\item $\sigma''(x) = \frac{\alpha^2 e^{\alpha x}(e^{\alpha x} - 1)}{(e^{\alpha x} + 1)^3}$;
\item $\left|\frac{\sigma'(x)}{\sigma(x)}\right| = \left|\alpha\cdot\frac{1}{e^{\alpha x} + 1}\right| \leq \alpha$;
\item $\left|\frac{\sigma''(x)}{\sigma'(x)}\right| = \left|\alpha \cdot \frac{e^{\alpha x} - 1}{e^{\alpha x} + 1}\right| \leq \alpha$.
\end{itemize}
Thus both $\sigma(x)$ and $DS_{\sigma(x)} = \abs{\sigma'(x)}$ are $\alpha$-smooth. If we want less DP noise, we should decrease $\alpha$, which in turn makes the sigmoid itself less precise, so there is a tradeoff.

\textbf{Tauoid.} Consider the function $\tau(x) = \frac{2}{e^{-\alpha x} + e^{\alpha x}}$ (let us call it a \emph{tauoid}). This function can be viewed as a continuous approximation of the indicator function $I_{\set{0}}:\RR \to \set{0,1}$, which works similarly to a sigmoid. We have:
\begin{itemize}
\item $\tau'(x) = - \frac{2\alpha(e^{\alpha x} - e^{-\alpha x})}{(e^{-\alpha x} + e^{\alpha x})^2} = \frac{2\alpha(e^{-\alpha x} - e^{\alpha x})}{(e^{-\alpha x} + e^{\alpha x})^2} = \frac{2\alpha e^{\alpha x}(1 - e^{2\alpha x})}{(1 + e^{2\alpha x})^2}$;
\item $\left|\frac{\tau'(x)}{\tau(x)}\right| = \frac{|\alpha| \cdot |e^{-\alpha x} - e^{\alpha x}|}{e^{-\alpha x} + e^{\alpha x}} \le |\alpha|$;
\item $|\tau'(x)| \le \frac{2|\alpha| e^{\alpha x}}{1 + e^{2\alpha x}} = \frac{2|\alpha|}{e^{-\alpha x} + e^{\alpha x}} = |\alpha|\tau(x) =: \ubdsf{\tau}(x)$;
\item $\ubdsf{\tau}'(x) = |\alpha|\tau'(x)$;
\item $\left|\frac{\ubdsf{\tau}'(x)}{\ubdsf{\tau}(x)}\right| = \left|\frac{\tau'(x)}{\tau(x)}\right| \le |\alpha|$.
\end{itemize}
Thus both $\tau$ itself and $\ubdsf{\tau}$, an upper bound on its derivative sensitivity, are $\alpha$-smooth.

%TODO we could leave it if there is space, as a simpler example
%\textbf{The $\ell_2$-norm.} Let $f(x,y) = \sqrt{x^2+y^2}, x,y\in\RR$.
%\[ \nabla f(x,y) = \left(\frac{x}{\sqrt{x^2+y^2}},\frac{y}{\sqrt{x^2+y^2}}\right) \]
%The derivative sensitivity of $f$ in $(\RR^2,\ell_2)$ is
%\[ \dersens{f}(x,y) = \left\|\left(\frac{x}{\sqrt{x^2+y^2}},\frac{y}{\sqrt{x^2+y^2}}\right)\right\|_2 = 1 \]

\textbf{An $\ell_p$-norm.} Consider the function $f(x) = \|x\|_p = \left(\sum x_i^p\right)^{1/p}$, $x\in\RR^n,x = (x_1,\ldots,x_n)$. We have

\[ \frac{\partial f}{\partial x_i}(x) = \frac{px_i^{p-1}}{p\left(\sum x_i^p\right)^{(p-1)/p}} =
\frac{x_i^{p-1}}{\left(\sum x_i^p\right)^{(p-1)/p}} = \left(\frac{x_i}{\|x\|_p}\right)^{p-1} \enspace.\]

% TODO maybe this could be moved into appendix, I am not sure if we use it anywhere
%The derivative sensitivity of $f$ in $(\RR^n,\ell_{r/(r-1)})$ is
%\begin{eqnarray*}
%\dersens{f}(x) & = & \left(\sum_i \left(\frac{x_i}{\|x\|_p}\right)^{(p-1)r}\right)^{1/r} =
%\frac{1}{\|x\|_p^{p-1}} \left(\sum_i x_i^{(p-1)r}\right)^{1/r} \\
%& = &
%%\frac{\|x\|_{(p-1)r}^{p-1}}{\|x\|_p^{p-1}} =
%\left(\frac{\|x\|_{(p-1)r}}{\|x\|_p}\right)^{p-1}
%\end{eqnarray*}

%The derivative sensitivity of $f$ in $(\RR^n,\ell_q)$ (where $q = \frac{p}{p-1}$) is
By Lemma~\ref{lemma:banach-base}, the derivative sensitivity of $f$ in $(\RR^n,\ell_p)$ is
\[ \dersens{f}(x) = \left(\sum \left(\frac{x_i^{p-1}}{\left(\sum x_i^p\right)^{\frac{p-1}{p}}}\right)^{\frac{p}{p-1}}\right)^{\frac{p-1}{p}} =
\left(\sum \frac{x_i^p}{\sum x_i^p}\right)^{\frac{p-1}{p}} = 1\enspace.
\]
This is constant and thus $\beta$-smooth for all $\beta$. The function $f$ itself is $\beta$-smooth if $\frac{1}{\|x\|_p} \le \beta$, i.e. if $\|x\|_p \ge \frac{1}{\beta}$.
By Lemma~\ref{lm:smoothub}, a $\beta$-smooth upper bound on $f$ is
\[ \ubf{f}(x) =
\begin{cases}
\|x\|_p & \mbox{if $\|x\|_p \ge \frac{1}{\beta}$} \\
\frac{e^{\beta\|x\|_p - 1}}{\beta} & \mbox{otherwise}
\end{cases} \]
This also holds for $p = \infty$.

\textbf{The $\ell_{\infty}$-norm.} Let $f(x) = \|x\|_\infty = \max_i \abs{x_i}$. We have
\[ \frac{\partial f}{\partial x_i}(x) =
\begin{cases}
1 & \mbox{if $i = \argmax_j \abs{x_j}$} \\
\mbox{undefined} & \mbox{if $\argmax_j \abs{x_j}$ is not unique} \\
0 & \mbox{otherwise}
\end{cases}
\]
The derivative sensitivity of $f$ in $(\RR^n,\ell_\infty)$ is
\[ \dersens{f}(x) =
\begin{cases}
1 & \mbox{if $\argmax_j \abs{x_j}$ is unique} \\
\mbox{undefined} & \mbox{if $\argmax_j \abs{x_j}$ is not unique} \\
\end{cases}
\]
Because we are interested in upper bounds on the derivative sensitivity, we define
\[ \dersens{f}(x_0) := \limsup_{x\rightarrow x_0} \dersens{f}(x) = 1 \]
for those $x_0$ for which $\dersens{f}(x_0)$ is undefined.
Thus $\dersens{f}(x) = 1$, which is constant and $\beta$-smooth for all $\beta$. The smooth upper bound on the function $f$ itself can be found similarly to the $\ell_p$-norm case.

\textbf{Product.} Let $f:\prod_{i=1}^n X_i \rightarrow \RR,f(x_1,\ldots,x_n) = \prod_{i=1}^n f_i(x_i)$ where $X_i$ are
Banach spaces. Let $X = \prod_{i=1}^n X_i$ and $x = (x_1,\ldots,x_n)$. First, suppose that variables $x_i$ are independent. We have $\frac{\partial f}{\partial x_i}(x) = \prod_{i\neq j=1}^n f_j(x_j)\cdot f_i'(x_i)$, and $\abs{\frac{\partial f}{\partial x_i}(x) \cdot \frac{1}{f(x)}} = \abs{\frac{f_i'(x_i)}{f_i(x_i)}}$, hence:
\begin{itemize}
\item If $\left|\frac{f_i'(x_i)}{f_i(x_i)}\right| \le \beta$, then $f$ is $\beta$-smooth w.r.t. $x_i$.\\
%\item If $f$ is $\beta$-smooth in $(X,\ell_p)$ then it is $\beta$-smooth w.r.t. all $x_i$.\\
%\item If $f$ is $\beta$-smooth w.r.t. all $x_i$ then it is $\beta$-smooth in $(X,\ell_1)$ (this direction does not hold for the general $\ell_p$-norm). It is $n\beta$-smooth in $(X,\ell_\infty)$.
\item By Lemmas~\ref{lemma:banach-base} and~\ref{lemma:banach-combine}, $\norm{df_x} = \norm{\left( \prod_{i\neq j=1}^n f_j(x_j)\cdot f_i'(x_i) \right)_{i=1}^{n}}_{\frac{p}{p-1}}$ in $(X,\ell_p)$, so we have
\[\frac{\norm{df_x}}{\abs{f(x)}} = \frac{\norm{df_x}}{\abs{\prod_{i=1}^n f_i(x_i)}} = \norm{\left( \frac{f_i'(x_i)}{f_i(x_i)} \right)_{i=1}^{n}}_{\frac{p}{p-1}} \leq \norm{\left( \beta_i \right)_{i=1}^{n}}_{\frac{p}{p-1}} \enspace,\]
where $\beta_i$ is the smoothness of $f_i$. Hence, if $f_i$ is $\beta$-smooth w.r.t. $x_i$ for all $i$, then $f$ is $\beta$-smooth in $(X,\ell_1)$ and $n\beta$-smooth in $(X,\ell_\infty)$.
%In $(X,\ell_1)$, $\frac{\norm{df_x}}{\abs{f(x)}} = \frac{\norm{\prod_{i=2}^{n}f_i(x_i) \cdot f'_1(x_1),\ldots,\prod_{i=1}^{n-1}f_i(x_i) \cdot f'_n(x_n)}_{\infty}}{\abs{\prod_{i=1}^{n}f_i(x_i)}} = \abs{\frac{f'_j(x_j)}{f_j(x_j)}}$ for some $j$
\end{itemize}
The derivative sensitivity of $f$ w.r.t. $x_i$ is $c_i(x) = \dersens{f_i}(x_i) \cdot \left|\frac{f(x)}{f_i(x_i)}\right|$.
The derivative sensitivity of $f$ in $(X,\ell_p)$ is, by Lemma~\ref{lemma:banach-combine},
\[ \dersens{f}(x) = \|(c_1(x),\ldots,c_n(x))\|_{\frac{p}{p-1}} = \left\|\left(\frac{\dersens{f_i}(x_i)}{|f_i(x_i)|}\right)_{i=1}^n\right\|_{\frac{p}{p-1}} \cdot |f(x)|\enspace. \]

We have $c_i(x) = \dersens{f_i}(x_i) \cdot \abs{\frac{f(x)}{f_i(x_i)}} = \dersens{f_i}(x_i) \cdot \prod_{j\neq i} \abs{f_j(x_j)}$. Since $\prod_{j\neq i} \abs{f_j(x_j)}$ does not depend on $x_i$ and $\dersens{f_i}(x_i) \geq 0$, by Lemma~\ref{lm:sumprod}, if $\dersens{f_i}$ is $\beta$-smooth in $X_i$ then $c_i(x)$ is also $\beta$-smooth in $X_i$. Similarly, if $f_j(x_j)$ is $\beta$-smooth, then $c_i(x)$ is also $\beta$-smooth in $X_j$. Hence, if $f_i$ and $\dersens{f_i}$ are $\beta$-smooth for all $i$, by Lemma~\ref{lm:lpnorm2}, $\dersens{f}$ is
%if $f_i$ and $\dersens{f_i}$ are $\beta$-smooth in $X_i$ then $c_i(x)$ is also $\beta$-smooth in $X_i$. If it holds for all $i$, by Lemma~\ref{lm:lpnorm}, $\dersens{f}$ is
$\beta$-smooth in $(X,\ell_1)$ and $n\beta$-smooth in $(X,\ell_\infty)$.
If $\dersens{f_i}$ are not all $\beta$-smooth then we can use their $\beta$-smooth upper bounds when computing $c_i$.
Then we get a $\beta$-smooth upper bound on $\dersens{f}$ instead of the actual $\dersens{f}$.

We may also consider the case where the variables $x_i$ are fully
dependent, i.e. equal (the case where they are partially dependent is currently not considered).
Consider a function $f(x) = g(x) \cdot h(x)$ where $g,h : X \rightarrow \RR_+$ and $X$ is a Banach space. We have
\[ \dersens{f}(x) = g(x) \cdot \dersens{h}(x) + h(x) \cdot \dersens{g}(x) \enspace.\]
By Lemma~\ref{lm:sumprod}, if $g$ is $\beta_g$-smooth, $h$ is $\beta_h$-smooth, $\dersens{g}$ is $\beta_{g'}$-smooth, and $\dersens{h}$ is
$\beta_{h'}$-smooth, then $\dersens{f}$ is $\max(\beta_g + \beta_{h'}, \beta_{h} + \beta_{g'})$-smooth.
The function $f$ itself is $(\beta_g + \beta_h)$-smooth.

\textbf{Sum.} Let $f:\prod_{i=1}^n X_i \rightarrow \RR,f(x_1,\ldots,x_n) = \sum_{i=1}^n f_i(x_i)$ where $X_i$ are
Banach spaces. Let $X = \prod_{i=1}^n X_i$ and $x = (x_1,\ldots,x_n)$. First, suppose that the variables $x_i$ are independent.
The derivative sensitivity of $f$ w.r.t. $x_i$ is $\dersens{f_i}(x_i)$.
By Lemmas~\ref{lemma:banach-base} and~\ref{lemma:banach-combine}, the derivative sensitivity of $f$ in $(X,\ell_p)$ is $\dersens{f}(x) = \|\dersens{f_1}(x_1),\ldots,\dersens{f_n}(x_n)\|_{\frac{p}{p-1}}$. 

\begin{itemize}
\item Let $f_i \geq 0$ for all $i \in \set{1,\ldots,n}$ (or $f_i \leq 0$ for all $i \in \set{1,\ldots,n}$) and $\beta_i$-smooth w.r.t. $X_i$. Now we have $\abs{f(x)} = \sum_{i=1}^n \abs{f_i(x_i)} = \norm{\abs{f_i(x_i)}_{i=1}^n}_1$. By Lemma~\ref{lm:lpnorm}, $f(x)$ is $\beta := \max_i(\beta_i)$-sensitive in $(X,\ell_p)$. We do not get a good bound in the case when $f_i$ may have different signs, since then $f_i(x)$ may cancel each other out and make $f(x)$ arbitrarily small even if $\abs{f_i(x)}$ are large.
\item Let $\dersens{f_i}$ be $\beta_i$-smooth for $i \in \set{1,\ldots,n}$. By Lemma~\ref{lm:lpnorm}, $\dersens{f}$ is $\norm{(\beta_i)_{i=1}^n}_{\frac{p}{p-1}}$-smooth in $(X,\ell_p)$, and if all $\dersens{f_i}$ are $\beta$-smooth, then $\dersens{f}$ is also $\beta$-smooth.
\end{itemize}

Consider the case where $x_i$ are equal: $f(x) = \sum_{i=1}^n g_i(x)$ where $g_i : X \rightarrow \RR$ and $X$ is a Banach
space. Then
\[ \dersens{f}(x) = \sum_{i=1}^n \dersens{g_i}(x) \]
By Lemma~\ref{lm:sumprod}, if all $\dersens{g_i}$ are $\beta$-smooth then $\dersens{f}$ is $\beta$-smooth.
If all $g_i$ are non-negative and $\beta$-smooth then $f$ is $\beta$-smooth.

\textbf{Min / max.} Let $f:\prod_{i=1}^n X_i \rightarrow \RR,f(x_1,\ldots,x_n) = \min_{i=1}^n f_i(x_i)$ where $X_i$ are
Banach spaces (the case with $\max$ instead of $\min$ is similar). Let $X = \prod_{i=1}^n X_i$ and $x = (x_1,\ldots,x_n)$. Let the variables $x_i$ be independent.

If for all $i$, $f_i$ is $\beta$-smooth in $X_i$ then $f$ is $\beta$-smooth in $(X,\ell_p)$. The same holds with $\max$
or sum (with non-negative $f_i$) or $\ell_{p'}$-norm instead of $\min$.

The derivative sensitivity of $f$ w.r.t. $x_i$ is $\dersens{f_i}(x_i)$ if $i = \argmin f_i(x_i)$ and $0$ otherwise.
The derivative sensitivity of $f$ in $(X,\ell_p)$ is $\dersens{f}(x) = \dersens{f_i}(x_i)$ where $i = \argmin f_i(x_i)$.
In general, $\dersens{f}$ is discontinuous at points where $\argmin f_i(x_i)$ is not unique.
One $\beta$-smooth (in $(X,\ell_p)$) upper bound on $\dersens{f}$ is $\max c_i(x_i)$ where $c_i$ is a $\beta$-smooth upper bound on
$\dersens{f_i}$.

\textbf{Norm scaling.} Let $f:X\rightarrow\RR$ in the Banach space $(X,\norm{\cdot})$. Scaling the norm by $a$ scales the derivative $f'(x)$ by $\frac{1}{a}$ while keeping the value of $f(x)$ the same. Hence, if $f$ is $\beta$-smooth in $(X,\norm{\cdot})$ then it is
$\frac{\beta}{a}$-smooth in $(X,a\cdot \norm{\cdot})$. 

Let $c(x)$ be a $\beta$-smooth upper bound on the derivative sensitivity of $f$ at $x$ in $(X,\norm{\cdot})$.
Then $\frac{c(x)}{a}$ is a $\frac{\beta}{a}$-smooth upper bound on the derivative sensitivity of $f$ at $x$ in $(X,a\cdot\norm{\cdot})$.

This construction will be very useful in the cases when we want to compute the sensitivity of some function w.r.t. a different norm.

\textbf{Constants.} Consider the constant function $f(x) = c, f:X\rightarrow\RR$ where $X$ is a Banach space (with arbitrary norm).
Then $f$ is $0$-smooth, $\dersens{f}(x) = 0$, and $\dersens{f}$ is also $0$-smooth.

\textbf{Composition with a real function.} Let $f(x) = h(g(x)), x\in X$ where $g : X \rightarrow \RR, h : \RR \rightarrow \RR$ and $X$ is a Banach space.
\[ \dersens{f}(x) = |h'(g(x))| \cdot \dersens{g}(x) \]
\[ \frac{\dersens{f}(x)}{|f(x)|} = \frac{|h'(g(x))|}{|h(g(x))|} \cdot \dersens{g}(x) \]
Suppose that $h$ is $\beta_h$-smooth and $\dersens{g}(x) \le B$ for all $x$. Then $f$ is $\beta_h B$-smooth.
\[ \dersens{\dersens{f}}(x) = |h''(g(x))| (\dersens{g}(x))^2 + |h'(g(x))| \cdot \dersens{\dersens{g}}(x) \]
\[ \frac{\dersens{\dersens{f}}(x)}{\dersens{f}(x)} = \frac{|h''(g(x))|}{|h'(g(x))|} \cdot \dersens{g}(x) + \frac{\dersens{\dersens{g}}(x)}{\dersens{g}(x)} \]
By Lemma~\ref{lm:sumprod}, if $h'$ is $\beta_{h'}$-smooth, $\dersens{g}$ is $\beta_{g'}$-smooth, and $\dersens{g}(x) \le B$ for all $x$ then
$\dersens{f}$ is $(\beta_{h'} B + \beta_{g'})$-smooth.

% vim: set foldmethod=marker foldmarker=%{{{,%}}}:
%------------------------------
\section{Application to SQL Queries}\label{sec:sqlapplication}

In this section, we describe how the theory of Sec.~\ref{sec:banach-theory} has been applied to SQL queries. Our theory deals with functions that return a numeric value, so the query should end up in a single output. Another constraint is that the theory only supports continuous functions. That is, we consider queries of the form
\[\texttt{SELECT}\ \textit{aggr expr}\ \texttt{FROM t1 AS s1,...,tn AS tn WHERE}\ \textit{condition}\enspace,\]
where:
\begin{itemize}
\item \emph{expr} is an expression over table columns, computed as a continuous function. 
\item \emph{condition} is a boolean expression over predicates $P(x) \in \set{x < 0, x = 0}$, where $x$ is an expression of the same form as \emph{expr}. Since all functions have to be continuous, these predicates are computed using \emph{sigmoid} and \emph{tauoid} functions defined in Sec~\ref{sec:banach-smoothing}.

\item $aggr \in \set{\mathtt{SUM}, \mathtt{COUNT}, \mathtt{PRODUCT}, \mathtt{MIN}, \mathtt{MAX}}$ is implemented using corresponding functions of Sec.~\ref{sec:banach-smoothing}.
\end{itemize}

\subsection{Banach spaces of databases}\label{sec:banach-db}%{{{1

Suppose we have a database of $n$ tables. The schema of each table, and its number of rows are fixed. %If necessary, a sensitive (but bounded) number of rows can be simulated using an extra (sensitive) column containing a flag for each row that determines whether the row is used or not.

Our ultimate goal is to enforce differential privacy w.r.t. a certain \emph{component} of a database. A component corresponds to a subset of columns and rows of a data table. The set of all (sensitive) entries can be viewed as a vector over $\RR^m$. By default, the $\ell_{\infty}$-norm of all sensitive table entries can be computed, so that computed noise is enough to ensure differential privacy w.r.t. \emph{each} sensitive entry: $\ell_{\infty}$-norm changes by $d$ if any (or all) of its arguments change by $d$. Alternatively, it can be reasonable to combine the norm of several columns, e.g. treating the Euclidean norm of space coordinates as distance. This depends on the data that the owner actually wants to hide, so we allow to define a customized norm for each data table. We assume that the components are independent, i.e. if a certain component's value has changed, it will not prompt a change in some other component's value. The vector of sensitive entries together with the norm forms a Banach space.

Let $T_i = R_i^{n_i}$ be the Banach space for the
$i$-th table, where $n_i$ is the number of rows in the $i$-th table and $R_i$ is the Banach space of potential values of a row of the $i$-th table.
The input contains a tuple of $n$ tables, which is an element of $D = T_1 \times \cdots \times T_n$. We can make $D$ a
Banach space by combining the norms of $T_i$ using any $\ell_p$-norm.  If we take $\ell_{\infty}$-norm, then we achieve differential privacy w.r.t. each table $T_i$ at once.

\subsubsection{Query without a filter.}
To make a query on the database, we want to join those $n$ tables. Consider an input $(t_1,\ldots,t_n) \in D$.
Let us first consider the cross product $t = t_1 \times \cdots \times t_n$, i.e. joining without any filters.
First, let us assume that the $n$ joined tables are distinct, i.e. no table is used more than once.
Each row of the cross product is an element of $R = R_1 \times \cdots \times R_n$, thus $t$ is an element of
$T = R^{n_1\cdots n_n}$.

The query contains an aggregating function $f : T \rightarrow \RR$. All non-sensitive entries of the data tables are treated as constants. Our analyzer supports the real-valued functions listed in Sec.~\ref{sec:banach-smoothing}; let us call them \emph{basic} functions. These functions are summarized in Table~\ref{tab:funbasic}, where $x_i$ belongs to some Banach space $X_i$ with norm $N_i$, and $LUB(N_1,\ldots,N_n)$ is a norm $N$ such that $N_i \preceq N$ for all $i \in \set{1,\ldots,n}$. The column ``Norm'' in Table~\ref{tab:funbasic} indicates, according to which norm(s) of $T$ we can ``naturally'' compute the smooth derivative sensitivity of $f$. We allow some more functions, which do not not require any additional smoothness analysis and are used as syntactic sugar, being reduced to the basic functions. These functions are given in Table~\ref{tab:funext}.

\begin{table}[t]
\caption{Basic functions}\label{tab:funbasic}
\begin{center}
\begin{tabular}{| c | c | c |}
\hline
 Function & Variables  & Norm \\ 
          & may repeat &      \\
\hline
\hline
$\min(x_1,\ldots,x_n)$ & no & $\norm{N_1,\ldots,N_n}_p$ \\ \hline
$\max(x_1,\ldots,x_n)$ & no & $\norm{N_1,\ldots,N_n}_p$ \\ \hline
$e^{rx}$ & -- & $N$ \\ \hline
$x^r$ (for $r \in \RR^{+}$) & -- & $N$ \\ \hline
$\sum(x_1,\ldots,x_n)$ & no & $\norm{N_1,\ldots,N_n}_p$ \\ \hline
$\prod(x_1,\ldots,x_n)$ & no & $\norm{N_1,\ldots,N_n}_p$ \\ \hline
$\sum(x_1,\ldots,x_n)$ & yes & $LUB(N_1,\ldots,N_n)$ \\ \hline
$\prod(x_1,\ldots,x_n)$ & yes & $LUB(N_1,\ldots,N_n)$ \\ \hline
$\norm{x_1,\ldots,x_n}_p$ & no & $\norm{N_1,\ldots,N_n}_p$ \\
\hline
\end{tabular}
\end{center}
\end{table}

\begin{table}[t]
\caption{Extended functions}\label{tab:funext}
\begin{center}
\begin{tabular}{| c | c | c |}
\hline
 Function & Reduction to basic & Norm \\
\hline
\hline
%$x - y$ & $x + (-1) \cdot y$ \\ \hline
$|x|$ & $\norm{x}_1$ & $|x|$\\ \hline
$x^r$ (for $r \in \RR^{-}$) & $e^{-r\ln(x)}$ & $\ln(x)$ \\ \hline
$x / y$ & $x \cdot e^{-\ln(y)}$ & $\norm{x,\ln(y)}_p$\\
\hline
\end{tabular}
\end{center}
\end{table}

We can consider $t\in T$ as one large table whose number
of rows is the product of the numbers of rows of tables $t_i \in T_i$ and number of columns is the sum of the numbers of
columns of tables $t_i \in T_i$. We can compute the derivative
sensitivity of $f$ w.r.t. the components of $t$ specified by a subset of rows and a subset of columns of $t$ as follows.

Suppose we want to compute the derivative sensitivity of $f$ w.r.t. a row $r$ of $t_i$. We can compute the
sensitivity w.r.t. the following component of $t$: the subset of rows is the set of rows affected by $r$, i.e.
$t_r = t_1 \times \cdots \times t_{i-1} \times \{r\} \times t_{i+1} \times \cdots \times t_n$, the subset of columns is the
set of columns corresponding to table $t_i$. To perform the computation, we also need to specify a norm for combining
the rows of $t$. Because changing $r$ by distance $d$ changes each row of $t_r$ by distance $d$, we must combine the
rows of $t_r$ by $\ell_\infty$-distance to ensure that $t_r$ also changes by distance $d$.

Suppose we now want to compute the sensitivity of $f$ w.r.t. a subset $s$ of rows of $t_i$.
Then the subset of rows affected by $s$, is
$t_s = t_1 \times \cdots \times t_{i-1} \times s \times t_{i+1} \times \cdots \times t_n$.
Each row in $t_s$ is affected by exactly one row in $s$.
Let $s = {r_1,\ldots,r_k}$ and let $t_s = \bigcup_{j=1}^k t_{r_j}$ where $t_{r_j}$ is the subset of rows affected by
$r_j$. The $t_{r_j}$ are disjoint. Let $t_{r_j} = \{u_{j1},\ldots,u_{jm_j}\}$. Then the norm for $u_{jm}$ is the same as
the norm for $r_j$, with the additional columns having zero norm. The norm for $t_{r_j}$ is computed by combining the
norms for $u_{jm}$ using $\ell_\infty$-norm. The norm for $t_s$ is then computed by combining the norms for
$t_{r_j}$ using the norm that combined the norms of the rows of $t_i$, i.e. $\ell_{p_i}$.

\subsubsection{Query with a filter.}

A filter that \emph{does not depend} on sensitive data can be applied directly to the cross product of the input tables, and we may then proceed with the query without a filter. A filter that \emph{does depend} on sensitive data is treated as a part of the query. It should be treated as a continuous function, applied in such a way that the the discarded rows would be ignored by the aggregating function. We combine sigmoids and tauoids to obtain the approximated value of the indicator $\sigma(x_i) \in \set{0,1}$, denoting whether the row $x_i$ satisfies the filter. Let $f_i$ be the function applied to the row $x_i$ before aggregation. The value $\sigma(x_i)$ is treated by different aggregators as follows.

\begin{itemize}
\item $\mathrm{SUM}$. The values $0$ do not affect the sum, so we may compute sum over $f_i(x_i) \cdot \sigma(x_i)$.

\item $\mathrm{COUNT}$: We may apply the $\mathrm{SUM}$ function directly to the output of $\sigma(x_i)$, counting all entries for which $\sigma(x_i) = 1$. In this case, we do not use $f_i$ at all, and the sensitivity only depends on the sensitivity of $\sigma$.

\item $\mathrm{PRODUCT}$: The values discarded by the filter should be mapped to $1$, so that they would not affect the product. The simplest way to do it is to take $\sigma(x_i)\cdot f_i(x_i) + (1 - \sigma(x_i))$. Here the problem of sigmoid approach is that the accuracy error grows fast with the number of rows due to multiplication.

\item $\mathrm{MIN}$, $\mathrm{MAX}$: We need to map $f_i(x_i)$ to special values that would be ignored by these functions. If we take $f_i(x_i) \cdot \sigma(x_i)$, it will be correct only if all inputs of $\mathrm{MAX}$ have non-negative values, or all inputs of $\mathrm{MIN}$ have non-positive values. In general, we need to map dropped values to $\max_i{f(x_i)}$ for MIN, and $\min_i{f(x_i)}$ for $\mathrm{MAX}$, doing it only for the rows where $\sigma(x_i) = 0$, and keep $f_i(x_i)$ for the rows where $\sigma(x_i) = 1$. One possible solution is to add or subtract the largest difference of any two values $\Delta(x_1,\ldots,x_n) := \mathrm{MAX}(f(x_1),\ldots,f(x_n)) - \mathrm{MIN}(f(x_1),\ldots,f(x_n))$, where the $\mathrm{MIN}$ and $\mathrm{MAX}$ are now applied to \emph{all} rows. Instead of applying $\mathrm{MIN}$ directly to $f_i(x_i)$, we apply it to $f_i(x_i) + (1 - \sigma(x_i)) \cdot \Delta(x_1,\ldots,x_n)$. Similarly, we apply $\mathrm{MAX}$ to $f_i(x_i) - (1 - \sigma(x_i)) \cdot \Delta(x_1,\ldots,x_n)$. We note that the answer would be wrong in the case when no rows satisfy the filter. We find it fine, since in that case a special N/A value should be returned, and the fact that ``the answer is N/A'' may itself be sensitive, so uncovering it would in general violate differential privacy.
\end{itemize}

If we know that the compared values are integers and hence $d(x,x') \geq 1$ for $x \neq x'$, we can do better than use sigmoids or tauoids, defining precise functions:
\begin{itemize}
\item $x > y \iff \min(1, \max(0, x - y))$.
\item $x = y \iff 1 - \min(1,\max(0, \abs{x1 - x2}))$.
\end{itemize}
If the filter expression is a complex boolean formula over several conditions, an important advantage of these functions is that they do not lose precision due to addition and multiplication.

For real numbers, we may bound the precison and assume e.g. that $d(x,x') \geq 1/k$ for some $k \geq 1$, which allows to use similar functions. The sensitvity of such comparisons will be $k$ times larger than for integers.

\subsection{Query Norm vs Database Norm}\label{sec:qdbnorms}

In Sec.~{\REFBANACHSmooth}, we described how to compute smooth upper bounds on function sensitivity w.r.t. certain norms. These \emph{query norms} (denoted $N_{query}$; listed in the last column of Tables~\ref{tab:funbasic} and~\ref{tab:funext}) can be different from the \emph{database norm} (denoted $N_{db}$) specified by the data owner. We describe how our analyser solves this problem, so that the computed sensitivity would not be underestimated, and the computed noise would be sufficient for differential privacy.

%In this section, we deal with Banach norms of the form summarized in Def.~\ref{def:compbanach}.

Theorem~\ref{thm:dersenscauchy-banach} tells us how to make a function $f$ differentially private, given its derivative sensitivity $\dersens{f}$. We show that the noise that we get from this theorem is valid not only for the particular norm w.r.t. which we have computed $\dersens{f}$, but also w.r.t. any larger norm. The next theorem is proven in App.~\ref{app:proofs}.
\begin{theorem}\label{thm:diffnoisenorm}
Let $\gamma,b,\beta\in\RR_+$, $\gamma>1$. Let $\epsilon=(\gamma+1)(b+\beta)$. Let $\eta$ be a random variable distributed according to $\gencauchy{\gamma}$. Let $c$ be a $\beta$-smooth upper bound on $\dersens{f}$ w.r.t. norm $\norm{\cdot}_N$ for a function $f:X\rightarrow\RR$. Then, $g(x) : f(x)+\frac{c(x)}{b}\cdot\eta$ is $\epsilon$-differentially private w.r.t. any norm $\norm{\cdot}_M \succeq \norm{\cdot}_N$.
\end{theorem}

By Theorem~\ref{thm:diffnoisenorm}, we do not need to change anything if $N_{query} \preceq N_{db}$, as the noise would be sufficient for differential privacy w.r.t. $N_{db}$ (although scaling may still be useful to reduce the noise). However, we will need to modify the query if $N_{query} \not\preceq N_{db}$. In this section, we describe how to do it without affecting the output of the initial query.

\subsection{Adjusting norms through variable scaling}\label{sec:rescaling}
Let $N_{query}$ be a query norm, and let $N_{db}$ be a database norm, both defined over the same variables $x_1,\ldots,x_m$. As described in Sec.~{\REFBANACHSmooth}, we know how to rescale norms in such a way that $N_{query}$ will change while the query itself remains the same. The main idea behind adjusting the query is to add scalings to $N_{query}$, so that we eventually get $N_{query} \preceq N_{db}$. The outline of our methods is the following:
\begin{enumerate}
\item Find $N'_{query} \succeq N_{query}$ and $N'_{db} \preceq N_{db}$, such that there exist $\alpha_1,\ldots,\alpha_m \in \RR^{+}$ satisfying the inequality
\[\norm{\alpha_1 x_1, \ldots, \alpha_m x_m}_{N'_{query}} \leq \norm{x_1,\ldots,x_m}_{N'_{db}}\]
for all $(x_1,\ldots,x_m) \in \RR^m$. Here $x_i$ can be norm variables as well as composite norms over mutually disjoint variables.
\item Apply the scalings $\alpha_1,\ldots,\alpha_m$ to $x_1,\ldots,x_m$. By Lemma~\ref{lm:scaleineq},
\begin{eqnarray*}
\norm{\alpha_1 x_1, \ldots, \alpha_m x_m}_{N_{query}} & \leq & \norm{\alpha_1 x_1, \ldots, \alpha_m x_m}_{N'_{query}} \\
& \leq & \norm{x_1,\ldots,x_m}_{N'_{db}} \\
& \leq & \norm{x_1,\ldots,x_m}_{N_{db}} \enspace.
\end{eqnarray*}

\item Modify the query in such a way that all norms $|x_i|$ are substituted by the norm $\alpha_i \cdot |x_i|$. The norm of this new query is $\norm{\alpha_1 x_1, \ldots, \alpha_m x_m}_{N_{query}} \leq  \norm{x_1,\ldots,x_m}_{N_{db}}$, so by Theorem~\ref{thm:diffnoisenorm} we get enough noise for differential privacy, based on the derivative sensitivity of this new query.
\end{enumerate}
We will show that scaling is always possible for the {\compbanach} if $N_{db}$ contains all the variables that $N_{query}$ does. Note that, if there is a variable that $N_{db}$ does not contain, then it cannot be a sensitive variable, and hence it should be removed also from $N_{query}$, being treated as a constant.

\subsubsection{Straightforward scaling}\label{sec:method1}
First, we describe a straightforward way of finding the scalings $\alpha_1,\ldots,\alpha_m$, that always succeeds for {\compbanach}s, but is clearly not optimal. We assume that $N_{db}$ contains all the variables that $N_{query}$ does. We assume that all scalings inside the composite norm are pushed directly in front of variables, which is easy to achieve using the equality $\alpha \norm{x_1,\ldots,x_n} = \norm{\alpha x_1,\ldots,\alpha x_n}$.

We use Lemma~\ref{lm:hammer} to construct the intermediate norms $N'_{query}$ and $N'_{db}$ discussed in Sec.~\ref{sec:rescaling}. Given two norms $\norm{x_1,\ldots,x_m}_{N_{query}}$ and $\norm{x_1,\ldots,x_n}_{N_{db}}$, such that $m \leq n$ (i.e. each variable of $N_{query}$ is also present in $N_{db}$) we do the following:
\begin{enumerate}
\item Using Lemma~\ref{lm:hammer}, find $\alpha_1,\ldots,\alpha_m$ and $p$ such that \\$\norm{x_1,\ldots,x_m}_{N_{query}} \leq \norm{\alpha_1 x_1,\ldots,\alpha_m x_m}_p$.
\item Using Lemma~\ref{lm:hammer}, find $\beta_1,\ldots,\beta_n$ and $q$ such that \\$\norm{\beta_1 x_1,\ldots,\beta_n x_n}_q \leq \norm{x_1,\ldots,x_n}_{N_{db}}$.
\item Take $\gamma_i := \min{(\alpha_i, \beta_i)}$. If $p \leq q$, take $\gamma := n^{1/p-1/q}$, otherwise $\gamma := 1$. Since Lemma~\ref{lm:hammer} holds for all $x_i$, we have  $\norm{\frac{\gamma_1}{\alpha_1}x_1,\ldots,\frac{\gamma_m}{\alpha_m}x_m}_{N_{query}} \leq \norm{\gamma_1 x_1,\ldots,\gamma_m x_m}_p$.

Since $n > m$ and $\gamma_i \leq \beta_i$, applying Lemma~\ref{lm:injmap}, we get \\ $\norm{\gamma_1 x_1,\ldots,\gamma_m x_m}_q \leq \norm{\beta_1 x_1,\ldots,\beta_n x_n}_q$.

By Fact~\ref{fact:lpeqv}, we get $\gamma\cdot\norm{\gamma_1 x_1,\ldots,\gamma_m x_m}_p \leq \norm{\gamma_1 x_1,\ldots,\gamma_m x_m}_q$. We get that $\gamma\cdot\norm{\frac{\gamma_1}{\alpha_1} x_1,\ldots,\frac{\gamma_m}{\alpha_m} x_m}_{N_{query}}  \leq  \gamma\cdot\norm{\gamma_1 x_1,\ldots,\gamma_m x_m}_p \leq  \norm{\gamma_1 x_1,\ldots,\gamma_m x_m}_q \leq  \norm{\beta_1 x_1,\ldots,\beta_n x_n}_q \leq \norm{x_1,\ldots,x_n}_{N_{db}}$.
\item Add norm scalings to the initial query: scale each variable $x_i$ by $\gamma_i$, and the entire query by $\gamma$.
\end{enumerate}

\subsubsection{A more elaborated scaling}\label{sec:method2}

In Sec.~\ref{sec:method1}, we tried to fit all variables under the same $\ell_p$-norm. This may give us very rough bounds, if there is only a small subnorm that should be matched in $N_{query}$ and $N_{db}$, and the other variables do not need to be scaled at all.

In our second method, we start comparing two {\compbanach}s as terms, starting from the toplevel operation, gradually going deeper and keeping a record of additional scalings that may be necessary for some subterms to satisfy the desired inequality. Let $x$ denote atoms of the terms, and $v_i$ intermediate subterms. When comparing $N_{query}$ and $N_{db}$, we use the following inequalities:

\begin{enumerate}
\item\label{item:ineq1} $x \leq x$ for an atom $x$, all comparisons finally end up in this base case;
\item\label{item:ineq2} $v \leq \norm{v_1,\ldots,v_n}_p$ if $v \leq v_i$ for some $i \in [n]$;
\item\label{item:ineq3}$a \cdot v_a \leq b \cdot v_b$ for $a,b \in \RR$ if $a \leq b$ and $v_a \leq v_b$.
\item\label{item:ineq4} $\norm{v_1,\ldots,v_m}_p \leq \norm{w_1,\ldots,w_n}_q$ if $p \geq q$ and there is an injective mapping $f: [m] \injto [n]$ such that $|v_i| \leq |w_{f(i)}|$ for all $i \in [m]$ (here we use Fact~\ref{fact:lpeqv} and Lemma~\ref{lm:injmap}).
\end{enumerate}

\paragraph{Scaling.} The check (\ref{item:ineq3}) immediately fails if $a > b$, and (\ref{item:ineq4}) fails if $p < q$. In both cases it is easy to solve the problem by scaling. In (\ref{item:ineq3}), if $a > b$, then we scale each variable of $v_a$ by $\frac{b}{a}$, and it suffices to verify $v_a \leq v_b$. In (\ref{item:ineq4}), if $p < q$, then we use Fact~\ref{fact:lpeqv} and scale each variable of $v_1,\ldots,v_m$ by $n^{1/q-1/p}$, and it becomes sufficient to find the injective map.

\paragraph{Regrouping.} It is possible that an injective map does not exist because the term structure is too complicated. For example, it fails for the norms $N_{query} = \norm{\norm{x,y}_2, z}_1$ and $N_{db} = \norm{x, y, z}_1$, since neither $x$ nor $y$, taken alone, is at least as large as $\norm{x,y}_2$. We use Lemma~\ref{lm:ungroup} together with Fact~\ref{fact:lpsingle} to ungroup the variables and make the matching easier. In this particular example, we have $\norm{\norm{x,y}_2, z}_1 \leq \norm{\norm{x}_2,\norm{y}_2, z}_1 = \norm{x, y, z}_1$.

On the other hand, if the same variable is used multiple times in the norms, it is possible that there are several possible injective mappings. For example, if $N_{query} = \norm{\norm{x,y}_2, x}_1$ and $N_{db} = \norm{\norm{x,y}_2, \norm{x,y}_3}_1$, then the matching $x \leq \norm{x,y}_3$ and $\norm{x,y}_2 \leq \norm{x,y}_2$ requires no additional scaling, but $x \leq \norm{x,y}_2$ and $\norm{x,y}_2 \leq \norm{x,y}_3$ is possible only if we scale $x$ and $y$ with $n^{1/3 - 1/2}$. We treat an injective mapping as a weighted matching in a bipartite graph, where the terms are vertices, and the $\leq$ relation forms edges. If any of $\leq$ relations requires additional scaling, we assign to it a weight defined by the scaling. We then look for a minimum-weight perfect matching.

\paragraph{The case of failure.} If the subnorms are too complicated, then ungrouping the variables may still fail to detect similar subterms, and the matching fails. If it happens, then we apply the solution of Sec.~\ref{sec:method2}.

%\subsection{Adding logarithm norms}
%According to Sec~{\REFBANACHSmooth}, there are functions for which we are only able to compute sensitivity w.r.t. \emph{logarithm norm}, which is $\norm{x} := |\ln(x)|$. In Sec~{\REFBANACHSmooth} this norm is defined only for norm variables, and not for subnorms. Hence, we may treat $\ln(x)$ itself as a new variable, and we do not decompose it. The only additional inequality that we use is $x \leq \ln(x)$, and it is impossible to achieve $\ln(x) \leq x$ using scaling, so $\ln(x) \leq x$ is treated similarly to $y \leq x$ for a different variable $y \neq x$. Another difference is that, if two different subterms use the variables $x$ and $\ln(x)$ respectively, we do not claim that their variables are mutually disjoint.

\section{How to choose the norm and the $\varepsilon$}

In the standard definition of DP, we want that the output would be (sufficiently) indistinguishable if we add/remove \emph{one} row to the table. In the new settings, the notion of \emph{unit} can be different. If we scale the norm by $a$ and keep noise level the same, the $\epsilon$ will increase proportionally to $a$. Since there is no standard definition of ``good'' $\varepsilon$, it may be unclear which norm is reasonable. For this, we need to understand what the table attributes mean. For example, if the ship location is presented in meters, and we want to conceal a change in a kilomter, we need to scale the location norm by $0.001$ to capture a larger change.

To give a better interpretation to $\varepsilon$, we may relate it to other security definitions such as guessing probability advantage. Let $X' \subseteq X$ be the subset of inputs for which we consider the attacker guess as ``correct'' (e.g. he guesses ship location precisely enough). Let the posterior belief of the adversary be expressed by the probability distribution $\prpost{\cdot}$. Let the initial distribution of $X$ be $\prpre{\cdot}$, and $\fxpre$ the corresponding probability density function, i.e. $\prpre{X'} = \int_{X'} \fxpre(x) dx$ for $X' \subseteq X$. We need an upper bound on $\prpost{X'}$.

The quantity $\prpost{X'}$ characterizes the knowledge of attacker after seeing the output of a differential privacy mechanism $\noised{q}$. Let $\rndvar{X}$ be the random variable for $x \in  X$. We define $\prpost{X'} := \pr{\rndvar{X} \in X'\ |\ \noised{q}(\rndvar{X}) = y}$, where $y$ is the output that the attacker observed. Since the probability weight of a single element can be $0$ in the continuous case, we actually write $\prpost{X'} := \pr{\rndvar{X} \in X'\ |\ \noised{q}(\rndvar{X}) \in Y}$, where $Y \subseteq \noised{q}(X)$ is any subset of outputs to which $y$ may belong.

\begin{eqnarray*}
\prpost{X'} & = & \pr{\rndvar{X} \in X'\ |\ \noised{q}(\rndvar{X}) \in Y}\\
& = & \frac{\pr{\rndvar{X} \in X'\ \wedge \noised{q}(\rndvar{X}) \in Y}}{\pr{\noised{q}(\rndvar{X}) \in Y}}\\
& = & \frac{\int_{X'} \pr{\noised{q}(x') \in Y} \fxpre(x') dx'}{\pr{\noised{q}(\rndvar{X}) \in Y}}\\
& = & \frac{\int_{X'} \pr{\noised{q}(x') \in Y} \fxpre(x') dx'}{\int_{X} \pr{\noised{q}(x) \in Y} \fxpre(x) dx}\\
& = & \frac{\int_{X'} \pr{\noised{q}(x') \in Y} \fxpre(x') dx'}{\int_{X'} \pr{\noised{q}(x) \in Y} \fxpre(x) dx + \int_{\bar{X'}} \pr{\noised{q}(x) \in Y} \fxpre(x) dx}\\
& = & \frac{1}{1 + \frac{\int_{\bar{X'}} \pr{\noised{q}(x) \in Y} \fxpre(x) dx}{\int_{X'} \pr{\noised{q}(x') \in Y} \fxpre(x') dx'}}\\
& \leq & \frac{1}{1 + \frac{\int_{X''} \pr{\noised{q}(x) \in Y} \fxpre(x) dx}{\int_{X'} \pr{\noised{q}(x') \in Y} \fxpre(x') dx'}} (\forall X'' \subseteq \bar{X'})\enspace.
\end{eqnarray*}

Let $X'' := \set{x\ | \ d(x,x') \leq a}$ for any $a \in \RR$ (e.g. $a := \max_{x \in X}{d(x,x')}$). Differential privacy gives us $\frac{\pr{\noised{q}(x') \in Y}}{\pr{\noised{q}(x) \in Y}} \leq e^{\varepsilon\cdot a}$ for all $x \in X''$. We get
\[\frac{1}{1 + \frac{\int_{X''} \pr{\noised{q}(x) \in Y} \fxpre(x) dx}{\int_{X'} \pr{\noised{q}(x') \in Y} \fxpre(x') dx'}} \leq \frac{1}{1 + \frac{e^{-\varepsilon a}\int_{X''}\fxpre(x) dx}{\int_{X'} \fxpre(x') dx'}} = \frac{1}{1 + e^{-\varepsilon a}\cdot\frac{1 - \prpre{X''}}{\prpre{X'}}}\enspace.\]

The optimal value of $a$ depends on the distribution of $\rndvar{X}$. We defer this research to future work.

Applying this approach to standard DP would give us the probaility of guessing that the record is present in the table. Applying it to component privacy proposed in this work, we could answer questions like ``how likely the attacker guesses that the location of some ship is within 5 miles from the actual location''.

\section{Implementation and Evaluation}\label{sec:eval}

\subsection{Implementation}

Our analyser (available on GitHub) has been implemented in Haskell. As an input, it takes an SQL query, a database schema, and a description of the norm w.r.t. which the data owner wants to achieve differential privacy. It returns \emph{another query} (as a string) that describes the way in which derivative sensitivity should be computed from a particular database. This new query represents the function $c(x)$ such that the additive noise would be $\frac{c(x)}{b} \cdot \eta$ for $\eta \gets \gencauchy{\gamma}$, according to Theorem~\ref{thm:dersenscauchy-banach}. In our analyser, $\gamma = 4$ is fixed, and $b = \varepsilon / (\gamma + 1) - \beta$, where $\beta$ is a parameter given as an input to the analyser (by default, $\beta = 0.1$), and $\varepsilon$ is the desired differential privacy level. The resulting query can be fed to a database engine to evaluate the sensitivity on particular data.

The data owner decides which rows of the database are sensitive an which are not. We assume that each table contains a row $ID$ of unique keys. For each table $X$, we expect a table named $X\_sensRows$ that contains the same column $ID$ of keys, and another column $sensitive$ of boolean values that tell for each row whether it is sensitive or not.

%Meie tuum + RDBMS. Haskellis kirjutatud tuumast kutsutakse DBMS-i v\"alja. Iga tabeli jaoks on teine tabel, kus on kirjas, millised esimese tabeli read on tundlikud. Norm ise l\"aheb meie tuumale eraldi argumendiks.

\subsection{Evaluation}

We performed evaluation on 4 x Intel(R) Core(TM) i5-6300U CPU @ 2.40GHz laptop, Ubuntu 16.04.4 LTS, using PostgreSQL 9.5.12.

We have taken the queries of TCP-H set~\cite{tpc-h} for benchmarking. Most of these queries contain GROUP BY constructions and nested aggregating queries that our analyser does not support. Hence, we had to modify these queries. Instead of GROUP BY, we have added filtering that chooses one particular group. Theoretically, it could be possible to compute the entire GROUP BY query this way, substituting it with a number of queries, each aggregating one particular group. However, it would be reasonable only as far as the number of possible groups is small.

The largest challenge was coming from the filters. The queries that contained many filterings by sensitive values required increase in $\beta$ and hence in $\varepsilon$. We had to manually rewrite the filters in such a way that public filters would be easily extractable. For example, if we left the filters as they are or even converted them to CNF form, we may get combinations of public and private filters related by an OR, so that we can no longer push the public part of the filter into WHERE clause. Knowing that some OR filters were obviously mutually exclusive, we replaced (on the analyser level) OR with XOR, which improved sensitivity since instead of $x_1\vee x_2 = x_1 + x_2 - x_1 \cdot x_2$ we got $x_1 + x_2$.

By default, we took $\varepsilon = 1.0$, and $\beta = 0.1$. This choice gives $b = 0.1$, and the additive noise with $78\%$ probability is below $10 \cdot c(x)$, where the value $78\%$ comes from analysing distribution $\gencauchy{4.0}$. Too large value of $\beta$ makes $b$ (and hence the noise) larger, and too small $\beta$ makes the sensitivity larger, so $\beta$ is a parameter that could be optimized. For some queries, it was not possible to achieve smoothness $0.1$. In those cases, we took as small $\beta$ as possible and increased $\varepsilon$ so that we would have $b = 0.1$ to make the sensitivities easier to compare.

We have run our analyser on the TCP-H dataset with scale factors $0.1$, $0.5$, $1.0$, denoting how much data is generated for the sample database. For $1.0$, the size of the largest table is ca $6$ million rows. The table schema, together with numbers of rows for different tables, in given in App.~\ref{app:eval:schema}. We have considered integer, decimal, and date columns of these tables as sensitive, assigning to them different weights, described more precisely in App.~\ref{app:eval:norms}. Sensitive entries have been combined using an $\ell_1$-norm.

We adjusted (as described above) the queries $Q1$ (splitting it to $5$ queries), $Q3 - Q7$, $Q9 - Q10$, $Q12$ (splitting it to $2$ queries), $Q16$, $Q17$, $Q19$ of the TCP-H dataset to our analyser. The queries on which we actually ran the analysis can be found in App.~\ref{app:eval:queries}. The time benchmarks can be found in Table~\ref{tab:benchtime}, and precision in Table~\ref{tab:benchprec}, where $K$ denotes $\cdot 10^3$, $M$ denotes $\cdot 10^6$, and $G$ denotes $\cdot 10^9$. The time spent to generate the query that computes sensitivity is negligible (below $20ms$), as it does not depend on the number of rows, so we only report the time spent to execute that query. In addition to the initial query and the sensitivity, we have also measured time and output the \emph{modified} query, where private filtering is replaced with continuous approximation (sigmoids). The reason is that the sensitivity has been computed for the modified query, and it does not guarantee differential privacy if we add noise to the output of the initial query. The error has been computed as the ratio between the actual result and the difference between the actual and the private results, which is $\frac{\abs{(mod.res + 10\cdot sens) - init.res}}{init.res}\cdot 100$.

The worst precision is achieved for the queries that use several filters over sensitive values. While $b4$ and $b5$ are similar queries, in $b4$ the filtering aims to capture a $3$-unit span, while it is a $12$-unit span in $b5$, and sigmoids better distinguish values that are further apart.
When the two corresponding sigmoids are multiplied, the function gets too low-scaled in the short span. The result seems especially bad for $b10$, where no rows satisfy the filtering in the real query and the correct answer should be $0.0$, but the small errors of single rows aggregate too much for large tables. The error would not seem so large if compared to the result that we would get if all rows actually satisfied the filter, and for sparse filtering the relative error will be larger than for dense filtering.

The time overhead is also larger for queries with multiple sensitive filters, not only because the query itself becomes longer and requires computation of exponent functions, but mainly since the sensitive filters are no longer a part of WHERE clause, and the modified query function needs to be applied to more rows.

\begin{table*}
\begin{center}
\begin{tabular}{|l |r |r |r |r |r |r |r |r |r |}
\hline & \multicolumn{3}{c|}{SF = 0.1} & \multicolumn{3}{c|}{SF = 0.5} & \multicolumn{3}{c|}{SF = 1.0} \\ 
\hline
 & init.query & mod.query & sens. & init.query & mod.query & sens. & init.query & mod.query & sens. \\ 
\hline
\hline
$\mathtt{b1\_1}$ & $205.61$ & $5.83K$ & $15.97K$ & $919.07$ & $29.49K$ & $80.69K$ & $1.76K$ & $57.13K$ & $155.34K$\\ 
$\mathtt{b1\_2}$ & $180.87$ & $5.72K$ & $16.12K$ & $883.25$ & $29.47K$ & $87.0K$ & $1.78K$ & $56.7K$ & $155.42K$\\ 
$\mathtt{b1\_3}$ & $191.4$ & $5.8K$ & $16.54K$ & $1.07K$ & $29.08K$ & $83.59K$ & $1.88K$ & $55.67K$ & $158.74K$\\ 
$\mathtt{b1\_4}$ & $196.88$ & $5.91K$ & $17.25K$ & $1.02K$ & $28.55K$ & $81.27K$ & $1.91K$ & $56.64K$ & $159.98K$\\ 
$\mathtt{b1\_5}$ & $177.97$ & $5.75K$ & $5.82K$ & $1.0K$ & $28.91K$ & $29.62K$ & $1.72K$ & $55.23K$ & $57.13K$\\ 
$\mathtt{b3}$ & $176.42$ & $137.24$ & $473.92$ & $1.03K$ & $722.25$ & $2.44K$ & $337.31$ & $486.79$ & $1.11K$\\ 
$\mathtt{b4}$ & $239.58$ & $9.82K$ & $33.35K$ & $1.18K$ & $46.5K$ & $168.3K$ & $2.16K$ & $92.54K$ & $332.54K$\\ 
$\mathtt{b5}$ & $139.85$ & $304.15$ & $2.55K$ & $1.02K$ & $1.06K$ & $4.73K$ & $1.48K$ & $2.16K$ & $10.65K$\\ 
$\mathtt{b6}$ & $156.01$ & $38.64K$ & $199.75K$ & $795.69$ & $188.32K$ & $946.76K$ & $1.5K$ & $375.49K$ & $1.86M$\\ 
$\mathtt{b7}$ & $190.33$ & $325.91$ & $1.23K$ & $999.83$ & $1.4K$ & $6.01K$ & $2.0K$ & $3.07K$ & $12.68K$\\ 
$\mathtt{b9}$ & $190.48$ & $169.74$ & $6.27K$ & $856.59$ & $800.38$ & $4.81K$ & $1.76K$ & $1.67K$ & $10.25K$\\ 
$\mathtt{b10}$ & $168.22$ & $174.9$ & $588.18$ & $943.38$ & $917.2$ & $3.03K$ & $319.81$ & $1.78K$ & $5.86K$\\ 
$\mathtt{b12\_1}$ & $250.2$ & $16.92K$ & $79.13K$ & $1.39K$ & $87.39K$ & $391.23K$ & $2.79K$ & $163.11K$ & $787.36K$\\ 
$\mathtt{b12\_2}$ & $232.81$ & $7.03K$ & $32.49K$ & $1.25K$ & $36.13K$ & $163.06K$ & $2.39K$ & $67.93K$ & $312.06K$\\ 
$\mathtt{b16}$ & $19.02$ & $214.31$ & $476.22$ & $101.77$ & $1.22K$ & $2.8K$ & $210.06$ & $2.27K$ & $5.96K$\\ 
$\mathtt{b17}$ & $120.45$ & $111.22$ & $312.88$ & $722.6$ & $585.05$ & $1.69K$ & $1.18K$ & $1.17K$ & $3.24K$\\ 
$\mathtt{b19}$ & $165.74$ & $332.22$ & $1.62K$ & $933.98$ & $1.76K$ & $8.47K$ & $1.74K$ & $3.6K$ & $16.02K$\\ 
\hline
\end{tabular}
\end{center}
\caption{Time benchmarks (ms)}\label{tab:benchtime}
\end{table*}

\begin{table*}
\begin{center}
\begin{tabular}{|l |l |r |r |r |r |r |r |r |r |r |r |r |r |}
\hline & & \multicolumn{4}{c|}{SF = 0.1} & \multicolumn{4}{c|}{SF = 0.5} & \multicolumn{4}{c|}{SF = 1.0} \\ 
\hline
 & $\varepsilon$ & init.res & mod.res & sens. & \% error & init.res & mod.res & sens. & \% error & init.res & mod.res & sens. & \% error \\ 
\hline
\hline
$\mathtt{b1\_1}$ & 1.0 & $3.79M$ & $3.55M$ & $1.0$ & $6.18$ & $18.87M$ & $17.7M$ & $1.0$ & $6.2$ & $37.72M$ & $35.38M$ & $1.0$ & $6.2$\\ 
$\mathtt{b1\_2}$ & 1.0 & $5.34G$ & $5.01G$ & $9.96K$ & $6.18$ & $27.35G$ & $25.65G$ & $9.96K$ & $6.2$ & $56.57G$ & $53.06G$ & $9.96K$ & $6.2$\\ 
$\mathtt{b1\_3}$ & 1.0 & $5.07G$ & $4.76G$ & $11.16K$ & $6.18$ & $25.98G$ & $24.37G$ & $11.16K$ & $6.2$ & $53.74G$ & $50.41G$ & $11.16K$ & $6.2$\\ 
$\mathtt{b1\_4}$ & 1.0 & $5.27G$ & $4.95G$ & $12.06K$ & $6.18$ & $27.02G$ & $25.34G$ & $12.06K$ & $6.2$ & $55.89G$ & $52.43G$ & $12.06K$ & $6.2$\\ 
$\mathtt{b1\_5}$ & 1.0 & $148.3K$ & $139.12K$ & $0.0006$ & $6.19$ & $739.56K$ & $693.7K$ & $0.0006$ & $6.2$ & $1.48M$ & $1.39M$ & $0.0006$ & $6.2$\\ 
$\mathtt{b3}$ & 1.0 & $3.62K$ & $1.09K$ & $7.98K$ & $2.13K$ & $3.21K$ & $963.25$ & $7.98K$ & $2.42K$ & $0.0$ & $0.0$ & $0.0$ & 0.0\\ 
$\mathtt{b4}$ & 1.0 & $2.92K$ & $8.58K$ & $0.0068$ & $194.14$ & $14.17K$ & $42.89K$ & $0.0069$ & $202.66$ & $28.07K$ & $85.67K$ & $0.0069$ & $205.18$\\ 
$\mathtt{b5}$ & 1.0 & $5.43M$ & $5.05M$ & $4.96K$ & $5.98$ & $25.28M$ & $23.82M$ & $4.96K$ & $5.56$ & $47.56M$ & $45.33M$ & $4.96K$ & $4.6$\\ 
$\mathtt{b6}$ & 8.5 & $17.45M$ & $17.09M$ & $192.5K$ & $9.03$ & $88.13M$ & $87.69M$ & $196.87K$ & $1.74$ & $181.93M$ & $181.41M$ & $201.81K$ & $0.82$\\ 
$\mathtt{b7}$ & 1.0 & $22.07M$ & $22.12M$ & $22.25K$ & $1.24$ & $95.63M$ & $100.99M$ & $22.51K$ & $5.85$ & $212.11M$ & $219.31M$ & $22.58K$ & $3.5$\\ 
$\mathtt{b9}$ & 1.0 & $30.32M$ & $30.32M$ & $40.0K$ & $1.32$ & $137.73M$ & $137.73M$ & $49.2K$ & $0.36$ & $283.82M$ & $283.82M$ & $49.2K$ & $0.17$\\ 
$\mathtt{b10}$ & 1.0 & $100.31K$ & $111.39K$ & $3.42K$ & $45.15$ & $149.6K$ & $175.47K$ & $3.42K$ & $40.18$ & $0.0$ & $93.54K$ & $3.31K$ & $\infty$\\ 
$\mathtt{b12\_1}$ & 1.0 & $3.12K$ & $9.05K$ & $0.0041$ & $190.47$ & $15.41K$ & $45.15K$ & $0.0041$ & $193.04$ & $30.84K$ & $90.3K$ & $0.0041$ & $192.83$\\ 
$\mathtt{b12\_2}$ & 1.0 & $1.29K$ & $3.65K$ & $0.0041$ & $183.08$ & $6.2K$ & $18.11K$ & $0.0041$ & $191.85$ & $12.37K$ & $36.24K$ & $0.0041$ & $193.01$\\ 
$\mathtt{b16}$ & 4.5 & $4.94K$ & $249.44K$ & $16.69$ & $4.95K$ & $24.51K$ & $1.25M$ & $16.69$ & $5.01K$ & $49.01K$ & $2.51M$ & $16.69$ & $5.02K$\\ 
$\mathtt{b17}$ & 1.0 & $31.54K$ & $265.49K$ & $902.59$ & $770.26$ & $256.24K$ & $1.61M$ & $902.59$ & $531.72$ & $531.93K$ & $3.53M$ & $902.59$ & $565.73$\\ 
$\mathtt{b19}$ & 7.0 & $155.25K$ & $207.81K$ & $27.0K$ & $207.74$ & $1.1M$ & $1.17M$ & $28.23K$ & $31.52$ & $1.73M$ & $2.3M$ & $29.8K$ & $50.84$\\ 
\hline
\end{tabular}
\end{center}
\caption{Precision benchmarks}\label{tab:benchprec}
\end{table*}

\section{Limitations and Future Work}

Our analyser supports queries only of certain form. In particular, it does not support DISTINCT queries. The problems is that, to use derivative sensitivity, we need to use continuous functions in the queries. We do not know how to approximate efficiently a function that removes repeating elements from a list of arguments. Support of GROUP BY queries is limited, allowing grouping only by a non-sensitive attribute, or an attribute whose range is known to be small, since we only protect real-valued outputs, and not the number of outputs.

%The sensitivity is defined over real numbers. While boolean values can be treated as real numbers $0.0$ and $1.0$, and strings can be hashed, it does not make much sense to define distance between real-valued hashes, so we not allow non-real values to be sensitive. It would be interesting to extend the analyser to support this.

%In our current implementation, if no sufficiently smooth bounds are obtained, we get an increase in $\varepsilon$. Instead, we could search for a less smooth bound that increases the noise (probably significantly), but keeps $\varepsilon$ on the same level. This would also make results of different queries more easily comparable.

%Another important issue, seen from the evaluation results, is that using numerous filters over sensitive values makes the error extremely high. The problem is that filtering is indeed quite sensitive operation, and if some entry changes by a very small quantity, it may pass the filter bound and modify the final output by a very large quantity. Even using $\ell_1$ norm allows the filtered value to change a bit in \emph{each row}, which in the worst case may move \emph{all} rows to the other side of the choice. We should probably define some other norms for values that undergo filtering, e.g. trying to capture the change in the filter decision, and not the value itself. 

From the evaluation results, we see that that using numerous filters over sensitive values makes the error high. Clearly, adding noise \emph{before} filtering is the path towards smaller errors. While we believe that the derivative sensitivity framework allows us to locate the points where the noise has to be added, and to determine its magnitude, this is not the topic of the current paper.

\section{Conclusion}

%\todo[inline]{Peeter writes what he thinks about our results.}

We have started the study of complete norms to define the quantitative privacy properties of information release mechanisms, and have discovered their high expressivity for different kinds of numeric inputs, as well as the principles of the parallel composition of norms in a manner that allows the sensitivity of the information release mechanism to be found. Our results show how the similarity of local sensitivity and the derivative can be exploited in constructing differentially private mechanisms. The result is also practically significant because of the need to precisely model the privacy requirements of data owner(s), for which the flexibility of specifying the metric over possible inputs is a must.

Our results open up the study of the combinations of metrics over more varied types of input data, including categorical data, structured data or data with consistency constraints. Such study would also look for possibilities to express the constraints through suitable combinations of metrics over components. We note that inputs with constraints~\cite{DBLP:journals/tods/KiferM14} or with particular structure (sequences indexed by time points)~\cite{DBLP:conf/stoc/DworkNPR10,DBLP:conf/ccs/ChenMHM17} have been considered in the literature. Our hope is that it is possible to find suitable complete norms that define the metrics that have been used in the privacy definitions for such data, and thereby express these constructions inside our framework.

%\todo[inline]{Filters. There are inconvenient trade-offs. ``Flat'' sigmoids differ from sharp step functions of the filter, but ``steep'' sigmoids are more sensitive. The ``correct'' way is to rewrite queries (to recognize that they are less sensitive) or put in noise not at the end, but that is a topic for another paper.}

\bibliographystyle{plain}
\bibliography{references}

\begin{thebibliography}{10}

\bibitem{tpc-h}
{TPC} {BENCHMARK\textsuperscript{TM}} {H}, revision 2.17.3.
\newblock Transaction Processing Performance Council, 2017.
\newblock
  \url{http://www.tpc.org/TPC{\_}Documents{\_}Current{\_}Versions/pdf/tpc-h{\_}v2.17.3.pdf}.

\bibitem{DBLP:conf/icalp/ArapinisFG16}
Myrto Arapinis, Diego Figueira, and Marco Gaboardi.
\newblock Sensitivity of counting queries.
\newblock In Ioannis Chatzigiannakis, Michael Mitzenmacher, Yuval Rabani, and
  Davide Sangiorgi, editors, {\em 43rd International Colloquium on Automata,
  Languages, and Programming, {ICALP} 2016, July 11-15, 2016, Rome, Italy},
  volume~55 of {\em LIPIcs}, pages 120:1--120:13. Schloss Dagstuhl -
  Leibniz-Zentrum fuer Informatik, 2016.

\bibitem{awan2018structure}
Jordan Awan and Aleksandra Slavkovic.
\newblock Structure and sensitivity in differential privacy: Comparing k-norm
  mechanisms.
\newblock {\em arXiv preprint arXiv:1801.09236}, 2018.

\bibitem{baggett1991functional}
L.W. Baggett.
\newblock {\em Functional Analysis: A Primer}.
\newblock Chapman \& Hall Pure and Applied Mathematics. Taylor \& Francis,
  1991.

\bibitem{DBLP:conf/pet/ChatzikokolakisABP13}
Konstantinos Chatzikokolakis, Miguel~E. Andr{\'{e}}s, Nicol{\'{a}}s~Emilio
  Bordenabe, and Catuscia Palamidessi.
\newblock Broadening the scope of differential privacy using metrics.
\newblock In Emiliano~De Cristofaro and Matthew Wright, editors, {\em Privacy
  Enhancing Technologies - 13th International Symposium, {PETS} 2013,
  Bloomington, IN, USA, July 10-12, 2013. Proceedings}, volume 7981 of {\em
  Lecture Notes in Computer Science}, pages 82--102. Springer, 2013.

\bibitem{DBLP:conf/icdcit/ChatzikokolakisPS15}
Konstantinos Chatzikokolakis, Catuscia Palamidessi, and Marco Stronati.
\newblock Geo-indistinguishability: {A} principled approach to location
  privacy.
\newblock In Raja Natarajan, Gautam Barua, and Manas~Ranjan Patra, editors,
  {\em Distributed Computing and Internet Technology - 11th International
  Conference, {ICDCIT} 2015, Bhubaneswar, India, February 5-8, 2015.
  Proceedings}, volume 8956 of {\em Lecture Notes in Computer Science}, pages
  49--72. Springer, 2015.

\bibitem{DBLP:conf/ccs/ChenMHM17}
Yan Chen, Ashwin Machanavajjhala, Michael Hay, and Gerome Miklau.
\newblock Pegasus: Data-adaptive differentially private stream processing.
\newblock In Bhavani~M. Thuraisingham, David Evans, Tal Malkin, and Dongyan Xu,
  editors, {\em Proceedings of the 2017 {ACM} {SIGSAC} Conference on Computer
  and Communications Security, {CCS} 2017, Dallas, TX, USA, October 30 -
  November 03, 2017}, pages 1375--1388. {ACM}, 2017.

\bibitem{DBLP:conf/icalp/Dwork06}
Cynthia Dwork.
\newblock Differential privacy.
\newblock In Michele Bugliesi, Bart Preneel, Vladimiro Sassone, and Ingo
  Wegener, editors, {\em Automata, Languages and Programming, 33rd
  International Colloquium, {ICALP} 2006, Venice, Italy, July 10-14, 2006,
  Proceedings, Part {II}}, volume 4052 of {\em Lecture Notes in Computer
  Science}, pages 1--12. Springer, 2006.

\bibitem{DBLP:conf/stoc/DworkNPR10}
Cynthia Dwork, Moni Naor, Toniann Pitassi, and Guy~N. Rothblum.
\newblock Differential privacy under continual observation.
\newblock In Schulman \cite{DBLP:conf/stoc/2010}, pages 715--724.

\bibitem{FeatherweightPINQ}
Hamid Ebadi and David Sands.
\newblock Featherweight {PINQ}.
\newblock {\em Journal of Privacy and Security}, 7(2):159--184, 2016.

\bibitem{DBLP:conf/popl/EbadiSS15}
Hamid Ebadi, David Sands, and Gerardo Schneider.
\newblock Differential privacy: Now it's getting personal.
\newblock In Sriram~K. Rajamani and David Walker, editors, {\em Proceedings of
  the 42nd Annual {ACM} {SIGPLAN-SIGACT} Symposium on Principles of Programming
  Languages, {POPL} 2015, Mumbai, India, January 15-17, 2015}, pages 69--81.
  {ACM}, 2015.

\bibitem{DBLP:conf/birthday/ElSalamounyCP14}
Ehab ElSalamouny, Konstantinos Chatzikokolakis, and Catuscia Palamidessi.
\newblock Generalized differential privacy: Regions of priors that admit robust
  optimal mechanisms.
\newblock In Franck van Breugel, Elham Kashefi, Catuscia Palamidessi, and Jan
  Rutten, editors, {\em Horizons of the Mind. {A} Tribute to Prakash Panangaden
  - Essays Dedicated to Prakash Panangaden on the Occasion of His 60th
  Birthday}, volume 8464 of {\em Lecture Notes in Computer Science}, pages
  292--318. Springer, 2014.

\bibitem{DBLP:conf/stoc/HardtT10}
Moritz Hardt and Kunal Talwar.
\newblock On the geometry of differential privacy.
\newblock In Schulman \cite{DBLP:conf/stoc/2010}, pages 705--714.

\bibitem{DBLP:conf/sigmod/HayMMCZ16}
Michael Hay, Ashwin Machanavajjhala, Gerome Miklau, Yan Chen, and Dan Zhang.
\newblock Principled evaluation of differentially private algorithms using
  dpbench.
\newblock In Fatma {\"{O}}zcan, Georgia Koutrika, and Sam Madden, editors, {\em
  Proceedings of the 2016 International Conference on Management of Data,
  {SIGMOD} Conference 2016, San Francisco, CA, USA, June 26 - July 01, 2016},
  pages 139--154. {ACM}, 2016.

\bibitem{ElasticSensitivity}
Noah~M. Johnson, Joseph~P. Near, and Dawn Song.
\newblock Towards practical differential privacy for {SQL} queries.
\newblock {\em {Proceedings of the VLDB Endowment}}, 11(5):526--539, 2018.

\bibitem{10.1007/978-3-642-36594-2_26}
Shiva~Prasad Kasiviswanathan, Kobbi Nissim, Sofya Raskhodnikova, and Adam
  Smith.
\newblock Analyzing graphs with node differential privacy.
\newblock In Amit Sahai, editor, {\em Theory of Cryptography}, pages 457--476,
  Berlin, Heidelberg, 2013. Springer Berlin Heidelberg.

\bibitem{DBLP:journals/tods/KiferM14}
Daniel Kifer and Ashwin Machanavajjhala.
\newblock Pufferfish: {A} framework for mathematical privacy definitions.
\newblock {\em {ACM} Trans. Database Syst.}, 39(1):3:1--3:36, 2014.

\bibitem{DBLP:conf/sigmod/McSherry09}
Frank McSherry.
\newblock Privacy integrated queries: an extensible platform for
  privacy-preserving data analysis.
\newblock In Ugur {\c{C}}etintemel, Stanley~B. Zdonik, Donald Kossmann, and
  Nesime Tatbul, editors, {\em Proceedings of the {ACM} {SIGMOD} International
  Conference on Management of Data, {SIGMOD} 2009, Providence, Rhode Island,
  USA, June 29 - July 2, 2009}, pages 19--30. {ACM}, 2009.

\bibitem{DBLP:conf/stoc/NissimRS07}
Kobbi Nissim, Sofya Raskhodnikova, and Adam~D. Smith.
\newblock Smooth sensitivity and sampling in private data analysis.
\newblock In David~S. Johnson and Uriel Feige, editors, {\em Proceedings of the
  39th Annual {ACM} Symposium on Theory of Computing, San Diego, California,
  USA, June 11-13, 2007}, pages 75--84. {ACM}, 2007.

\bibitem{DBLP:journals/corr/abs-1207-0872}
Catuscia Palamidessi and Marco Stronati.
\newblock Differential privacy for relational algebra: Improving the
  sensitivity bounds via constraint systems.
\newblock In Herbert Wiklicky and Mieke Massink, editors, {\em Proceedings 10th
  Workshop on Quantitative Aspects of Programming Languages and Systems, {QAPL}
  2012, Tallinn, Estonia, 31 March and 1 April 2012.}, volume~85 of {\em
  {EPTCS}}, pages 92--105, 2012.

\bibitem{DBLP:conf/stoc/2010}
Leonard~J. Schulman, editor.
\newblock {\em Proceedings of the 42nd {ACM} Symposium on Theory of Computing,
  {STOC} 2010, Cambridge, Massachusetts, USA, 5-8 June 2010}. {ACM}, 2010.

\bibitem{DBLP:conf/sigmod/0001LKCJN17}
Xi~Wu, Fengan Li, Arun Kumar, Kamalika Chaudhuri, Somesh Jha, and Jeffrey~F.
  Naughton.
\newblock Bolt-on differential privacy for scalable stochastic gradient
  descent-based analytics.
\newblock In Semih Salihoglu, Wenchao Zhou, Rada Chirkova, Jun Yang, and Dan
  Suciu, editors, {\em Proceedings of the 2017 {ACM} International Conference
  on Management of Data, {SIGMOD} Conference 2017, Chicago, IL, USA, May 14-19,
  2017}, pages 1307--1322. {ACM}, 2017.

\end{thebibliography}

\appendix

\section{Example}\label{app:theoryexample}

\subsection{Computing derivative sensitivity w.r.t. different components}\label{app:theoryexample1}
Consider the example of computing differentially privately the time that a ship takes to reach the port. This time can be
expressed as
\[ f(x,y,v) = \frac{\sqrt{x^2+y^2}}{v} \]
\[ f : \RR^3 \rightarrow \RR \]
where $(x,y)$ are the coordinates of the ship (with the port at $(0,0)$) and $v$ is the speed of the ship.
%Let $X$ be the set of all possible inputs $(x,y,v)$ and $Y$ the set of all possible outputs.
For differential privacy, we need to define distances on $\RR^3$ and $\RR$. Because we want to use Fr\'echet derivatives to
compute the sensitivities, we instead define norms, which then also induce distances.
For $\RR$ it is natural to use the absolute value norm but for $\RR^3$ there are more choices.

Consider the $\ell_1$-distance that we used previously. Then moving the ship by geographical distance $s$ in a direction
parallel or perpendicular to its velocity, changes the whole input by $\ell_1$-distance $s$. But moving the ship in any
other direction changes the whole input by $\ell_1$-distance more than $s$. This is unnatural. We would like the change
not to depend on the direction.

Consider the $\ell_2$-distance. Then moving the ship by geographical distance $s$ always changes the whole input by
$\ell_2$-distance $s$, regardless of direction. If we change the speed of the ship by $u$ (either up or down) then the
whole input changes by $\ell_2$-distance $u$. If we simultaneously move the ship by distance $s$ and change its speed by
$u$ however, and numerically $s = u$ (ignoring the units) then the whole input changes by $\ell_2$-distance
$\sqrt{2}\cdot s$ instead of the more natural $2s$.

Thus it is better to combine $\ell_1$- and $\ell_2$-distances. We first combine the change in the coordinates of the ship using $\ell_1$-norm,
then combine the result with the change in speed using $\ell_2$-norm. Thus
\[ \|(\Delta x, \Delta y, \Delta v)\| = \|(\|(\Delta x, \Delta y)\|_2, \Delta v)\|_1 = \sqrt{(\Delta x)^2 + (\Delta y)^2} + |\Delta v| \]
This still has a problem. Suppose that $x = y = 2000\ \mathrm{km}, v = 20\ \mathrm{km/h}$.
Then changing $v$ to $10\ \mathrm{km/h}$ changes the whole input by the same distance as changing $y$ to
$2010\ \mathrm{km}$. But the former change seems much more important in most cases than the latter.
Thus we scale the geographical change and the speed change by different constants before combining them:
\[ \|(\Delta x, \Delta y, \Delta v)\| = a \sqrt{(\Delta x)^2 + (\Delta y)^2} + b|\Delta v| = \|(\|(a\Delta x, a\Delta y)\|_2, b\Delta v)\|_1 \]
This norm should now be good enough.

Let us now compute the derivatives. We first compute ordinary partial derivatives, i.e. with respect to the absolute
value norm.
\[ \frac{\partial f}{\partial x} = \frac{2x}{2v \sqrt{x^2+y^2}} = \frac{x}{v \sqrt{x^2+y^2}} \]
\[ \frac{\partial f}{\partial y} = \frac{2y}{2v \sqrt{x^2+y^2}} = \frac{y}{v \sqrt{x^2+y^2}} \]
\[ \frac{\partial f}{\partial v} = -\frac{\sqrt{x^2+y^2}}{v^2} \]
Now we compute the derivatives w.r.t. the scaled one-dimensional norms, i.e. $x$ is considered not as an element of
the Banach space $(\RR, |\ |)$ but as an element of the Banach space $(\RR, \|\ \|_x)$ where $\|\Delta x\|_x = a|\Delta x|$.
The operator norm of the Fr\'echet derivative in $(\RR, \|\ \|_x)$ (which is also the derivative sensitivity) is then
\[ \dersens{f^{y,v}}(x) = \|df^{y,v}_x\| = \left|\frac{\partial f}{a\partial x}\right| = \frac{|x|}{a|v| \sqrt{x^2+y^2}} \]
where $f^{y,v}$ is the one-variable function defined by $f^{y,v}(x) = f(x,y,v)$, i.e. $y$ and $v$ are considered as
constants. Similarly, we compute
\[ \dersens{f^{x,v}}(y) = \|df^{x,v}_y\| = \left|\frac{\partial f}{a\partial y}\right| = \frac{|y|}{a|v| \sqrt{x^2+y^2}} \]
\[ \dersens{f^{x,y}}(v) = \|df^{x,y}_v\| = \left|\frac{\partial f}{b\partial v}\right| = \frac{\sqrt{x^2+y^2}}{bv^2} \]
where $f^{x,v}$ and $f^{x,y}$ are the one-variable functions defined by $f^{x,v}(y) = f(x,y,v)$ and $f^{x,y}(v) =
f(x,y,v)$.

Now we use Lemma~\ref{lemma:banach-combine} to combine the Banach spaces $(\RR, \|\ \|_x)$ and $(\RR, \|\ \|_y)$
into the Banach space $(\RR^2, \|\ \|_{xy})$ where
\[ \|(\Delta x,\Delta y)\|_{xy} = \|(\|\Delta x\|_x, \|\Delta y\|_y)\|_2 =
\|(a\Delta x, a\Delta y)\|_2 \]
We get
\[ \|df^v_{(x,y)}\| = \dersens{f^v}(x,y) = \|(\dersens{f^{y,v}}(x),\dersens{f^{x,v}}(y))\|_2 = \frac{\sqrt{2}}{a|v|} \]
Now we again use Lemma~\ref{lemma:banach-combine} to combine the Banach spaces $(\RR^2, \|\ \|_{xy})$ and $(\RR, \|\
\|_v)$ into the Banach space $(\RR^2, \|\ \|_{xyv})$ where
\[ \|(\Delta x, \Delta y, \Delta v)\|_{xyz} = a \sqrt{(\Delta x)^2 + (\Delta y)^2} + b|\Delta v| = \|(\|(a\Delta x, a\Delta y)\|_2, b\Delta v)\|_1 \]
We get
\begin{multline*}\|df_{(x,y,v)}\| = \dersens{f}(x,y,v) = \|(\dersens{f^v}(x,y),\dersens{f^{x,y}}(v))\|_\infty = \\
= \max\left(\frac{\sqrt{2}}{a|v|}, \frac{\sqrt{x^2+y^2}}{bv^2}\right)
\end{multline*}
Thus we have found not only the derivative sensitivity of $f$ ($\dersens{f}$), but also its derivative sensitivities w.r.t. components
(partial derivatives $\dersens{f^{y,v}},\dersens{f^{x,v}},\dersens{f^{x,y}},\dersens{f^v}$). The sensitivities w.r.t. components depend on
other components but this is not a problem because when computing $f(x,y,v)$ differentially privately, even w.r.t.
a component, we need to use all of $x,y,v$ anyway, so why not use all this information for computing the noise level too.

To achieve differential privacy, we need to find a smooth upper bound on the derivative sensitivity. When
differential privacy is required only w.r.t. a component then the sensitivity w.r.t. to the component only needs to be
smoothed for changes in that component (changes in other components are allowed to change the smooth upper bound without
smoothness restriction).

After a smooth upper bound has been found, we can use Theorem~\ref{thm:dersenscauchy-banach} to compute $f$ differentially
privately.

\subsubsection{Finding a smooth upper bound}\label{app:theoryexample2}

Consider the example of computing differentially privately the time it takes for the next ship to reach the port. This time can be
expressed as
\[ f(x_1,y_1,v_1,\ldots,x_n,y_n,v_n) = \min_{i=1}^n \frac{\sqrt{x_i^2+y_i^2}}{v_i} \]
\[ f : \RR^{3n} \rightarrow \RR \]
where $(x_i,y_i)$ are the coordinates of the $i^\mathrm{th}$ ship (with the port at $(0,0)$) and $v_i$ is its speed.
%Let $z_i = (x_i,y_i,v_i), z = (z_1,\ldots,z_n) = (x_1,y_1,v_1,\ldots,x_n,y_n,v_n)$.

Note that $v_i$ is in the power $-1$ and we do not know how to find the smooth derivative sensitivity of the function
$f_{v,i}(v_i) = v_i^{-1}$ (we only know how to do it for power functions with exponent $\ge 1$).
Let us define $w_i = \zeta \ln v_i$.
The coefficient $\zeta$ is used to control the distance by which the whole input vector changes if $\ln v_i$ is changed by 1.
Similarly, we add a coefficient to the geographical coordinates: $s_i = \alpha x_i, t_i = \alpha y_i$. Then we consider
$(s_1,t_1,w_1,\ldots,s_n,t_n,w_n)$ as an element of the Banach space $(\RR^{3n},\|\ \|)$ where
\begin{multline*}
\|(s_1,t_1,w_1,\ldots,s_n,t_n,w_n)\| = \\ = \|(\|(\|(s_1,t_1)\|_2,w_1)\|_1,\ldots,\|(\|(s_n,t_n)\|_2,w_n)\|_1)\|_p
\end{multline*}
Then $v_i = e^{w_i/\zeta}, x_i = \frac{s_i}{\alpha}, y_i = \frac{t_i}{\alpha}$ and
\[ g(s_1,t_1,w_1,\ldots,s_n,t_n,w_n) = \frac{1}{\alpha} \min_{i=1}^n \frac{\sqrt{s_i^2+t_i^2}}{e^{w_i/\zeta}} \]
Now the derivative sensitivity of $g_{w,i}(w_i) = e^{-w_i/\zeta}$ is
\[ \dersens{g_{w,i}}(w_i) = \frac{1}{\zeta} e^{-w_i/\zeta} \]
which is $\frac{1}{\zeta}$-smooth. The function $g_{w,i}$ itself is also $\frac{1}{\zeta}$-smooth.

The derivative sensitivity of $g_{st,i}(s_i,t_i) = \sqrt{s_i^2+t_i^2}$ in $(\RR^2,\ell_2)$ is
\[ \dersens{g_{st,i}}(s_i,t_i) = 1 \]
which is $\beta$-smooth for all $\beta$.
The function $g_{st,i}$ is $\frac{1}{\zeta}$-smooth if $\frac{1}{\sqrt{s_i^2+t_i^2}} \le \frac{1}{\zeta}$, i.e.
if $\sqrt{s_i^2+t_i^2} \ge \zeta$. A $\frac{1}{\zeta}$-smooth upper bound on $g_{st,i}$ is
\[ \hat{g}_{st,i}(s_i,t_i) =
\begin{cases}
\sqrt{s_i^2+t_i^2} & \mbox{if $\sqrt{s_i^2+t_i^2} \ge \zeta$} \\
\zeta e^{\frac{\sqrt{s_i^2+t_i^2}}{\zeta} - 1} & \mbox{otherwise}
\end{cases}
\]

An upper bound on
the derivative sensitivity of $g_i(s_i,t_i,w_i) = \frac{\sqrt{s_i^2+t_i^2}}{e^{w_i/\zeta}}$ is
\begin{align*}
& c_{g_i}(s_i,t_i,w_i) =
\left\|(\dersens{g_{st,i}}(s_i,t_i) \cdot g_{w,i}(w_i), \dersens{g_{w,i}}(w_i) \cdot \hat{g}_{st,i}(s_i,t_i))\right\|_\infty = \\
& = \left\|\left(1 \cdot e^{-w_i/\zeta}, \frac{1}{\zeta} e^{-w_i/\zeta} \cdot \hat{g}_{st,i}(s_i,t_i)\right)\right\|_\infty =
\frac{\max\left(1,\frac{\hat{g}_{st,i}(s_i,t_i)}{\zeta}\right)}{e^{w_i/\zeta}}
\end{align*}
and it is $\frac{1}{\zeta}$-smooth because $\dersens{g_{w,i}}(w_i)$, $\dersens{g_{st,i}}$, $g_{w,i}(w_i)$,
and $\hat{g}_{st,i}(s_i,t_i))$ are $\frac{1}{\zeta}$-smooth.

A $\frac{1}{\zeta}$-smooth upper bound on $\dersens{g}$ is
\[ c(u) = \frac{1}{\alpha} \max_i c_{g_i}(s_i,t_i,w_i) =
\frac{1}{\alpha} \max_i \frac{\max\left(1,\frac{\hat{g}_{st,i}(s_i,t_i))}{\zeta}\right)}{e^{w_i/\zeta}}
%= \max_i \frac{\max\left(\frac{1}{\alpha},\frac{\sqrt{x_i^2+y_i^2}}{\zeta}\right)}{v_i}
\]
where $u = (s_1,t_1,w_1,\ldots,s_n,t_n,w_n)$.

Now we can use Theorem~\ref{thm:dersenscauchy-banach} to compute an $\epsilon$-differentially private version of $g$:
\[ h(u) = g(u) + \frac{c(u)}{b}\cdot \eta \]
\[ \epsilon = (\gamma + 1)(b + \frac{1}{\zeta}) \]
\[ \gamma>1,b>0,\eta \sim \gencauchy{\gamma} \]
To compute an $\epsilon$-differentially private version of $f$, we first transform $(x_1,y_1,v_1,\ldots,x_n,y_n,v_n)$
into $u$ and then compute $h(u)$.

\section{Sensitivities of some particular functions}\label{app:functions}
We show how to find smooth upper bounds on derivatives of composite functions, whose non-composite versions have been considered in Sec.~\ref{sec:banach-smoothing}.

\textbf{Composite power function.} Let $f(x) = (g(x))^r, r\ge 1$ where $g : X \rightarrow \RR_+$ and $X$ is a Banach space. We have
\[ \dersens{f}(x) = |r| (g(x))^{r-1} \cdot \dersens{g}(x) \enspace .\]
If $g$ is $\beta_1$-smooth and $\dersens{g}$ is $\beta_2$-smooth then, by Lemma~\ref{lm:sumprod}, $\dersens{f}$ is $(|r-1|\beta_1 +
\beta_2)$-smooth.
The function $f$ itself is $\beta$-smooth where
\[ \beta = \sup_x \frac{|r| (g(x))^{r-1} \cdot \dersens{g}(x)}{(g(x))^r} = |r| \cdot \sup_x \frac{\dersens{g}(x)}{g(x)} = |r|\beta_1 \enspace.\]
This can be extended to $r < 1$ if $g$ itself is $\beta_1$-smooth, not only its upper bound.
For $r \ge 1$, if $\bar{g}$ is a $\beta_1$-smooth upper bound on $g$ then
\[ r\cdot (\bar{g}(x))^{r-1} \cdot \dersens{g}(x) \]
is a $(|r-1|\beta_1 + \beta_2)$-smooth upper bound on $\dersens{f}$.

\textbf{Composite exponent.} Let $f(x) = e^{r\cdot g(x)}, r\in\RR, x\in X$ where $g : X \rightarrow \RR$ and $X$ is a Banach space. We have
\[ \dersens{f}(x) = |r|e^{r\cdot g(x)} \cdot \dersens{g}(x)\enspace; \]
\[ \frac{\dersens{f}(x)}{|f(x)|} = |r|\cdot \dersens{g}(x) \enspace.\]
Suppose that $\dersens{g}(x) \le B$ for all $x$. Then $f$ is $|r|B$-smooth.
\[ \dersens{\dersens{f}}(x) = |r|\dersens{f}(x) \dersens{g}(x) + |r|e^{r\cdot g(x)} \cdot \dersens{\dersens{g}}(x) \enspace; \]
\[ \frac{\dersens{\dersens{f}}(x)}{\dersens{f}(x)} = |r|\cdot \dersens{g}(x) +
\frac{\dersens{\dersens{g}}(x)}{\dersens{g}(x)} \enspace . \]
If $\dersens{g}(x) \le B$ for all $x$ (i.e. $B$ is a $0$-smooth upper bound on $\dersens{g}(x)$) and $\dersens{g}(x)$ is
$\beta_2$-smooth then $f$ is $(|r|B + \beta_2)$-smooth by Lemma~\ref{lm:sumprod}.

\textbf{Composite sigmoid and tauoid.} The composition with a real function of Sec.~\ref{sec:banach-smoothing} can be combined with the derivative sensitivities of a sigmoid and a tauoid to get a smooth derivative sensitivity for the functions
$f_1(x) = \sigma(g(x))$ and $f_2(x) = \tau(g(x))$ where $g : X \rightarrow \RR$, $X$ is a Banach space, and $\dersens{g}(x) \le B$ for all $x$.

\textbf{Composite product.} Consider a function $f(x) = \prod_{i=1}^n g_i(x)$ where $g_i : X \rightarrow \RR$ and $X$ is a Banach
space. We have
\begin{eqnarray*}
\dersens{f}(x) & = & \sum_{i=1}^n \dersens{g_i}(x) \cdot \prod_{j\ne i} |g_j(x)| = \sum_{i=1}^n \dersens{g_i}(x) \cdot
\left|\frac{f(x)}{g_i(x)}\right|\\
& = & |f(x)| \cdot \sum_{i=1}^n \frac{\dersens{g_i}(x)}{|g_i(x)|}\enspace.
\end{eqnarray*}
By Lemma~\ref{lm:sumprod}, if $g_i$ is $\beta_i$-smooth and $\dersens{g_i}$ is $\beta_i'$-smooth then $f$ is $(\sum_i \beta_i)$-smooth, and
$\dersens{f}$ is $\max_i (\beta_i' + \sum_{j \ne i} \beta_j)$-smooth.
As we see, the smoothness guarantees of $f$ and $\dersens{f}$ are much worse than in the case where the variables are
independent.

\textbf{Composite $\ell_p$-norm.} Let $f:\prod_{i=1}^n X_i \rightarrow \RR,f(x_1,\ldots,x_n) = \|(f_1(x_1),\ldots,f_n(x_n)\|_p$ where $X_i$ are
Banach spaces. Let $X = \prod_{i=1}^n X_i$ and $x = (x_1,\ldots,x_n)$.
Let $y_i = f_i(x_i)$ and $y = (y_1,\ldots,y_n)$.
The derivative sensitivity of $f$ w.r.t. $x_i$ is
\[ c_i(x) = \frac{\partial f}{\partial y_i}(y) \cdot \dersens{f_i}(x_i) = \left(\frac{y_i}{\|y\|_p}\right)^{p-1} \cdot \dersens{f_i}(x_i) \]
The derivative sensitivity of $f$ in $(X,\ell_{p'})$ is
\[ \dersens{f}(x) = \|(c_1(x),\ldots,c_n(x))\|_{q'} \enspace.\]
Note that \[ \left(\frac{y_i}{\|y\|_p}\right)^{p-1} \le 1\enspace. \]
Thus a $\beta$-smooth upper bound on $c_i(x)$ is $\dersens{f_i}(x_i)$.
A $\beta$-smooth upper bound on $\dersens{f}(x)$ is $\|\dersens{f_1}(x_1),\ldots,\dersens{f_n}(x_n)\|_{q'}$.
Also note that
\[ \left\|\left(\left(\frac{y_i}{\|y\|_p}\right)^{p-1}\right)_{i=1}^n\right\|_q = 1 \]
and if $p' = p$ then
\[ \dersens{f}(x) \le \left\|\left(\left(\frac{y_i}{\|y\|_p}\right)^{p-1}\right)_{i=1}^n\right\|_q \cdot \max \dersens{f_i}(x_i) = \max \dersens{f_i}(x_i) \enspace.\]
Thus for $p' = p$, we have $\max \dersens{f_i}(x_i)$ as another $\beta$-smooth upper bound on $\dersens{f}(x)$.
It is at least as good as the more general bound.

\section{Postponed Proofs}\label{app:proofs}

\subsection{Proof of Lemma~\ref{lm:surj}}
First of all, a norm cannot be a constant function $0$ due to the condition $\norm{\vec{x}} = 0 \iff \vec{x} = \vec{0}$. Let $\vec{x}$ be such that $\norm{\vec{x}}_N = y'$ for some $0 \neq y' \in \RR^{+}$. For all $y \in \RR^{+}$, we have $y = \frac{y}{y'}\cdot y' = \frac{y}{y'}\cdot \norm{\vec{x}}_N = \norm{\frac{y}{y'}\cdot\vec{x}}_N$, where $\frac{y}{y'}\cdot\vec{x} \in \RR^m$.

\subsection{Proof of Lemma~\ref{lm:scalecopy}}
Since an $\ell_p$-norm is defined for $p \geq 1$, we may raise both sides of equation to the power $p$. We use the definition of $\ell_p$-norm and rewrite the term.
\begin{eqnarray*}
\norm{\alpha_1 x,\ldots,\alpha_k x,y_1,\ldots,y_m}^p_p & = & \sum_{i=1}^{k} (\alpha_i x)^p + \sum_{i=1}^{m} y^p_i \\
& = & \left(\sum_{i=1}^{k} \alpha^p_i\right) \cdot x^p + \sum_{i=1}^{m} y^p_i \\
& = & \norm{\sqrt[p]{\sum_{i=1}^{k} \alpha^p_i} \cdot x,y_1,\ldots,y_m}^p_p\enspace.
\end{eqnarray*}

\subsection{Proof of Lemma~\ref{lm:ungroup}}
Since $p,q \geq 1$, we may raise both sides of equations to the powers $p$ or $q$. The main inequalities that we use in the proof are $a^n + b^n \leq (a + b)^n$ for $n \geq 1$, and $a^n + b^n \geq (a + b)^n$ for $n \leq 1$.

\begin{eqnarray*}
\norm{\norm{\vec{x}}_q,\norm{\vec{y}}_q, z_1,\ldots,z_m}_p^p & = & \left(\sum_{i=1}^{k} x_i^q\right)^{\frac{p}{q}} + \left(\sum_{i=1}^{n} y_i^q\right)^{\frac{p}{q}} + \sum_{i=1}^{m} z_i^p \\
& \leq & \left(\sum_{i=1}^{k} x_i^q + \sum_{i=1}^{n} y_i^q\right)^{\frac{p}{q}} + \sum_{i=1}^{m} z_i^p \\
& = & \norm{\vec{x} | \vec{y}}_q^{p} + \sum_{i=1}^{m} z_i^p \\
& = & \norm{\norm{\vec{x} | \vec{y}}_q, z_1,\ldots,z_m}_p^p\enspace.
\end{eqnarray*}
\begin{eqnarray*}
\norm{\norm{\vec{x}}_p,\norm{\vec{y}}_p, z_1,\ldots,z_m}_q^q & = & \left(\sum_{i=1}^{k} x_i^p\right)^{\frac{q}{p}} + \left(\sum_{i=1}^{n} y_i^p\right)^{\frac{q}{p}} + \sum_{i=1}^{m} z_i^q \\
& \geq & \left(\sum_{i=1}^{k} x_i^p + \sum_{i=1}^{n} y_i^p\right)^{\frac{q}{p}} + \sum_{i=1}^{m} z_i^q \\
& = & \norm{\vec{x} | \vec{y}}_p^{q} + \sum_{i=1}^{m} z_i^q \\
& = & \norm{\norm{\vec{x} | \vec{y}}_p, z_1,\ldots,z_m}_q^q\enspace.
\end{eqnarray*}
If $p = q$, then the inequalities in these equation arrays are equalities.

\subsection{Proof of Lemma~\ref{lm:injmap}}
By Fact~\ref{fact:lpeqv}, if $p \geq q$, then $\norm{\vec{x}}_p \leq \norm{\vec{x}}_q$, so it suffices to prove $\norm{x_1,\ldots,x_m}_q \leq \norm{y_1,\ldots,y_n}_q$. Since $q \geq 1$, we may instead prove $\norm{x_1,\ldots,x_m}_q^q \leq \norm{y_1,\ldots,y_n}_q^q$. Assuming that there exists an injective mapping $f$ such that $|x_i| \leq |y_{f(i)}|$, we have
\begin{eqnarray*}
\norm{x_1,\ldots,x_m}_q^q & = & \sum_{i \in [m]} x_i^q \leq \sum_{i \in [m]} y_{f(i)}^q \\
 & \leq & \sum_{i \in [m]} y_{f(i)}^q + \sum_{j \in [n]\setminus Im(f)} y_j^q \\
 & = & \norm{y_1,\ldots,y_n}_q^q\enspace.
\end{eqnarray*}\qedhere

\subsection{Proof of Lemma~\ref{lemma:banach-base}}

Let $\nabla f(\vec{x}) = (a_i)_{i=1}^n$. Assuming $a_i\ne 0$ for all $i$ (otherwise remove the indices $i$ for which $a_i = 0$ from the summations containing $a_i$):
\begin{align*}
& |df_{\vec{x}}(\vec{y})| = |\nabla f(\vec{x}) \cdot \vec{y}| \le \sum_{i=1}^n |a_i| |y_i| =
\sum |a_i|^{\frac{p}{p-1}} \cdot \frac{|y_i|}{|a_i|^{\frac{1}{p-1}}} \le \\
& \le \left(\sum |a_i|^{\frac{p}{p-1}}\right) \left(\frac{\sum |y_i|^p}{\sum |a_i|^{\frac{p}{p-1}}}\right)^{\frac{1}{p}} =
\left(\sum |a_i|^{\frac{p}{p-1}}\right)^{\frac{p-1}{p}} \left(\sum |y_i|^p\right)^{\frac{1}{p}} = \\
& = \|\nabla f(\vec{x})\|_q \cdot \|\vec{y}\|_p
\end{align*}
for all $\vec{y}\in X$. The second inequality used here is the weighted power means inequality with exponents $1$ and $p$.
Equality is achievable (and not only for $\vec{y} = 0$): for example, by taking $y_i = |a_i|^{\frac{1}{p-1}}$.
Thus $\|\nabla f(\vec{x})\|_q$ is the smallest value of $c$ such that
for all $\vec{y}$, $|df_{\vec{x}}(\vec{y})| \le c \cdot \|\vec{y}\|_p$,
i.e. it is the operator norm $\|df_{\vec{x}}\|$.

The cases $p=1$ and $p=\infty$ can be achieved as limits of the general case.

\subsection{Proof of Lemma~\ref{lemma:banach-combine}}

(a)
We first prove that $(V,\norm{\cdot}_V)$ is a normed vector space.
We prove only the triangle inequality. The rest of the properties of norm are easy to check.
\begin{align*}
& \|(v_1,v_2)+(v_1',v_2')\|_V = \|(v_1+v_1',v_2+v_2')\|_V \\
& = \|(\|v_1+v_1'\|_{V_1},\|v_2+v_2'\|_{V_2})\|_p \le \\
& \le \|(\|v_1\|_{V_1}+\|v_1'\|_{V_1},\|v_2\|_{V_2}+\|v_2'\|_{V_2})\|_p \le \\
& \le \|(\|v_1\|_{V_1},\|v_2\|_{V_2})\|_p + \|(\|v_1'\|_{V_1},\|v_2'\|_{V_2})\|_p = \\
& = \|(v_1,v_2)\|_V+\|(v_1',v_2')\|_V
\end{align*}
The first inequality uses the triangle inequalities of $\norm{\cdot}_{V_1}$ and $\norm{\cdot}_{V_2}$ and the monotonicity of
$\norm{\cdot}_p$ in the absolute values of the coordinates of its argument vector.
The second inequality uses the triangle inequality of $\norm{\cdot}_p$.

Thus $(V,\norm{\cdot}_V)$ is a normed vector space. It remains to prove that it is complete. Consider a Cauchy sequence $\{x_n\}$ in $V$.
Then
\[ \forall \epsilon>0.\ \exists N\in\NN.\ \forall m,n>N.\ \|x_m - x_n\|_V < \epsilon \]
Let $x_n = (y_n,z_n)$ where
$y_n \in V_1$ and $z_n \in V_2$.
Note that
\[ \|y_m - y_n\|_{V_1} = \|(y_m - y_n, 0)\|_V \le \|(y_m - y_n, z_m - z_n)\|_V = \|x_m - x_n\|_V \]
Thus
\[ \forall \epsilon>0.\ \exists N\in\NN.\ \forall m,n>N.\ \|y_m - y_n\|_{V_1} < \epsilon \]
i.e. $\{y_n\}$ is a Cauchy sequence in $V_1$. Because $V_1$ is a Banach space, there exists $y\in V_1$ such that
\[ \lim_{n\rightarrow \infty} \|y_n - y\|_{V_1} = 0 \]
Similarly, we get that there exists $z \in V_2$ such that
\[ \lim_{n\rightarrow \infty} \|z_n - z\|_{V_2} = 0 \]
Let $x = (y,z)$. Note that
\[ \|x_n - x\|_V = \|(y_n - y, z_n - z)\|_V = \|(\|y_n - y\|_{V_1}, \|z_n - z\|_{V_2})\|_p \]
Then, because $\|\ \|_p$ is continuous,
\[ \lim_{n\rightarrow\infty} \|x_n - x\|_V =
\|(\lim_{n\rightarrow\infty} \|y_n - y\|_{V_1}, \lim_{n\rightarrow\infty} \|z_n - z\|_{V_2})\|_p = \|(0,0)\|_p = 0 \]
Thus $V$ is a Banach space.

\vspace{0.2cm}
\noindent (b) Let $c_1 = \|dg_{v_1}\|$, $c_2 = \|dh_{v_2}\|$. Note that
\begin{align*}
& \lim_{x_1\rightarrow 0_{V_1}} \frac{|g(v_1+x_1) - g(v_1) - df_v(x_1,0)|}{\|x_1\|_{V_1}} = \\
& = \lim_{x_1\rightarrow 0_{V_1}} \frac{|f(v_1+x_1,v_2) - f(v_1,v_2) - df_v(x_1,0)|}{\|(x_1,0)\|_V} = \\
& = \lim_{(x_1,0)\rightarrow 0_V} \frac{|f(v + (x_1,0)) - f(v) - df_v(x_1,0)|}{\|(x_1,0)\|_V} = \\
& = \lim_{x\rightarrow 0_V} \frac{|f(v + x) - f(v) - df_v(x)|}{\|x\|_V} = 0
\end{align*}
The last equality holds by the definition of Fr\'echet derivative. The equality before that holds because the limit on
the right-hand side exists. Then, again by the definition of Fr\'echet derivative, we get that the linear map that maps
$x_1$ to $df_v(x_1,0)$, is $dg_{v_1}$. Thus $df_v(x_1,0) = dg_{v_1}(x_1)$. Similarly, we get $df_v(0,x_2) = dh_{v_2}(x_2)$.
Now
\begin{align*}
& |df_v(x_1,x_2)| = |df_v(x_1,0) + df_v(0,x_2)| = \\
& = |dg_{v_1}(x_1) + dh_{v_2}(x_2)| \le |dg_{v_1}(x_1)| + |dh_{v_2}(x_2)| \le \\
& \le c_1 \|x_1\|_{V_1} + c_2 \|x_2\|_{V_2} \le \|(c_1,c_2)\|_q \cdot \|(\|x_1\|_{V_1},\|x_2\|_{V_2})\|_p = \\
& = \|(c_1,c_2)\|_q \cdot \|(x_1,x_2)\|_V
\end{align*}
The last inequality follows from the weighted power means inequality, similarly to the proof of
Lemma~\ref{lemma:dersens-lpnorm}.
Equality is also achievable: because $c_1 = \|dg_{v_1}\|$ and $c_2 = \|dh_{v_2}\|$, there exist $x_1$ and $x_2$ that
achieve equality in the second inequality. Then scale $x_1$ and $x_2$ by constants such that $\|x_1\|_{V_1}$ and
$\|x_2\|_{V_2}$ (which scale by the same constants) achieve equality in the third inequality.
To achieve equality in the first inequality, we may further need to multiply $x_1$ and/or $x_2$ by $-1$.
Thus $\|df_v\| = \|(c_1,c_2)\|_q$.

\subsection{Proof of Lemma~\ref{lm:scaleineq}}
$N \preceq M$ implies $\norm{x_1,\ldots,x_n}_{N} \leq \norm{x_1,\ldots,x_n}_{M}$ for any valuation $x_1,\ldots,x_n \in \RR^n$. Instead of $N \preceq M$, we may write $N'(V_1,\ldots,V_m) \preceq M'(V_1,\ldots,V_m)$. By Lemma~\ref{lm:surj}, each subnorm $V_i$ is surjective, and since all these subnorms use distinct sets of variables, we have $\norm{v_1,\ldots,v_m}_{N'} \leq \norm{v_1,\ldots,v_m}_{M'}$ for all valuations $(v_1,\ldots,v_m) \in (\RR^+)^m$.

Since $\alpha_i \in \RR^{+}, v_i \in \RR^{+}$, also $\alpha_i v_i \in \RR^{+}$. Since $\norm{v_1,\ldots,v_m}_{N'} \leq \norm{v_1,\ldots,v_m}_{M'}$ holds for \emph{any} possible valuation of $v_i \in \RR^{+}$, it also holds for $\alpha_i v_i \in \RR^{+}$. This also holds for $v_i = \norm{x_1,\ldots,x_n}_{V_i}$, so we get $N'(\alpha_1 V_1,\ldots,\alpha_m V_m) \preceq M'(\alpha_1 V_1,\ldots,\alpha_m V_m)$.

\subsection{Proof of Lemma~\ref{lm:hammer}}

We prove the first inequality, and the proof would be analogous for the second one. Consider any subnorm $\norm{M_1,\ldots,M_k}_r$ of the norm $N$. Since $p$ is the largest used subnorm of $N$, we have $\norm{M_1,\ldots,M_k}_r \succeq \norm{M_1,\ldots,M_k}_p$. Applying this inequality to every possible subnorm of $N$ and substituting $r$ with $p$, we get a {\compbanach} in which all the subnorms are $\norm{\cdot}_p$ for the same $p \geq 1$. It allows us to apply Lemma~\ref{lm:ungroup} and ungroup all the subnorms. We assume that all scalings in $N$ are applied directly to the variables. Applying ungrouping procedure recursively, we finally reach the scaled variables, getting $N = \norm{\alpha_{11} x_1,\ldots, \alpha_{mk_m} x_m}_p$, where some variables may repeat if they were repeating in different subnorms of $N$ before. We may now use Lemma~\ref{lm:scalecopy} to merge repeating variables into one, rewriting 
\begin{multline*}
\norm{\alpha_{11} x_1,\ldots,\overbrace{\alpha_{i1} x_i,\ldots,\alpha_{ik_i} x_i}^{k_i},\ldots, \alpha_{mk_m} x_m}_p = \\ = \norm{\alpha_{11} x_1,\ldots,\sqrt[p]{\sum_{j=1}^{k_i} \alpha_{ij}} x_i,\ldots, \alpha_{mk_m} x_m}_p.
\end{multline*}
 After doing it for all $i \in [m]$, we get $\norm{\alpha_1 x_1,\ldots,\alpha_m x_m}_p$.

\subsection{Proof of Lemma~\ref{lm:smoothness}}
By Def.~\ref{def:smoothness}, the mapping $f$ is \emph{$\beta$-smooth}, if $f(x)\leq e^{\beta\cdot \|x'-x\|}\cdot f(x')$ for all $x,x'\in X$. We may rewrite it as $\frac{f(x)}{f(x')} \leq e^{\beta\cdot \|x'-x\|}\cdot f(x')$. Applying $\ln$ to both sides, it suffices to prove that $\ln (f(x)) - \ln(f(x')) \leq \beta \cdot \|x'-x\|$, which is $\frac{\ln (f(x)) - \ln(f(x'))}{\|x'-x\|} \leq \beta$.

Applying mean value theorem to the function $\ln \circ f: X \to \RR$, we get $\frac{\abs{\ln (f(x)) - \ln(f(x'))}}{\|x'-x\|} = \norm{d(\ln \circ f)_v}$ for some $v \in X$. Applying derivative chain rule, since $\frac{\partial \ln}{\partial x}(x) = \frac{1}{\abs{x}}$, we get $\norm{d(\ln \circ f)_v} = \frac{\norm{df_v}}{\abs{f(v)}} = \frac{\dersens{f}(v)}{|f(v)|} \leq \beta$, where the last inequality comes from the lemma statement.

\subsection{Proof of Lemma~\ref{lm:sumprod}}
We have:
\begin{enumerate}
\item $\abs{\frac{(f(x) + g(x))'}{f(x) + g(x)}} = \frac{\abs{f'(x) + g'(x)}}{\abs{f(x)} + \abs{g(x)}} \leq \frac{\abs{f'(x)} + \abs{g'(x)}}{\abs{f(x)} + \abs{g(x)}} \leq$\\ $\max\left(\abs{\frac{f'(x)}{f(x)}},\abs{\frac{g'(x)}{g(x)}}\right) \leq \max(\beta_f,\beta_g)$.
\item $\abs{\frac{(f(x) \cdot g(x))'}{f(x) \cdot g(x)}} = \abs{\frac{f'(x) \cdot g(x) + f(x) \cdot g'(x)}{f(x) \cdot g(x)}} \leq \abs{\frac{f'(x)}{f(x)}} + \abs{\frac{g'(x)}{g(x)}} \leq \beta_f + \beta_g$.
\item $\abs{\frac{(f(x) / g(x))'}{f(x)/ g(x)}} = \abs{\frac{f'(x) \cdot g(x) - f(x) \cdot g'(x)}{g(x)^2}\cdot\frac{g(x)}{f(x)}} = \abs{\frac{f'(x)}{f(x)} - \frac{g'(x)}{g(x)}} \leq \abs{\frac{f'(x)}{f(x)}} + \abs{\frac{g'(x)}{g(x)}} \leq \beta_f + \beta_g$.
\end{enumerate}

\subsection{Proof of Lemma~\ref{lm:lpnormderiv}}
Let $y_i = f_i(x_i)$ and $y = (y_1,\ldots,y_n)$. We have
\[\frac{\partial f}{\partial x_i}(x) = \frac{\partial f}{\partial y_i}(y) \cdot \frac{\partial f_i}{\partial x_i}(x_i) = \left(\frac{y_i}{\|y\|_p}\right)^{p-1} \cdot \frac{\partial f_i}{\partial x_i}(x_i) \enspace.\]
Since $\frac{y_i}{\|y\|_p} = \frac{f_i(x_i)}{f(x)} = \frac{f_i(x_i)}{\norm{(f_i(x_i))_{i=1}^{n}}_p}$, we have $\frac{f_i(x_i)}{f(x)} \leq 1$, and hence also $\left(\frac{f_i(x_i)}{f(x)}\right)^{p-1} \leq 1$, getting $\frac{\partial f}{\partial x_i}(x) \leq \frac{\partial f_i}{\partial x_i}(x_i)$.

\subsection{Proof of Lemma~\ref{lm:lpnormderiv2}}
Let $y_j = f_j(x)$, $z = \sum_{j=1}^n y_j^p$. We have
\begin{multline*}
\frac{\partial f}{\partial x_i}(x) = \frac{\partial f}{\partial z}(z) \cdot \sum_{j=1}^{n} \left(\frac{\partial f_j(x)^p}{\partial y_j}(y_j) \cdot \frac{\partial f_j}{\partial x_i}(x)\right)\\ = \sum_{j=1}^n \left(\frac{y_j}{\norm{y}_p}\right)^{p-1} \cdot \frac{\partial f_j}{\partial x_i}(x) = \sum_{j=1}^n \left(\frac{f_j(x)}{f(x)}\right)^{p-1} \cdot \frac{\partial f_j}{\partial x_i}(x) \enspace.
\end{multline*}
As in the proof of Lemma~\ref{lm:lpnormderiv}, $\left(\frac{y_j}{\|y\|_p}\right)^{p-1} = \left(\frac{f_j(x)}{f(x)}\right)^{p-1} \leq 1$. We get $\frac{\partial f}{\partial x_i}(x)  \leq \sum_{j=1}^n \frac{\partial f_j}{\partial x_i}(x)$. We can proceed with the inequality in another way.
\begin{eqnarray*}
\sum_{j=1}^n \left(\frac{f_j(x)}{f(x)}\right)^{p-1} \cdot \frac{\partial f_j}{\partial x_i}(x)
& = & \frac{\sum_{j=1}^n f_j(x)^{p-1} \cdot \frac{\partial f_j}{\partial x_i}(x)}{f(x)^{p-1}} \\
& = & \frac{\sum_{j=1}^n f_j(x)^{p} \cdot \frac{\partial f_j}{\partial x_i}(x) \cdot \frac{1}{f_j(x)}}{f(x)^{p-1}} \\
& \leq & \max_{j=1}^n \left( \frac{1}{f_j(x)} \cdot \frac{\partial f_j}{\partial x_i}(x) \right) \cdot \frac{\sum_{j=1}^n f_j(x)^p}{f(x)^{p-1}} \\
& = & \max_{j=1}^n \left(\frac{1}{f_j(x)} \cdot \frac{\partial f_j}{\partial x_i}(x) \right)\cdot \frac{f(x)^p}{f(x)^{p-1}}\\
& = & \max_{j=1}^n \left(\frac{f(x)}{f_j(x)} \cdot \frac{\partial f_j}{\partial x_i}(x)\right)\enspace.\qedhere
\end{eqnarray*}

\subsection{Proof of Lemma~\ref{lm:lpnorm}}
Let $X = \prod_{i=1}^n X_i$ and $x = (x_1,\ldots,x_n)$. Let $\beta=\max_i\beta_i$. By Lemma~\ref{lm:lpnormderiv}, an upper bound on $\frac{\partial f}{\partial x_i}(x)$ is $c_i(x) = f'_i(x_i)$.
We have
\begin{multline*}
\abs{c_i(x)} = \abs{f'_i(x_i)} \leq \dersens{f_i}(x_i) = \abs{f_i(x_i)} \cdot \frac{\dersens{f_i}(x_i)}{\abs{f_i(x_i)}} \leq \abs{f_i(x_i)} \cdot \beta_i \enspace.
\end{multline*}

By Lemma~\ref{lemma:banach-base} and Lemma~\ref{lemma:banach-combine}, the derivative sensitivity of $f$ in $(X,\ell_{\frac{p}{p-1}})$ is
\[ \dersens{f}(x) = \|(c_1(x),\ldots,c_n(x))\|_{p} \]

%\begin{enumerate}
%\item Let $q = p$. Note that 
%\begin{multline*}
%\left\|\left(\left(\frac{y_i}{\|y\|_p}\right)^{p-1}\right)_{i=1}^n\right\|_{\frac{p}{p-1}} = \left(\sum \left(\frac{y_i^{p-1}}{\left(\sum y_i^p\right)^{\frac{p-1}{p}}}\right)^%{\frac{p}{p-1}}\right)^{\frac{p-1}{p}} = \\ = \left(\sum \frac{y_i^p}{\sum y_i^p}\right)^{\frac{p-1}{p}} = 1 \enspace.
%\end{multline*}
%We get
%\[\frac{\dersens{f}(x)}{\abs{f(x)}} \leq \frac{\abs{f(x)} \cdot \beta \cdot \left\|\left(\abs{\frac{y_i}{\|y\|_p}}^{p-1}\right)_{i=1}^n\right\|_{\frac{p}{p-1}}}{\abs{f(x)}} = %\beta\enspace. \]
%\item Let $q = \frac{p}{p-1}$.
Using inequality $|f_i(x_i)| \leq \abs{f(x)}$, we get
\[\frac{\dersens{f}(x)}{\abs{f(x)}} \leq \frac{\left\|\left(\abs{f_i(x_i)}\cdot\beta_i \right)_{i=1}^n\right\|_{p}}{\abs{f(x)}} \leq \frac{\abs{f(x)} \cdot \left\|\left(\beta_i \right)_{i=1}^n\right\|_{p}}{\abs{f(x)}} \leq \norm{(\beta_i)_{i=1}^{n}}_{p}\enspace. \]
On the other hand, using inequality $\beta_i \leq \beta$, we get
%\[\frac{\dersens{f}(x)}{\abs{f(x)}} \leq \frac{\beta \cdot \left\|\left(\abs{f_i(x_i)} \cdot \right)_{i=1}^n\right\|_{p}}{\abs{f(x)}} \leq \frac{\beta \cdot \left\|\abs{f_i(x_i)}_{i=1}^n\right\|_{p}}{\abs{f(x)}} = \beta\enspace. \]
%\end{enumerate}

\[\frac{\dersens{f}(x)}{\abs{f(x)}} \leq \frac{\left\|\left(\abs{f_i(x_i)} \cdot \beta\right)_{i=1}^n\right\|_{p}}{\abs{f(x)}} \leq \frac{\beta \cdot \left\|(\abs{f_i(x_i)})_{i=1}^n\right\|_{p}}{\abs{f(x)}} = \beta\enspace\qedhere\]

\subsection{Proof of Lemma~\ref{lm:lpnorm2}}

Let $X = \prod_{i=1}^n X_i$ and $x = (x_1,\ldots,x_n)$. By Lemma~\ref{lm:lpnormderiv2}, an upper bound on $\frac{\partial f}{\partial x_i}(x)$ is $c_i(x) = \max_{j}^n \frac{f(x)}{f_j(x)} \cdot \frac{\partial f_j}{\partial x_i}(x)$.
We have
\begin{multline*}
\abs{c_i(x)} = \abs{\max_{j} \frac{f(x)}{f_j(x)} \cdot \frac{\partial f_j}{\partial x_i}(x)} \\ \leq \abs{f(x)} \cdot \max_{j}^n  \abs{\frac{\partial f_j}{\partial x_i}(x) \cdot \frac{1}{f_j(x)}} \leq \abs{f(x)} \cdot \max_j \beta^j_i \enspace.
\end{multline*}

By Lemma~\ref{lemma:banach-base} and Lemma~\ref{lemma:banach-combine}, the derivative sensitivity of $f$ in $(X,\ell_{\frac{p}{p-1}})$ is
\[ \dersens{f}(x) = \|(c_1(x),\ldots,c_n(x))\|_{p} \]

We get
\[\frac{\dersens{f}(x)}{\abs{f(x)}} \leq \frac{\abs{f(x)} \cdot \left\|\left(\max_j \beta^j_i \right)_{i=1}^n\right\|_{p}}{\abs{f(x)}} \leq \norm{(\max_j \beta^j_i)_{i=1}^{n}}_{p}\enspace.\]

%\end{enumerate}

\subsection{Proof of Theorem~\ref{thm:diffnoisenorm}}

First of all, if $\norm{\cdot}_M \succeq \norm{\cdot}_N$, then $\norm{\cdot}_{N'} \succeq \norm{\cdot}_{M'}$ for the dual norms $\norm{\cdot}_{N'}$, $\norm{\cdot}_{M'}$. Indeed, by definition of a dual norm, $\norm{T}_{M'} = \sup\set{T(x)\ |\ \norm{x}_M \leq 1}$ for an operator $T$ from the dual space $X\to\RR$. Since $\norm{x}_N \leq \norm{x}_B$, we have $\forall x:\ \set{T(x)\ |\ \norm{x}_N \leq 1} \supseteq \set{T(x)\ |\ \norm{x}_M \leq 1}$. Hence, $\norm{T}_{N'} = \sup\set{T(x)\ |\ \norm{x}_N \leq 1} \geq$ \\ $\sup\set{T(x)\ |\ \norm{x}_M \leq 1} = \norm{T}_{N'}$.

By definition, we have $\dersens{f}(x) = \norm{df_x}_{N'}$, where $df_x$ is the Fr\'echet derivative of $f$ at $x$ and $\norm{\cdot}_{N'}$ is the dual norm of $\norm{\cdot}_{N}$. Since $\norm{\cdot}_{N'} \succeq \norm{\cdot}_{M'}$, we have $\norm{df_x}_{N'} \geq \norm{df_x}_{M'}$. Since $c(x)$ is a $\beta$-smooth upper bound on $\norm{df_x}_{N'}$, it is also a $\beta$-smooth upper bound on $\norm{df_x}_{M'}$. By Theorem~\ref{thm:dersenscauchy-banach}, $g(x) = f(x)+\frac{c(x)}{b}\cdot\eta$ is $\epsilon$-differentially private w.r.t. the norm $\norm{\cdot}_M$.

\section{Evaluation details}

\subsection{Database schema}\label{app:eval:schema}
The TPC-H testset~\cite{tpc-h} puts forth the following database schema, as given below. The tables are (randomly) filled with a number of rows, generated by a program that accompanies the schema. The number of rows depends on the \emph{scaling factor} $SF$. The tables, and the numbers of rows in them are the following:

\textbf{Part:} $SF \cdot 200,000$ rows.\nopagebreak

\begin{tabular}{| l | l |}
\hline
column & type \\
\hline
P\_PARTKEY & identifier \\
P\_NAME    & text \\ 
P\_MFGR    & text \\
P\_BRAND   & text \\
P\_TYPE    & text \\
P\_SIZE    & integer \\
P\_CONTAINER & text \\
P\_RETAILPRICE & decimal \\
P\_COMMENT & text \\
\hline
\end{tabular}
\vspace{0.5cm}

\textbf{Supplier:} $SF \cdot 10,000$ rows.\nopagebreak

\begin{tabular}{| l | l |}
\hline
column & type \\
\hline
S\_SUPPKEY & identifier \\
S\_NAME & text \\
S\_ADDRESS & text \\
S\_NATIONKEY & identifier \\
S\_PHONE & text \\
S\_ACCTBAL & decimal \\
S\_COMMENT & text\\
\hline
\end{tabular}
\vspace{0.5cm}

\textbf{Partsupp:} $SF \cdot 800,000$ rows.\nopagebreak

\begin{tabular}{| l | l |}
\hline
column & type \\
\hline
PS\_PARTKEY & Identifier \\
PS\_SUPPKEY & Identifier \\
PS\_AVAILQTY & integer \\
PS\_SUPPLYCOST & decimal\\
PS\_COMMENT  & text \\
\hline
\end{tabular}
\vspace{0.5cm}

\textbf{Customer:} $SF \cdot 150,000$ rows.\nopagebreak

\begin{tabular}{| l | l |}
\hline
column & type \\
\hline
C\_CUSTKEY & Identifier \\
C\_NAME & text \\
C\_ADDRESS &  text \\
C\_NATIONKEY & Identifier \\
C\_PHONE & text \\
C\_ACCTBAL &  decimal \\
C\_MKTSEGMENT & text \\
C\_COMMENT & text\\
\hline
\end{tabular}
\vspace{0.5cm}

\textbf{Orders:} $SF\cdot 1,500,000$ rows\nopagebreak

\begin{tabular}{| l | l |}
\hline
column & type \\
\hline
O\_ORDERKEY & Identifier \\
O\_CUSTKEY & Identifier \\
O\_ORDERSTATUS & text \\
O\_TOTALPRICE & Decimal \\
O\_ORDERDATE & Date \\
O\_ORDERPRIORITY & text \\
O\_CLERK & text \\
O\_SHIPPRIORITY & Integer \\
O\_COMMENT & text \\
\hline
\end{tabular}
\vspace{0.5cm}

\textbf{Lineitem:} $SF\cdot 6,000,000$ rows\nopagebreak

\begin{tabular}{| l | l |}
\hline
column & type \\
\hline
L\_ORDERKEY & identifier \\
L\_PARTKEY & identifier \\
L\_SUPPKEY & identifier \\
L\_LINENUMBER &  integer \\
L\_QUANTITY &  decimal \\
L\_EXTENDEDPRICE &  decimal \\
L\_DISCOUNT & decimal \\
L\_TAX &  decimal\\
L\_RETURNFLAG & text \\
L\_LINESTATUS & text \\
L\_SHIPDATE &  date \\
L\_COMMITDATE &  date \\
L\_RECEIPTDATE &  date \\
L\_SHIPINSTRUCT & text \\
L\_SHIPMODE & text \\
L\_COMMENT & text \\
\hline
\end{tabular}
\vspace{0.5cm}

\textbf{Nation:} $25$ rows\nopagebreak

\begin{tabular}{| l | l |}
\hline
column & type \\
\hline
N\_NATIONKEY & identifier \\
N\_NAME & text \\
N\_REGIONKEY & identifier\\
N\_COMMENT & text \\
\hline
\end{tabular}
\vspace{0.5cm}

\textbf{Region:} $5$ rows\nopagebreak

\begin{tabular}{| l | l |}
\hline
column & type \\
\hline
R\_REGIONKEY & identifier \\
R\_NAME & text \\
R\_COMMENT & text \\
\hline
\end{tabular}

\subsection{Sensitive components}\label{app:eval:norms}
In all tables except \textbf{Lineitem}, we consider the change that is the scaled sum of changes in all sensitive attributes. All attributes that are not a part of the norm are considered insensitive. We assumed that textual fields as well as the keys (ordinal data) are not sensitive.

A letter $G$ appended to \texttt{date} column names (e.g. $o\_shipdateG$) denotes that the initial \texttt{date} datatype has been converted to a floating-point number, which is the number of months passed from the date $1980$-$01$-$01$.

\begin{itemize}
\item \textbf{Part:} $\norm{p\_size, 0.01 \cdot p\_retailprice}_1$. The values of $p\_retailprice$ are measured in hundreds, so we consider larger changes (i.e. make such change causing a change of $1$ in the output correspond to unit sensitivity).
\item \textbf{Partsupp:} $\norm{ps\_availqty, 0.01 \cdot ps\_supplycost}_1$.
\item \textbf{Orders:} $\norm{30 \cdot o\_shipdateG, 0.01 \cdot ps\_supplycost}_1$.

\item \textbf{Customer:} $\norm{0.01 \cdot c\_acctbal}_1$.
\item \textbf{Supplier:} $\norm{0.01 \cdot s\_acctbal}_1$.
\item \textbf{Nation:} no sensitive columns.
\item \textbf{Region:} no sensitive columns.
\end{itemize}
In table \textbf{Lineitem}, several different norms would make sense and it is up to the data owner to choose the ``right'' one. We could again add up the sensitive attributes of a row, after suitably scaling them. But we could also think that the three different dates would probably move rather synchronously, and it is the maximum change among them that really matters. Hence we performed the tests with the norm
\begin{multline*}
\| l\_quantity, 0.0001 \cdot l\_extendedprice, 50 \cdot l\_discount,\\ 30 \cdot \| l\_shipdateG, l\_commitdateG, l\_receiptdateG\|_\infty \|_1\enspace.
\end{multline*}
Here the values of $l\_discount$ are very small (all around $0.1$), so we aim to protect the change in $0.02$ units. On the other hand, $l\_extendedprice$ can be tens of thousands, and we want to capture larger changes for it. The dates are measured in months, so we capture a change of one day.

Alternatively, we could take $\norm{\cdot}_1$ instead of $\norm{\cdot}_{\infty}$, giving us simply the scaled $\ell_1$ norm of all sensitive components. The benchmark results for the $\ell_1$-norm can be found in Table~\ref{tab:benchtimel1} and Table~\ref{tab:benchprecl1}. The error has not changed much compared to Table~\ref{tab:benchtime} and Table~\ref{tab:benchprec}, but for query $\mathtt{b6}$ we could decrease $\varepsilon$. The reason is that the most significant error is coming from date-related columns, and the queries mostly use only of those columns, so there is no difference whether they are related by $\ell_1$ or $\ell_{\infty}$. The queries that used several date comparisons have large errors in both cases.

\begin{table*}
\begin{footnotesize}
\begin{center}
\begin{tabular}{|l |r |r |r |r |r |r |r |r |r |}
\hline & \multicolumn{3}{c|}{SF = 0.1} & \multicolumn{3}{c|}{SF = 0.5} & \multicolumn{3}{c|}{SF = 1.0} \\ 
\hline
 & init.query & mod.query & sens. & init.query & mod.query & sens. & init.query & mod.query & sens. \\ 
\hline
\hline
$\mathtt{b1\_1}$ & $373.96$ & $5.77K$ & $15.98K$ & $866.2$ & $27.46K$ & $77.71K$ & $3.68K$ & $57.86K$ & $168.24K$\\ 
$\mathtt{b1\_2}$ & $172.31$ & $5.76K$ & $16.23K$ & $834.7$ & $28.84K$ & $83.91K$ & $1.78K$ & $55.28K$ & $159.15K$\\ 
$\mathtt{b1\_3}$ & $178.63$ & $5.9K$ & $16.21K$ & $882.59$ & $29.72K$ & $86.53K$ & $1.87K$ & $59.64K$ & $162.37K$\\ 
$\mathtt{b1\_4}$ & $186.06$ & $5.84K$ & $17.15K$ & $925.36$ & $29.47K$ & $86.29K$ & $1.87K$ & $57.28K$ & $164.86K$\\ 
$\mathtt{b1\_5}$ & $169.82$ & $5.82K$ & $5.99K$ & $837.92$ & $29.98K$ & $32.32K$ & $1.68K$ & $58.47K$ & $61.86K$\\ 
$\mathtt{b3}$ & $214.03$ & $136.24$ & $474.7$ & $782.8$ & $759.96$ & $2.62K$ & $783.24$ & $590.11$ & $1.34K$\\ 
$\mathtt{b4}$ & $199.64$ & $9.09K$ & $33.5K$ & $1.01K$ & $49.2K$ & $170.13K$ & $2.07K$ & $95.64K$ & $346.65K$\\ 
$\mathtt{b5}$ & $141.93$ & $308.07$ & $2.56K$ & $692.14$ & $1.21K$ & $5.1K$ & $1.82K$ & $2.42K$ & $14.6K$\\ 
$\mathtt{b6}$ & $146.74$ & $38.81K$ & $199.67K$ & $719.9$ & $195.99K$ & $959.3K$ & $1.53K$ & $392.06K$ & $2.02M$\\ 
$\mathtt{b7}$ & $186.22$ & $270.8$ & $1.13K$ & $941.71$ & $1.85K$ & $6.37K$ & $2.37K$ & $3.11K$ & $13.14K$\\ 
$\mathtt{b9}$ & $216.29$ & $153.68$ & $5.51K$ & $790.09$ & $830.06$ & $5.18K$ & $2.58K$ & $1.92K$ & $11.02K$\\ 
$\mathtt{b10}$ & $163.83$ & $173.5$ & $573.62$ & $829.69$ & $857.22$ & $2.79K$ & $345.26$ & $1.82K$ & $5.99K$\\ 
$\mathtt{b12\_1}$ & $249.9$ & $17.68K$ & $80.03K$ & $1.24K$ & $82.76K$ & $385.35K$ & $2.91K$ & $175.21K$ & $814.0K$\\ 
$\mathtt{b12\_2}$ & $232.16$ & $7.0K$ & $33.81K$ & $1.15K$ & $33.74K$ & $156.35K$ & $2.62K$ & $68.94K$ & $329.34K$\\ 
$\mathtt{b16}$ & $50.29$ & $216.56$ & $274.42$ & $155.31$ & $1.32K$ & $1.71K$ & $292.41$ & $2.47K$ & $3.22K$\\ 
$\mathtt{b17}$ & $117.5$ & $110.97$ & $317.11$ & $566.25$ & $595.3$ & $1.62K$ & $1.31K$ & $1.18K$ & $3.36K$\\ 
$\mathtt{b19}$ & $166.16$ & $349.24$ & $1.13K$ & $816.59$ & $1.68K$ & $5.85K$ & $1.76K$ & $3.34K$ & $12.1K$\\ 
\hline
\end{tabular}
\end{center}
\end{footnotesize}
\caption{Time benchmarks (ms)}\label{tab:benchtimel1}
\end{table*}

\begin{table*}
\begin{center}
\begin{tabular}{|l |l |r |r |r |r |r |r |r |r |r |r |r |r |}
\hline & & \multicolumn{4}{c|}{SF = 0.1} & \multicolumn{4}{c|}{SF = 0.5} & \multicolumn{4}{c|}{SF = 1.0} \\ 
\hline
 & $\varepsilon$ & init.res & mod.res & sens. & \% error & init.res & mod.res & sens. & \% error & init.res & mod.res & sens. & \% error \\ 
\hline
\hline
$\mathtt{b1\_1}$ & 1.0 & $3.79M$ & $3.55M$ & $1.0$ & $6.18$ & $18.87M$ & $17.7M$ & $1.0$ & $6.2$ & $37.72M$ & $35.38M$ & $1.0$ & $6.2$\\ 
$\mathtt{b1\_2}$ & 1.0 & $5.34G$ & $5.01G$ & $9.96K$ & $6.18$ & $27.35G$ & $25.65G$ & $9.96K$ & $6.2$ & $56.57G$ & $53.06G$ & $9.96K$ & $6.2$\\ 
$\mathtt{b1\_3}$ & 1.0 & $5.07G$ & $4.76G$ & $11.16K$ & $6.18$ & $25.98G$ & $24.37G$ & $11.16K$ & $6.2$ & $53.74G$ & $50.41G$ & $11.16K$ & $6.2$\\ 
$\mathtt{b1\_4}$ & 1.0 & $5.27G$ & $4.95G$ & $12.06K$ & $6.18$ & $27.02G$ & $25.34G$ & $12.06K$ & $6.2$ & $55.89G$ & $52.43G$ & $12.06K$ & $6.2$\\ 
$\mathtt{b1\_5}$ & 1.0 & $148.3K$ & $139.12K$ & $0.0006$ & $6.19$ & $739.56K$ & $693.7K$ & $0.0006$ & $6.2$ & $1.48M$ & $1.39M$ & $0.0006$ & $6.2$\\ 
$\mathtt{b3}$ & 1.0 & $3.62K$ & $1.09K$ & $7.98K$ & $2.13K$ & $3.21K$ & $963.25$ & $7.98K$ & $2.42K$ & $0.0$ & $0.0$ & $0.0$ & 0.0\\ 
$\mathtt{b4}$ & 1.0 & $2.92K$ & $8.58K$ & $0.0068$ & $194.14$ & $14.17K$ & $42.89K$ & $0.0069$ & $202.66$ & $28.07K$ & $85.67K$ & $0.0069$ & $205.18$\\ 
$\mathtt{b5}$ & 1.0 & $5.43M$ & $5.05M$ & $4.96K$ & $5.98$ & $25.28M$ & $23.82M$ & $4.96K$ & $5.56$ & $47.56M$ & $45.33M$ & $4.96K$ & $4.6$\\ 
$\mathtt{b6}$ & 2.5 & $17.45M$ & $17.09M$ & $48.12K$ & $0.75$ & $88.13M$ & $87.69M$ & $49.22K$ & $0.06$ & $181.93M$ & $181.41M$ & $50.45K$ & $0.0094$\\ 
$\mathtt{b7}$ & 1.0 & $22.07M$ & $22.12M$ & $22.25K$ & $1.24$ & $95.63M$ & $100.99M$ & $22.51K$ & $5.85$ & $212.11M$ & $219.31M$ & $22.58K$ & $3.5$\\ 
$\mathtt{b9}$ & 1.0 & $30.32M$ & $30.32M$ & $40.0K$ & $1.32$ & $137.73M$ & $137.73M$ & $49.2K$ & $0.36$ & $283.82M$ & $283.82M$ & $49.2K$ & $0.17$\\ 
$\mathtt{b10}$ & 1.0 & $100.31K$ & $111.39K$ & $3.42K$ & $45.15$ & $149.6K$ & $175.47K$ & $3.42K$ & $40.18$ & $0.0$ & $93.54K$ & $3.31K$ & $\infty$\\ 
$\mathtt{b12\_1}$ & 1.0 & $3.12K$ & $9.05K$ & $0.0014$ & $190.47$ & $15.41K$ & $45.15K$ & $0.0014$ & $193.04$ & $30.84K$ & $90.3K$ & $0.0014$ & $192.83$\\ 
$\mathtt{b12\_2}$ & 1.0 & $1.29K$ & $3.65K$ & $0.0014$ & $183.07$ & $6.2K$ & $18.11K$ & $0.0014$ & $191.85$ & $12.37K$ & $36.24K$ & $0.0014$ & $193.01$\\ 
$\mathtt{b16}$ & 4.5 & $9.95K$ & $249.44K$ & $16.69$ & $2.41K$ & $49.35K$ & $1.25M$ & $16.69$ & $2.44K$ & $98.97K$ & $2.51M$ & $16.69$ & $2.44K$\\ 
$\mathtt{b17}$ & 1.0 & $31.54K$ & $265.49K$ & $902.59$ & $770.26$ & $256.24K$ & $1.61M$ & $902.59$ & $531.72$ & $531.93K$ & $3.53M$ & $902.59$ & $565.73$\\ 
$\mathtt{b19}$ & 7.0 & $155.25K$ & $258.9K$ & $28.48K$ & $250.22$ & $1.1M$ & $1.52M$ & $29.79K$ & $65.16$ & $1.73M$ & $2.93M$ & $30.59K$ & $87.66$\\ 
\hline
\end{tabular}
\end{center}
\caption{Precision benchmarks}\label{tab:benchprecl1}
\end{table*}

\subsection{Queries}\label{app:eval:queries}
\begin{verbatim}

--b1_1
select
    sum(lineitem.l_quantity)
from
    lineitem
where
    lineitem.l_shipdateG <= 230.3 - 30
and
    lineitem.l_returnflag = 'R'
and
    lineitem.l_linestatus = 'F'
;

--b1_2
select
    sum(lineitem.l_extendedprice)
from
    lineitem
where
    lineitem.l_shipdateG <= 230.3 - 30
and
    lineitem.l_returnflag = 'R'
and
    lineitem.l_linestatus = 'F'
;

--b1_3
select
    sum(lineitem.l_extendedprice*(1-lineitem.l_discount))
from
    lineitem
where
    lineitem.l_shipdateG <= 230.3 - 30
and
    lineitem.l_returnflag = 'R'
and
    lineitem.l_linestatus = 'F'
;

--b1_4
select
    sum(lineitem.l_extendedprice*(1-lineitem.l_discount)
                                *(1+lineitem.l_tax))
from
    lineitem
where
    lineitem.l_shipdateG <= 230.3 - 30
and
    lineitem.l_returnflag = 'R'
and
    lineitem.l_linestatus = 'F'
;

--b1_5
select
    count(*)
from
    lineitem
where
    lineitem.l_shipdateG <= 230.3 - 30
and
    lineitem.l_returnflag = 'R'
and
    lineitem.l_linestatus = 'F'
;

--b3
select
    sum(lineitem.l_extendedprice*(1-lineitem.l_discount))
from
    customer,
    orders,
    lineitem
where
    customer.c_mktsegment = 'BUILDING'
and customer.c_custkey = orders.o_custkey
and lineitem.l_orderkey = orders.o_orderkey
and orders.o_orderdateG < 190
and lineitem.l_shipdateG > 190
and lineitem.l_orderkey = '162'
and orders.o_shippriority = '0'
;

--b4
select count(*)
from
    orders,
    lineitem
where
    orders.o_orderdateG >= 180
    and orders.o_orderdateG < 180 + 3
    and lineitem.l_orderkey = orders.o_orderkey
    and lineitem.l_commitdateG < lineitem.l_receiptdateG
    and orders.o_orderpriority = '1-URGENT'
;

--b5
select sum(lineitem.l_extendedprice*(1-lineitem.l_discount))
from
    customer,
    orders,
    lineitem,
    supplier,
    nation,
    region
where
    customer.c_custkey = orders.o_custkey
    and lineitem.l_orderkey = orders.o_orderkey
    and lineitem.l_suppkey = supplier.s_suppkey
    and customer.c_nationkey = supplier.s_nationkey
    and supplier.s_nationkey = nation.n_nationkey
    and nation.n_regionkey = region.r_regionkey
    and region.r_name = 'ASIA'
    and orders.o_orderdateG >= 213.3
    and orders.o_orderdateG < 213.3 + 12
    and nation.n_name = 'JAPAN'
;

--b6
select
        sum(lineitem.l_extendedprice * lineitem.l_discount)
from
        lineitem
where
        lineitem.l_shipdateG >= 170.5
        and lineitem.l_shipdateG < 170.5 + 12
        and lineitem.l_discount between 0.09 - 0.01
                                    and 0.09 + 0.01
        and lineitem.l_quantity < 24
;

--b7
select
    sum(lineitem.l_extendedprice * (1 - lineitem.l_discount))
from
    supplier,
    lineitem,
    orders,
    customer,
    nation as n1,
    nation as n2
where
    supplier.s_suppkey = lineitem.l_suppkey
    and orders.o_orderkey = lineitem.l_orderkey
    and customer.c_custkey = orders.o_custkey
    and supplier.s_nationkey = n1.n_nationkey
    and customer.c_nationkey = n2.n_nationkey
    and (
        (n1.n_name = 'JAPAN' and n2.n_name = 'INDONESIA')
        or (n1.n_name = 'INDONESIA' and n2.n_name = 'JAPAN')
    )
    and lineitem.l_shipdateG between 182.6 and 207
;

--b9
select
    sum(lineitem.l_extendedprice*(1-lineitem.l_discount)
           - partsupp.ps_supplycost*lineitem.l_quantity)
from
    part,
    supplier,
    lineitem,
    partsupp,
    orders,
    nation
where
    supplier.s_suppkey = lineitem.l_suppkey
    and partsupp.ps_suppkey = lineitem.l_suppkey
    and partsupp.ps_partkey = lineitem.l_partkey
    and part.p_partkey = lineitem.l_partkey
    and orders.o_orderkey = lineitem.l_orderkey
    and supplier.s_nationkey = nation.n_nationkey
    and part.p_name like '%violet%'
    and nation.n_name = 'UNITED KINGDOM'
;

--b10
select
    sum(lineitem.l_extendedprice * (1 - lineitem.l_discount))
from
    customer,
    orders,
    lineitem,
    nation
where
    customer.c_custkey = orders.o_custkey
    and lineitem.l_orderkey = orders.o_orderkey
    and orders.o_orderdateG >= 183.3
    and orders.o_orderdateG <  183.3 + 3
    and lineitem.l_returnflag = 'R'
    and customer.c_nationkey = nation.n_nationkey
    and customer.c_custkey = '64'
    and nation.n_name = 'CANADA'
;

--b12_1
select
    count(*)
from
    orders,
    lineitem
where
    orders.o_orderkey = lineitem.l_orderkey
    and (orders.o_orderpriority <> '1-URGENT'
        or orders.o_orderpriority <> '2-HIGH')
    and lineitem.l_shipmode in ('TRUCK', 'SHIP')
    and lineitem.l_commitdateG < lineitem.l_receiptdateG
    and lineitem.l_shipdateG < lineitem.l_commitdateG
    and lineitem.l_receiptdateG >= 183.3
    and lineitem.l_receiptdateG < 183.3 + 12
;

--b12_2
select
    count(*)
from
    orders,
    lineitem
where
    orders.o_orderkey = lineitem.l_orderkey
    and (orders.o_orderpriority = '1-URGENT'
        or orders.o_orderpriority = '2-HIGH')
    and lineitem.l_shipmode in ('TRUCK', 'SHIP')
    and lineitem.l_commitdateG < lineitem.l_receiptdateG
    and lineitem.l_shipdateG < lineitem.l_commitdateG
    and lineitem.l_receiptdateG >= 183.3
    and lineitem.l_receiptdateG < 183.3 + 12
;

--b16
select count(partsupp.ps_suppkey)
from
    partsupp,
    part,
    supplier
where
part.p_partkey = partsupp.ps_partkey
and partsupp.ps_suppkey = supplier.s_suppkey
and part.p_brand <> 'Brand#34'
and not (part.p_type like '%COPPER%')
and part.p_size in (5, 10, 15, 20, 25, 30, 35, 40)
and not (supplier.s_comment like '%Customer%Complaints%')
and part.p_brand = 'Brand#14'
and part.p_type  = 'LARGE ANODIZED TIN'
;

--b17
select
    sum(lineitem.l_extendedprice * 0.142857)
from
    lineitem,
    part
where
    part.p_partkey = lineitem.l_partkey
    and part.p_brand = 'Brand#34'
    and part.p_container = 'JUMBO PKG'
    and lineitem.l_quantity < 0.2 * 32
;

--b19
select
    sum(lineitem.l_extendedprice*(1-lineitem.l_discount))
from
    lineitem,
    part
where
    part.p_partkey = lineitem.l_partkey
    and lineitem.l_shipmode in ('AIR', 'AIR REG')
    and lineitem.l_shipinstruct = 'DELIVER IN PERSON'
    and part.p_size >= 1
    and
    ((
    part.p_brand = 'Brand#34'
    and part.p_container in ('SM CASE', 'SM BOX',
                             'SM PACK', 'SM PKG')
    and lineitem.l_quantity >= 35
    and lineitem.l_quantity <= 35 + 10
    and part.p_size <= 5
    )
    or
    (
    part.p_brand = 'Brand#22'
    and part.p_container in ('MED BAG', 'MED BOX',
                             'MED PKG', 'MED PACK')
    and lineitem.l_quantity >= 12
    and lineitem.l_quantity <= 12 + 10
    and part.p_size  <= 10
    )
    or
    (
    part.p_brand = 'Brand#14'
    and part.p_container in ('LG CASE', 'LG BOX',
                             'LG PACK', 'LG PKG')
    and lineitem.l_quantity >= 90
    and lineitem.l_quantity <= 90 + 10
    and part.p_size  <= 15
));
\end{verbatim}

\subsection{Examples of analyser output}
We give some examples of shorter queries that have been output by the analyser.

\subsubsection{Query $\mathtt{b1\_1}$}

\paragraph{Modified query.}
\begin{spverbatim}
SELECT sum((lineitem.l_quantity * (exp((0.1 * (200.3 + ((-1.0) * lineitem.l_shipdateG)))) / (exp((0.1 * (200.3 + ((-1.0) * lineitem.l_shipdateG)))) + 1.0)))) FROM lineitem WHERE (lineitem.l_linestatus = 'F') AND (lineitem.l_returnflag = 'R');
\end{spverbatim}

\paragraph{Sensitivity query.}
\begin{spverbatim}
SELECT max(abs(sdsg)) FROM
(SELECT sum(abs(greatest(abs((exp((0.1 * (200.3 + ((-1.0) * lineitem.l_shipdateG)))) / (exp((0.1 * (200.3 + ((-1.0) * lineitem.l_shipdateG)))) + 1.0))), abs(case when ((((0.1 * exp((0.1 * (200.3 + ((-1.0) * lineitem.l_shipdateG))))) / ((exp((0.1 * (200.3 + ((-1.0) * lineitem.l_shipdateG)))) + 1.0) ^ 2.0)) * 0.03) = 0.0) then 0.0 else ((((0.1 * exp((0.1 * (200.3 + ((-1.0) * lineitem.l_shipdateG))))) / ((exp((0.1 * (200.3 + ((-1.0) * lineitem.l_shipdateG)))) + 1.0) ^ 2.0)) * 0.03) * case when (abs(lineitem.l_quantity) >= 10.0) then abs(lineitem.l_quantity) else (exp(((0.1 * abs(lineitem.l_quantity)) - 1.0)) / 0.1) end) end)))) AS sdsg FROM lineitem, lineitem_sensRows WHERE
((lineitem.l_linestatus = 'F') AND(lineitem.l_returnflag = 'R') AND lineitem_sensRows.ID = lineitem.ID) AND lineitem_sensRows.sensitive GROUP BY lineitem_sensRows.ID) AS sub;
\end{spverbatim}

\subsubsection{Query $\mathtt{b1\_5}$}

\paragraph{Modified query.}
\begin{spverbatim}
SELECT sum(abs((exp((0.1 * (200.3 + ((-1.0) * 
lineitem.l_shipdateG)))) / (exp((0.1 * (200.3 + ((-1.0) * lineitem.l_shipdateG)))) + 1.0)))) FROM lineitem WHERE (lineitem.l_linestatus = 'F') AND (lineitem.l_returnflag = 'R');
\end{spverbatim}

\paragraph{Sensitivity query.}
\begin{spverbatim} SELECT max(sdsg) FROM (SELECT sum(abs((((0.1 * exp((0.1 * (200.3 + ((-1.0) * lineitem.l_shipdateG))))) / ((exp((0.1 * (200.3 + ((-1.0) * lineitem.l_shipdateG)))) + 1.0) ^ 2.0)) * 0.03))) AS sdsg FROM lineitem, lineitem_sensRows WHERE ((lineitem.l_linestatus = 'F')
AND (lineitem.l_returnflag = 'R') AND lineitem_sensRows.ID = lineitem.ID) AND lineitem_sensRows.sensitive GROUP BY lineitem_sensRows.ID) AS sub;
\end{spverbatim}

\subsubsection{Query $\mathtt{b16}$}

\paragraph{Modified query.}
\begin{spverbatim} SELECT sum(abs(((((((((2.0 / (exp(((-0.1) * (part.p_size - 10.0))) + exp((0.1 * (part.p_size - 10.0))))) + (2.0 / (exp(((-0.1) * (part.p_size - 15.0))) + exp((0.1 * (part.p_size - 15.0)))))) + (2.0 / (exp(((-0.1) * (part.p_size - 20.0))) + exp((0.1 * (part.p_size - 20.0)))))) + (2.0 / (exp(((-0.1) * (part.p_size - 25.0))) + exp((0.1 * (part.p_size - 25.0)))))) + (2.0 / (exp(((-0.1) * (part.p_size - 30.0))) + exp((0.1 * (part.p_size - 30.0)))))) + (2.0 / (exp(((-0.1) * (part.p_size - 35.0))) + exp((0.1 * (part.p_size - 35.0)))))) + (2.0 / (exp(((-0.1) * (part.p_size - 40.0))) + exp((0.1 * (part.p_size - 40.0)))))) + (2.0 / (exp(((-0.1) * (part.p_size - 5.0))) + exp((0.1 * (part.p_size - 5.0)))))))) FROM part, partsupp, supplier WHERE not((supplier.s_comment LIKE '%Customer%Complaints%')) AND not((part.p_type LIKE '%COPPER%')) AND
not((part.p_brand = 'Brand#34')) AND (part.p_partkey = partsupp.ps_partkey) AND
(partsupp.ps_suppkey = supplier.s_suppkey);
\end{spverbatim}

\paragraph{Sensitivity query.}
\begin{spverbatim} SELECT max(sdsg) FROM (SELECT sum(abs((((((((((0.1 * (2.0 / (exp(((-0.1) * (part.p_size - 10.0))) + exp((0.1 * (part.p_size - 10.0)))))) * 8.0) + ((0.1 * (2.0 / (exp(((-0.1) * (part.p_size - 15.0))) + exp((0.1 * (part.p_size - 15.0)))))) * 8.0)) + ((0.1 * (2.0 / (exp(((-0.1) * (part.p_size - 20.0))) + exp((0.1 * (part.p_size - 20.0)))))) * 8.0)) + ((0.1 * (2.0 / (exp(((-0.1) * (part.p_size - 25.0))) + exp((0.1 * (part.p_size - 25.0)))))) * 8.0)) + ((0.1 * (2.0 / (exp(((-0.1) * (part.p_size - 30.0))) + exp((0.1 * (part.p_size - 30.0)))))) * 8.0)) + ((0.1 * (2.0 / (exp(((-0.1) * (part.p_size - 35.0))) + exp((0.1 * (part.p_size - 35.0)))))) * 8.0)) + ((0.1 * (2.0 / (exp(((-0.1) * (part.p_size - 40.0))) + exp((0.1 * (part.p_size - 40.0)))))) * 8.0)) + ((0.1 * (2.0 / (exp(((-0.1) * (part.p_size - 5.0))) + exp((0.1 * (part.p_size - 5.0)))))) * 8.0)))) AS sdsg FROM part, partsupp, supplier, part_sensRows WHERE (not((supplier.s_comment LIKE '%Customer%Complaints%')) AND not((part.p_type LIKE '%COPPER%')) AND not((part.p_brand = 'Brand#34')) AND (part.p_partkey = partsupp.ps_partkey) AND (partsupp.ps_suppkey = supplier.s_suppkey) AND part_sensRows.ID = part.ID) AND part_sensRows.sensitive GROUP BY part_sensRows.ID) AS sub;
\end{spverbatim}

\end{document}